\documentclass[11pt]{article}

\usepackage{amsmath}
\usepackage{epsfig, amssymb, amsfonts}
\usepackage{amsthm}
\usepackage{amstext}
\usepackage{amsopn}
\usepackage{enumerate}
\usepackage{mathrsfs}
\usepackage{subfigure}
\usepackage{graphicx}
\usepackage[left=2.3cm,top=2cm,right=2.3cm,bottom=2cm, nohead,foot=1cm]{geometry}
\numberwithin{equation}{section}

\newtheorem{theorem}{Theorem}[section]
\newtheorem{lemma}[theorem]{Lemma}

\newtheorem{conjecture}[theorem]{Conjecture}

\newtheorem{proposition}[theorem]{Proposition}

\newcommand{\R}{{\mathbb R}}

\newcommand{\mfp}{{\mathfrak{p}}}

\newcommand{\mcJ}{{\mathcal{J}}}
\newcommand{\tuP}{{\Psi}}

\newcommand{\C}{{\mathbb C}}

\title{A limit theorem to a time-fractional diffusion}

\author{\textbf{Jeremy Thane Clark}\footnote{jtclark@math.msu.edu} \\ Department of Mathematics, Michigan State University }

\begin{document}
\maketitle

\begin{abstract}

We prove a limit theorem for an integral functional of a Markov process.  The Markovian dynamics is characterized by a linear Boltzmann equation modeling a one-dimensional test particle of mass $\lambda^{-1}\gg 1$ in an external periodic potential and undergoing collisions with a background gas of particles with mass one.  The object of our limit theorem is the time integral of the force exerted on the test particle by the potential, and we consider this quantity in the limit that  $\lambda$ tends to zero for time intervals on the scale $\lambda^{-1}$.  Under appropriate rescaling, the total drift in momentum generated by the  potential converges to a Brownian motion time-changed by the local time at zero of an Ornstein-Uhlenbeck process.

\end{abstract}

\section{Introduction}

\subsection{Model and results}

Consider the family $\lambda\in \R^{+}$ of Markov processes $(X_{t}^{(\lambda)},P_{t}^{(\lambda)} )\in \R^{2}$ whose densities  $\tuP_{t,\lambda}(x,p)$  obey the forward Kolmogorov equation
\begin{align}\label{TheModel}
\frac{d}{dt}\tuP_{t,\lambda}(x,\,p)= &-p\frac{\partial}{\partial x}\tuP_{t,\lambda}(x,p)+\frac{dV}{dx}\big(x\big)\frac{\partial}{\partial p}\tuP_{t,\lambda}(x,p) \nonumber     \\ &+\int_{\R}dp^{\prime}\big(\mcJ_{\lambda}(p^{\prime},p)\tuP_{t,\lambda}(x,p^{\prime})-\mcJ_{\lambda}(p,p^{\prime})\tuP_{t,\lambda}(x,p)     \big),   
\end{align}
where $V(x)=V(x+1)\geq 0$ is continuously differentiable, and the jump kernel $\mcJ_{\lambda}(p,p^{\prime})$ has the form 
\begin{align}\label{JumpRates}
  \mcJ_{\lambda}(p^{\prime},p):=  \frac{(1+\lambda)}{64}\big|p^{\prime}-p\big|e^{-\frac{1}{2}\big(\frac{1-\lambda}{2}p^{\prime} -\frac{1+\lambda}{2}p    \big)^{2}      }.   
  \end{align}
The values $\mcJ_{\lambda}(p^{\prime},p)$ correspond to  the rate of jumps from $(x,p')$ to $(x,p)$.  The Kolmogorov equation above is an idealized description of the phase space density for a test particle in dimension one that feels a spatially periodic force $-\frac{dV}{dx}(x)$ and receives elastic collisions with particles from a gas.  The jump rates $\mcJ_{\lambda}$ correspond to the one-dimensional case of equation (8.118) from~\cite{Spohn} in which the mass of a single reservoir particle is set to one, the temperature of the gas is set to one,  the  spatial density of the gas is set to $\frac{1}{32}(2\pi)^{\frac{1}{2}}$, and the mass of the test particle is $\lambda^{-1}$.   

We will subsequently suppress the $\lambda$-dependence of the dynamics by removing the superscript for the process: $(X_{t},P_{t})$.  The cumulative drift  in the particle's momentum up to time $t\in \R^{+}$ due to the periodic force field has the form $-D_{t}$ for
$$D_{t}= \int_{0}^{t}dr\frac{dV}{dx}(X_{r}).     $$ 
The momentum at time $t$ can be written in the form $P_{t}=P_{0}-D_{t}+J_{t}$,
 where $J_{t}$ is the sum of all the momentum jumps resulting from collisions with the gas.  To state our main result contained in Thm.~\ref{ThmMain} below, let us define the limiting processes.  Define $\mfp\in \R$ to be the process with  $\frak{p}_{0}=0$ and satisfying the Langevin equation
\begin{align}\label{TheLimit}
d\frak{p}_{t}=-\frac{1}{2}\frak{p}_{t}dt+d\mathbf{B}_{t}',
\end{align}
where  $\mathbf{B}'$ is a standard Brownian motion.  The solution $\frak{p}$ is referred to as the Ornstein-Uhlenbeck process~\cite{Uhlen}.   Moreover, let the process $\frak{l}$ denote the local time at zero for the process $\frak{p}$.  Recall that the local time at a point $a\in \R$ over the interval $[0,t]$ is formally given by the expression: $\int_{0}^{t}dr\delta_{a}(\frak{p}_{r})$.

In~\cite{Previous} it was shown that  $\lambda^{\frac{1}{2}}P_{\frac{\cdot}{\lambda} }$ converges in law as $\lambda\rightarrow 0$ to $\frak{p}$ over any finite time interval $[0,T]$ and that the expectation of $\sup_{0\leq t\leq T}|\lambda^{\frac{1}{4}}D_{\frac{t}{\lambda}}|$  is uniformly bounded for all $\lambda<1$.  Theorem~\ref{ThmMain} extends this result to a limit law for     $\lambda^{\frac{1}{4}}D_{\frac{\cdot}{\lambda}}$ which is joint with that of  $\lambda^{\frac{1}{2}}P_{\frac{\cdot}{\lambda} }$.  The rescaled momentum drift  $\lambda^{\frac{1}{4}}D_{\frac{\cdot}{\lambda}}$  converges to a diffusion process time-changed by the local time of the Ornstein-Uhlenbeck process $\frak{p}$ that $\lambda^{\frac{1}{2}}P_{\frac{\cdot}{\lambda} }$ limits to.  The diffusion constant $\kappa\in \R^{+}$ in the statement of Thm.~\ref{ThmMain} is formally given by a Green-Kubo form that is remarked upon in Sect.~\ref{SubSecDis}.

 \begin{theorem}\label{ThmMain}
Assume that $V(x)$ is continuously differentiable and that the initial distribution $\mu$ has finite moments in momentum: $\int_{\R^{2}}d\mu(x,p) |p|^{m}<\infty $ for $m\geq 1 $.  In the limit $\lambda\rightarrow 0$, there is convergence in law of the process pair
$$\hspace{3cm} \Big(\lambda^{\frac{1}{2}}P_{\frac{t}{\lambda} },\, \lambda^{\frac{1}{4}}D_{\frac{t}{\lambda}} \Big)\stackrel{\frak{L}}{\Longrightarrow} \big(\frak{p}_{t},\sqrt{\kappa} \mathbf{B}_{\frak{l}_{t}}   \big), \quad \quad \quad t\in[0,\,T],  $$
for a constant $\kappa>0$, and where $\frak{l}$ is the local time at zero of $\frak{p}$, and $\mathbf{B}$ is a copy of Brownian motion  independent of $\frak{p}$.  The convergence is with respect to the Skorokhod metric.

\end{theorem}

Theorem~\ref{ThmMain} implies that the contribution $J_{t}$ to the momentum generated by collisions has higher order than the forcing part $D_{t}$.  In particular, $\lambda^{\frac{1}{2}}J_{\frac{\cdot}{\lambda}}$ converges to the Ornstein-Uhlenbeck process as $\lambda\rightarrow 0$.  In the conjecture below, we give a more refined statement for the limiting law of the full momentum $\lambda^{\frac{1}{2}}P_{\frac{\cdot}{\lambda}}$ for small $\lambda$  that takes into account the perturbative contribution of the forcing term $\lambda^{\frac{1}{2}}D_{\frac{t}{\lambda}}$.  In this approximation, the contribution of the periodic force is given by a diffusive pulse  that the momentum feels when it returns to the region around the value zero. The $\frak{p}$ in the statement of the conjecture should be thought of as the limit in law  of the collision contribution $\lambda^{\frac{1}{2}}J_{\frac{\cdot}{\lambda}}$.

\begin{conjecture}\label{ThmMainII}
Take the assumptions of Thm.~\ref{ThmMain}, and let $F:C([0,T])\rightarrow \C$ be bounded and smooth with respect to the supremum norm.  Define the process $\frak{p}_{t,\lambda}$ as
\begin{align}\label{PertTheLimit}
\hspace{2cm}\frak{p}_{t,\lambda}:=\frak{p}_{t}+ \sqrt{\kappa}\lambda^{\frac{1}{4}}\Big(\mathbf{B}_{\frak{l}_{r}}-\frac{1}{2}   \int_{0}^{t}dre^{-\frac{1}{2}(t-r) }\mathbf{B}_{\frak{l}_{r}}\Big)  , 
\end{align}
where  $\frak{p}$, $\mathbf{B}$, $\frak{l}$, and $\kappa>0$ are defined as in Thm.~\ref{ThmMain}.  Then the law of the process $\lambda^{\frac{1}{2}}P_{\frac{\cdot}{\lambda}}$ is close to the law of $\frak{p}_{\cdot,\lambda}$ for $\lambda\ll 1$ in the sense that 
$$\mathbb{E}\big[ F\big(\lambda^{\frac{1}{2}}P_{\frac{\cdot}{\lambda}}   \big)\big]=  \mathbb{E}\big[ F\big(\frak{p}_{\cdot,\lambda}   \big)\big]+\mathit{O}(\lambda^{\frac{1}{2}}). $$
Note that if $\frak{p}_{t,\lambda} $ is replaced by  $\frak{p}_{t,0}=\frak{p}_{t} $ in the expectation above, then the error can at best be $\mathit{O}(\lambda^{\frac{1}{4}})$.   
\end{conjecture}

\subsection{Discussion}\label{SubSecDis}

Theorem~\ref{ThmMain} characterizes the limiting law for the integral functional of the Markov process $S_{t}=(X_{t},P_{t})$ given by
\begin{align}\label{NormDrift}
D_{t}= \int_{0}^{t}drg(S_{r}),   \hspace{1.5cm}  g(x,p)=\frac{dV}{dx}(x),
 \end{align}
for time scales $t\propto \lambda^{-1}$ and normalization factor $\lambda^{\frac{1}{4}}$.   The underlying law of the process $S_{t}$ depends on the parameter $\lambda$ through the jump rate kernel $\mathcal{J}_{\lambda}$.  Since the potential $V(x)$ has period one, it is convenient to view $S_{t}$ as having state space $\Sigma :=\mathbb{T}\times \R$, where $\mathbb{T}:=[0,1)$ is identified with the unit torus,  rather than $\R^{2}$.  The process $S_{t}\in \Sigma $ is ergodic to an equilibrium state given  by the Maxwell-Boltzmann distribution
\begin{align}\label{MaxBolt}
\Psi_{\infty,\lambda}(x,p):= \frac{ e^{-\lambda H(x,p) } }{N(\lambda)   },
\end{align} 
where $H(x,p):=\frac{1}{2}p^{2}+V(x)$ and for a normalization constant $N(\lambda)\in \R^{+} $.  Although the ergodicity is exponential in nature,  the rate of ergodicity decays as $\lambda$ goes to zero, and thus a limit theorem for a normalized version of $D_{\frac{t}{\lambda}}$  does not fall under the limit theory for integral functionals of an ergodic Markov process~\cite{Landim}.  This is also clear from the appropriate scaling factor for $D_{\frac{t}{\lambda}}$ being $\lambda^{\frac{1}{4}}$ rather the $\lambda^{\frac{1}{2}}$.  Heuristics for this scaling were given in~\cite[Sect. 1.2.2]{Previous}, and the smaller exponent for the scaling is driven by the fact that  $\frac{dV}{dx}(X_{r})$ is typically oscillating with high frequency ($\propto \lambda^{-\frac{1}{2}}$) around zero for most of the time interval $[0,\frac{T}{\lambda}]$.  These oscillations in $\frac{dV}{dx}(X_{r})$ occur as the particle revolves around the torus with speed $|P_{r}|$, which typically is found on the order $\lambda^{-\frac{1}{2}}$.  The fluctuations in $D_{t}$ have a chance to accumulate primarily when $|\lambda^{\frac{1}{2}}P_{\frac{t}{\lambda}} |$ dips down to ``small" values, and this suggests that a non-trivial limit law arising from a rescaled version of $D_{\frac{t}{\lambda}}$ should be related to the local time at zero for the limiting law of $\lambda^{\frac{1}{2}}P_{\frac{t}{\lambda}}$.          

As $\lambda\rightarrow 0$ the jump rates approach the form $\mcJ_{0}(p,p^{\prime})=j(p-p^{\prime})$ for 
\begin{align}\label{LambdaZero}
 j(p) :=   \frac{1}{64}|p|e^{-\frac{1}{8}p^{2}}   ,    
 \end{align}
which describe an unbiased random walk in momentum. Thus the process $S_{t}$ behaves more like a null-recurrent Markov process for small $\lambda$.  This idea breaks down at time scales $\propto \lambda^{-1}$ where a first-order contribution to $\mcJ_{\lambda}(p,p')$ around $\lambda=0$ generates the frictional drag to smaller momenta seen in the linear drift term of the Langevin equation~(\ref{TheLimit})  defining $\frak{p}_{t}$.  The diffusion constant $\kappa\in \R^{+}$ in Thm.~\ref{ThmMain} is formally given by the Green-Kubo expression 
\begin{align}\label{DiffConst} \kappa = 2 \int_{[0,\,1]\times \R }dx dp\frac{dV}{dx}(x)\frak{R}^{(0)}(\frac{dV}{dx})(x,p),   
\end{align}
where $\frak{R}^{(0)}=\int_{0}^{\infty}dr e^{r\mathcal{L}_{0}}  $ is the reduced resolvent of the backwards generator $\mathcal{L}_{0}$ formally acting on differentiable $F\in L^{\infty}(\Sigma)$   as
$$ (\mathcal{L}_{0} F)(x,p)=  p \frac{\partial}{\partial x}F(x,p)-\frac{dV}{dx}\big(x\big)\frac{\partial}{\partial p}F(x,p) \\+\int_{\R}dp^{\prime}j(p^{\prime})\big(F(x,p+p^{\prime})-F(x,p)     \big). $$

The null-recurrent behavior for the process $S_{t}=(X_{t},P_{t})$ emerging as $\lambda\rightarrow 0$  at short time scales and the relaxation behavior that takes place on time scales $\propto \lambda^{-1}$ are both apparent in the limiting law $\sqrt{\kappa} \mathbf{B}_{\frak{l}_{t}}$;  the diffusion constant $\kappa\in \R^{+}$ is defined in terms of the jump rates~(\ref{LambdaZero}) which correspond to an unbiased random walk, and on the other hand,  the local time process $\frak{l}_{t}$ is defined in terms of the Ornstein-Uhlenbeck process  which has exponential relaxation (in the correct norm) to the Maxwell-Boltzmann distribution $(\frac{1}{2\pi})^{\frac{1}{2}}   e^{-\frac{1}{2}q^{2}}$.

\subsubsection{The limiting processes} \label{SecLimitProc}

As before we let $\frak{l}$ denote the local time of the Ornstein-Uhlenbeck process $\frak{p}$ and $\mathbf{B}$ be a standard Brownian motion independent of $\frak{p}$.  Recall that the local time process $\frak{l}^{(a)}$ for a point  $a\in \R$ is the a.s. continuous increasing process formally given by
$$\frak{l}_{t}^{(a)}=\int_{0}^{t}dr \delta_{a}(\frak{p}_{r}).  $$
For each realization of the process $\frak{p}$ over the interval $[0,t]$, $\frak{l}_{t}^{(a)}$ is the density of time that the path for $\frak{p}$ spends at $a$, and thus $\int_{\R}da \frak{l}_{t}^{(a)}=t$.  For the case $a=0$, we neglect the superscript for $\frak{l}^{(a)}$.  The values of $\frak{l}$ stay fixed over the time intervals in which $\frak{p}$ moves away from the origin, and thus, in a sense, $\frak{l}$ makes its increases over the set of times such that $\frak{p}_{t}=0$, which has Hausdorff dimension $\frac{1}{2}$.   The fractional diffusion process $\sqrt{\kappa}\mathbf{B}_{\frak{l}}$, appearing as the $\lambda\rightarrow 0$ limit in law of $\lambda^{\frac{1}{4}}D_{\frac{\cdot}{\lambda}}$ in  Thm.~\ref{ThmMain}, has its fluctuations constrained to those  times in which $\frak{l}$ increases.   
Clearly, $\sqrt{\kappa}\mathbf{B}_{\frak{l}}$ is not Markovian since the amount time that the process $\sqrt{\kappa}\mathbf{B}_{\frak{l}}$ has held its current value, i.e., the excursion time of $\frak{p}$ from zero, is correlated with the amount time that it is likely to remain fixed at that value.  The probability densities $\rho_{t}:\R\rightarrow \R^{+}$  of $\sqrt{\kappa}\mathbf{B}_{\frak{l}_{t}}$ satisfy the Volterra-type integro-differential equation
\begin{align}\label{Volterra}
\hspace{1cm}\rho_{t}(q)=\rho_{0}(q)+\frac{\kappa}{2(2\pi)^{\frac{1}{2}} }\int_{0}^{t}dr \frac{ (\Delta_{q}\rho_{r} )(q)  }{\big(1-e^{-\frac{1}{2}(t-r)}     \big)^{\frac{1}{2}}},\hspace{1cm} \rho_{0}(q):=\delta_{0}(q).    
\end{align}
 The non-Markovian nature of the processes $\sqrt{\kappa} \mathbf{B}_{\frak{l}}$ is visible in the convolution form appearing  in~(\ref{Volterra}).  The master equation above is similar to the master equation for a  Brownian motion with diffusion constant $\kappa\in \R^{+}$ time-changed by an independent Mittag-Leffler process $\frak{m}^{(\alpha)}$  of index $0<\alpha<1$.   Note that our limiting processes does not satisfy any scale invariance because $\frak{p}$ does not and thus $\frak{l}$ does not. Some further discussion of local time and related material is included in Appx.~\ref{SecIdealProc}.

 \subsubsection{Related literature}\label{SecRelatedLit} \vspace{.2cm}

The limit theory for integral (or summation) functionals of Markov processes (respectively, chains) usually splits into several standard categories depending on whether the limiting procedure is of first- or second-order and whether the Markov process is positive-recurrent or null-recurrent.   Second-order limit theorems for integral functionals of ergodic Markov processes are well-understood (for instance~\cite{Kipnis}, and see the book~\cite{Landim} for a broader discussion of the literature).   In the null-recurrent case,   second-order limit theory for integral functionals is discussed in~\cite{TouatiUnpub}, in~\cite{Pap,Salminen} when the Markov process is a diffusion, and in~\cite{ChenII} for a Markov chain rather than a process.  The second-order theory is closely related to the limit theory for martingales by a standard construction~(\ref{MartConst}), which seems to have been introduced in~\cite{Gordon} in the analogous case of a chain.   Limit results for  martingales with quadratic variations that are additive functionals of null-recurrent Markov processes can be found in~\cite{TouatiUnpub,Hopfner}.  This literature builds on and applies the limit theory for additive functionals of  Markov processes  (see, for instance,~\cite{Chen,Csaki} and for more recent results~\cite{Loch,Louk}), which began with a paper by Darling and Kac~\cite{Darling}.  
 The monograph~\cite{Hopfner} is a particularly useful reference on this  subject, which, in addition to presenting new results, serves some purpose as a review.

 The usual recipe for finding a martingale close to an integral functional $\int_{0}^{t}dr g(S_{r})$ of a Markov process is given by the following:  if $S_{t}$ is a Harris recurrent Markov process and $g$ is a function defined on its state space such that the reduced resolvent $\frak{R}$ of the backward evolution  operating on $g$ is ``well-behaved" (e.g. lives in a suitable $L^{p}$ space), then
\begin{align} \label{MartConst}
 \tilde{M}_{t}'= (\frak{R}\,g)(S_{t})- (\frak{R}\,g)(S_{0}) +\int_{0}^{t}dr g(S_{r})   
 \end{align}
is a martingale.  The difference between $\int_{0}^{t}dr g(S_{r})$ and   $\tilde{M}_{t}'$ is a pair of terms that are comparatively small in many situations.  For our model, it is not clear  how to obtain the necessary bounds on the reduced resolvent $\big(\frak{R}^{(\lambda)}\,\frac{dV}{dx}\big)(s)$ in the limit $\lambda\rightarrow 0$ to exploit~(\ref{MartConst}), and we use a variant of this martingale; see Lem.~\ref{LemKeyMart}.   To build a martingale approximating $D_{t}$, we expand the state space from $\Sigma$ to $\tilde{\Sigma}=\Sigma\times \{0,1\}$ using a Nummelin splitting-type construction.   The benefit of viewing the process in the extended state space is  that the trajectories for the process $S_{t}$ can be decomposed into a series of nearly i.i.d. parts corresponding to time intervals  $[R_{n},R_{n+1})$, where $R_{n}$ are associated with the return times to an ``atom" identified with the subset $\Sigma\times 1 \subset \tilde{\Sigma}$.    This allows the integral functional $D_{t}$ to be written as a sum of boundary terms plus a random sum of nearly i.i.d. random variables.

The approach of this article differs from that suggested in~\cite{TouatiUnpub} in that we use Nummelin splitting techniques  for the construction of the martingale $\tilde{M}_{t}$ that approximates the integral functional $D_{t}$.  The former result begins with the martingale construction~(\ref{MartConst}) defined in the original statistics and applies  splitting tools to  study the additive functional associated with the predictable quadratic variation process $\langle \tilde{M}'\rangle$.   Our technical apparatus relies on inequalities  for a  generalized resolvent (see~\cite[Sect. 4]{Previous} for the application) given by
$$ \big(U^{(\lambda)}_{h}g\big)(s):=\mathbb{E}^{(\lambda)}_{s}\Big[ \int_{0}^{\infty}dt e^{-\int_{0}^{t}dr h(S_{r})}g(S_{t} )    \Big],    $$
where $h:\Sigma \rightarrow [0,1]$ has compact support~(\ref{Philobuster}), and the evaluation is for a function $g=g_{\lambda}$ that is closely related to the form
\begin{align}\label{YinYang}
g_{\lambda}(s)\approx\Big|\mathbb{E}_{s}^{(\lambda)}\Big[\int_{0}^{\infty}dt\,t\,e^{-t}\frac{dV}{dx}(X_{t})    \Big]\Big| .   
\end{align}
The function $g_{\lambda}(s)$ captures the averaging of the oscillations for $\frac{dV}{dx}(X_{t})$ that occur at high momentum $|P_{t}|\gg 1$.    The generalized resolvent  $\big(U^{(\lambda)}_{h}g\big)(s)$ may appear to be a more difficult object to work than the reduced resolvent $\big(\frak{R}^{(\lambda)}\,\frac{dV}{dx}\big)(s)$, however, the generalized resolvent is an integral of  positive values, and the reduced resolvent is a seemingly delicate cancellation of quantities with opposite sign.  Moreover, the generalized resolvent can be understood in the $\lambda\ll 1$ regime through intuition about the expected amount of time that a random walker should  linger in different parts of phase space before returning to a neighborhood of the origin when beginning from a phase space point $s\in \Sigma$.

We briefly discuss the history of these splitting techniques.  
 For Markov chains a technique for extending the dynamics from a state space $\Sigma$ to $\Sigma\times\{0,1\}$ in order to embed an atom was developed independently in~\cite{Nummelin} and~\cite{Athreya}, and this is referred to as  \textit{Nummelin splitting} or merely \textit{splitting}.   When it comes to the splitting of Markov processes, there are different schemes offered in~\cite{Hopfner} and~\cite{Loch}.  In~\cite{Hopfner} there is a sequence of split processes constructed which contain marginal processes that are arbitrarily close to the original process.  The construction in~\cite{Loch} involves a larger state space $\Sigma\times [0,1]\times \Sigma$ although an exact copy of the original process is embedded as a marginal.  The splitting construction  that we employed in~\cite{Previous} and use in the current article is a truncated version of that in~\cite{Loch} although the split process that we consider is not Markovian because of   the truncation.  The idea of applying splitting techniques to obtain limit theorems for  integral functionals of  null-recurrent Markov processes was introduced in~\cite{TouatiUnpub} and has been developed further in other limit theory in~\cite{Chen,ChenII,Hopfner}.

There are some basic differences that should be emphasized between our model and models for the results mentioned above. The law for our underlying Markovian process $S_{t}$ is itself $\lambda$-dependent.  This is not the case for the limit theorems discussed above in which there is a single fixed Markovian dynamics, and a parameter $\lambda $ only appears in the length of the time intervals considered and in the scaling factors for the variables of interest.  This is why it is possible for us to get a limit law $\sqrt{\kappa}\mathbf{B}_{\frak{l}_{t}}$ that has no scale invariance.  The limit theorems for integral functionals $\int_{0}^{t}dr g(S_{r})$ of null-recurrent Markov processes considered in~\cite{TouatiUnpub,Pap,Salminen} assume that the ``velocity function" $g$ exists in $L^{1}$ with respect to the invariant measure of the process.  This effectively means that the null-recurrent process $S_{t}$ spends most of the time in regions of phase space where $g(S_{t})$ is ``small".   In our case,  the function $g(x,p)=\frac{dV}{dx}(x)$ has no explicit decay as $|p|\rightarrow \infty$, and we rely on the rapid oscillations of $\frac{dV}{dx}(X_{r})$ that occur when $|P_{r}|\gg 1$.  The dependence of $g(x,p)$ on only the torus component $x\in \mathbb{T}$ is thus deceptive, and when time-averaging is properly taken into account as in~(\ref{YinYang}), $g(x,p)$ behaves more like a function that decays with order $|p|^{-2}$ for large $|p|$.

Our techniques could be used to prove analogous results for a related model in~\cite{Old}.  In that case, the limiting law for a rescaling of the pair $(P_{t},D_{t})$ (momentum and integral of the force) would have the form $(\sqrt{\sigma}\mathbf{B}', \sqrt{\kappa}\mathbf{B}_{\frak{l}})$ for some $\sigma, \kappa>0$, where $\mathbf{B}', \mathbf{B}$ are independent copies of standard Brownian motion, and $\frak{l}$ is the local time at zero for $\mathbf{B}'$.

\subsubsection{Comments on Conjecture~\ref{ThmMainII}}

 Conjecture~\ref{ThmMainII}  characterizes the perturbative influence for $\lambda\ll 1$ on the momentum of the particle  when the periodic force is turned on.  The process $\frak{p}_{t,\lambda}$  formally satisfies the Langevin equation
\begin{align}\label{Task}
\hspace{2cm} d\frak{p}_{t,\lambda}=-\frac{1}{2}\frak{p}_{t,\lambda}dt+d\mathbf{B}_{t}'+\lambda^{\frac{1}{4}}\sqrt{\kappa}\,\delta_{0}\big(\frak{p}_{t}\big)   d\mathbf{B}_{t}'' , 
\end{align}
where  $\frak{p}_{0,\lambda}=0$,  $\mathbf{B}'$ and $\frak{p}$ are defined as in~(\ref{TheLimit}), and $\mathbf{B}''$ is a copy of standard Brownian motion independent of $\mathbf{B}'$.    This makes the identification $\int_{0}^{t}d\mathbf{B}_{r}''\delta_{0}(\frak{p}_{r})\equiv  \mathbf{B}_{\frak{l}_{r}}$.  Through equation~(\ref{Task}), $\frak{p}_{t,\lambda}$ has the appearance of  what would be a first-order approximation for $\lambda\ll 1$  of a process  $\frak{p}_{t,\lambda}'$ satisfying the stochastic differential equation
\begin{align*}
\hspace{2cm} d\frak{p}_{t,\lambda}'=-\frac{1}{2} \frak{p}_{t,\lambda}'dt+d\mathbf{B}_{t}'+\lambda^{\frac{1}{4}}\sqrt{\kappa}\,\delta_{0}\big(\frak{p}_{t,\lambda}'\big)   d\mathbf{B}_{t}''. 
\end{align*}
However, this equation can not be made sensible.

\subsection{Organization of the article} 

 Section~\ref{SecLocalTime} contains the proof of Thm.~\ref{ThmLocalTime}, which effectively makes the connection between the normalized momentum process $\lambda^{\frac{1}{2}}P_{\frac{\cdot}{\lambda}}$ and the local time $\frak{l}$ appearing in the limiting law for $\lambda^{\frac{1}{4}}D_{\frac{\cdot}{\lambda}}$.   Section~\ref{SecMartProb} contains a formulation of the ``martingale problem" that determines the uniqueness of the limiting law $\big(\frak{p}, \sqrt{\kappa}\mathbf{B}_{\frak{l}}\big)$ in the proof of Thm.~\ref{ThmMain}.  Section~\ref{SecNumSplit} outlines the construction of a version of the process $S_{t}=(X_{t},P_{t})$ in an enlarged state space.    The proof of Thm.~\ref{ThmMain} is in Sect.~\ref{SecMainProof}.  Finally, the proofs for a few lemmas are postponed to Sect.~\ref{SecMiscProofs},   and Appx.~\ref{SecIdealProc} contains some discussion of the limit process $\mathbf{B}_{\frak{l}}$.  We will make the assumptions of Thm.~\ref{ThmMain} throughout the text.

\section{Convergence of a local time quantity}\label{SecLocalTime}

In this section, we work to prove Thm.~\ref{ThmLocalTime} below.  In the statement of the theorem, the process $L_{t}$ is defined as $$L_{t}:=U^{-1}\int_{0}^{t}dr  \chi\big(H(s)\leq l   \big),$$ where $l:=1+2\sup_{x}V(x)$ and $U\in \R^{+}$ is the Lebesgue measure of the set $\{ H(s)\leq l \}\subset \Sigma $.   The  process $L_{t}$ is important because it is close on the relevant scale to the bracket process $\langle \tilde{M}\rangle_{t}$ of a martingale $\tilde{M}_{t}$ approximating the cumulative drift $D_{t}$; see Lem.~\ref{LemKeyMart} for the definition of $\tilde{M}_{t}$.

\begin{theorem}\label{ThmLocalTime}
Let $\frak{p}_{t}$ be the Ornstein-Uhlenbeck process and $\frak{l}_{t}$ be its local time at zero.  As $\lambda\rightarrow 0$,  there is convergence in law
$$\hspace{4cm}\big(\lambda^{\frac{1}{2}} P_{\frac{t}{\lambda}},\,  \lambda^{\frac{1}{2} }L_{\frac{t}{\lambda}}\big)\stackrel{\frak{L}}{\Longrightarrow}(\frak{p}_{t},\frak{l}_{t}), \hspace{2cm} t\in [0,T],$$
where the convergence is with respect to the uniform metric.  Moreover, for any $t\in \R^{+}$
$$  \sup_{\lambda<1} \mathbb{E}^{(\lambda)}\big[ \lambda^{\frac{1}{2} }L_{\frac{t}{\lambda}}\big] <\infty \quad \text{and} \quad  \lim_{\lambda\rightarrow 0}\mathbb{E}^{(\lambda)}\big[ \lambda^{\frac{1}{2} }L_{\frac{t}{\lambda}}\big]= \mathbb{E}\big[ \frak{l}_{t}\big].   $$

\end{theorem}

We begin by making some remarks on the local time process $\frak{l}$.  Appendix~\ref{SecIdealProc} contains more information although without proofs.   
Define $\mathbf{\tilde{B}}_{t} =\int_{0}^{t}dr \textup{sgn}(\frak{p}_{r})d\mathbf{B}_{r}'$, where $\mathbf{B}'$  is the Brownian motion driving the Langevin equation~(\ref{TheLimit}) and $\textup{sgn}:\R\rightarrow \{\pm 1\}$ is the sign function.  The Tanaka-Meyer formula yields the local time at zero for $\frak{p}$ as
\begin{align}\label{ItoLocal}
\frak{l}_{t}= |\frak{p}_{t}|-|\frak{p}_{0}|-\mathbf{\tilde{B}}_{t}+\frac{1}{2}\int_{0}^{t}dr|\frak{p}_{r}|.
\end{align}
  The above relation follows from the formal definition $\frak{l}_{t}=\int_{0}^{t}dr \delta_{0}(\frak{p}_{r})$ and a formal application of the Ito formula for the function $|\cdot|$ of the process $\frak{p}$ that has differential $d\frak{p}_{t}=-\frac{1}{2}\frak{p}_{t}dt+d\mathbf{B}_{t}'$.  In~(\ref{ItoLocal}) $\frak{l}$ is the positive part of the drift for the diffusion process $\frak{p}$.

  Theorem~\ref{ThmLocalTime} states that a rescaling of the process $L_{t}$  converges in law to the local time $\frak{l}_{t}$.   Since $h(x,p)$  is compactly supported, it is not surprising that this quantity would be related to the local time when considered on  the appropriate scale: $\lambda^{\frac{1}{2}}L_{\frac{t}{\lambda}}$, $\lambda\ll 1$.    
The strategy in the proof resembles~\cite[Thm. 3.1]{Quantum} in which information related to the limiting behavior for the momentum process $P_{t}$ is found through a study of the semimartingale decomposition of the square root energy process $\mathbf{Q}_{t}:=(2H_{t})^{\frac{1}{2}}=\big(P_{t}^{2}+2V(X_{t})\big)^{\frac{1}{2}}$.   Since the potential $V(x)$ is bounded,  we have that  $\lambda^{\frac{1}{2}} |P_{\frac{t}{\lambda} }|\approx \lambda^{\frac{1}{2}}\mathbf{Q}_{\frac{t}{\lambda} }$.  The advantage of working with a function of the Hamiltonian is that there is no drift between collisions.  Let the processes $\mathbf{M}_{t}$, $\mathbf{A}_{t}^{+}$, and $-\mathbf{A}_{t}^{-}$ be respectively the martingale, predictable increasing, and predictable decreasing parts in the semimartingale decomposition of the process $\mathbf{Q}_{t}$.  The processes $\mathbf{A}_{t}^{\pm}$ and the predictable quadratic variation $\langle\mathbf{M}\rangle_{t}$ of the martingale $\mathbf{M}_{t}$ have the forms 
 \begin{align}\label{ZipZap}
  \mathbf{A}_{t}^{\pm}= \int_{0}^{t}dr \mathcal{A}_{\lambda}^{\pm}(S_{r})  \hspace{1cm} \text{and} \hspace{1cm}
  \langle \mathbf{M} \rangle_{t}= \int_{0}^{t}dr \mathcal{V}_{\lambda}(S_{r}),      
  \end{align}
where $\mathcal{A}_{\lambda}^{\pm}, \mathcal{V}_{\lambda}:\Sigma\rightarrow \R$ are defined below.

Also, since $L_{t}$ is difficult to work with directly, our strategy is to approximate it by $\mathbf{A}_{t}^{+}$.  Notice that we can rewrite the components in the semimartingale decomposition as
\begin{align}\label{Siren}
 \mathbf{A}_{t}^{+}=\mathbf{Q}_{t}-\mathbf{Q}_{0}-\mathbf{M}_{t}+\mathbf{A}_{t}^{-}.  
 \end{align}
in analogy with the Tanaka-Meyer formula~(\ref{ItoLocal}).  We approach the term $\lambda^{\frac{1}{2}}\mathbf{A}_{\frac{t}{\lambda}}^{+}$ through a study of the joint convergence of the terms
$$ \lambda^{\frac{1}{2}} \mathbf{Q}_{\frac{t}{\lambda}}\stackrel{\frak{L}}{\Longrightarrow} |\frak{p}_{t}|, \quad \, \lambda^{\frac{1}{2}}\mathbf{M}_{\frac{t}{\lambda} }   \stackrel{\frak{L}}{\Longrightarrow} \mathbf{\tilde{B}}_{t},\quad  \lambda^{\frac{1}{2}} \mathbf{A}_{\frac{t}{\lambda}}^{-}\stackrel{\frak{L}}{\Longrightarrow}\frac{1}{2}\int_{0}^{t}dr | \frak{p}_{r} |.  $$

Readers accustomed to the limit theory in Jacod and Shiryaev's book~\cite{Jacod} may find the appearance of the uniform  metric rather than the Skorokhod metric in the statement of Thm.~\ref{ThmLocalTime} unusual.  There is a result for the weak convergence of martingales with respect to the uniform metric in~\cite[Thm. VIII.2.13]{Pollard}.  This limit theorem for martingales has a role in proving the convergence in law of $\lambda^{\frac{1}{2}}P_{\frac{t}{\lambda}}$ with respect to the uniform metric of~\cite[Thm. 1.3]{Previous}.   Although $\lambda^{\frac{1}{2}}L_{\frac{t}{\lambda}}$ and equivalently  $\lambda^{\frac{1}{2}}\mathbf{A}_{\frac{t}{\lambda}}^{+}$ may seem to be converging to a ``singular" limit that would be awkward to treat with the uniform metric, the process $\lambda^{\frac{1}{2}}\mathbf{A}_{\frac{t}{\lambda}}^{+}$ can be viewed through~(\ref{Siren}) as a sum of terms that are more clearly treatable in the uniform metric. 

The next lemma gives a limiting procedure in which the trajectories for $\frak{l}$ and  $\mathbf{\tilde{B}}$ in the Tanaka-Meyer formula~(\ref{ItoLocal}) are determined by the trajectories for $|\frak{p}|$.

\begin{lemma}\label{LocalEstimate}
Let $\frak{p}_{t}$ be the Ornstein-Uhlenbeck process.  As $\epsilon\rightarrow 0$, the local time at zero $\frak{l}$ satisfies
$$\mathbb{E}\Big[\sup_{0\leq t
\leq T} \Big| \frak{l}_{t}-\frac{1}{2\epsilon}\int_{0}^{t}dr e^{-\frac{|\frak{p}_{r}  |}{\epsilon} }\Big|\Big]=\mathit{O}(\epsilon^{\frac{1}{2}}). $$
Also, the Brownian motion $\mathbf{\tilde{B}}_{t}$ in the Tanaka-Meyer formula~(\ref{ItoLocal}) satisfies
$$\mathbb{E}\Big[\sup_{0\leq t
\leq T} \Big| \mathbf{\tilde{B}}_{t}- |\frak{p}_{t}|+|\frak{p}_{0}|-\epsilon e^{-\frac{|\frak{p}_{t}| }{\epsilon} } - \frac{1}{2} \int_{0}^{t}dr |\frak{p}_{r}|\, \big(1-e^{-\frac{|\frak{p}_{r}| }{\epsilon} }\big)+\frac{1}{2\epsilon} \int_{0}^{t}dr e^{-\frac{|p_{r}|}{\epsilon} }\Big|\Big]=\mathit{O}(\epsilon^{\frac{1}{2}}). 
$$

\end{lemma}

Before proceeding to the proof of Thm.~\ref{ThmLocalTime}, we must recall some notation from~\cite{Previous}. For $n\in \R^{+}$ define the functions $\mathcal{A}_{\lambda},\mathcal{V}_{\lambda,n}:\mathbb{T}\times \R\rightarrow \R$ as 
\begin{align*}
 \mathcal{A}_{\lambda}(x,p):=&\int_{\R}dp^{\prime} \big( 2^{\frac{1}{2}}H(x,p^{\prime})^{\frac{1}{2}}- 2^{\frac{1}{2}}H(x,p)^{\frac{1}{2}}       \big)\mathcal{J}_{\lambda}(p,p^{\prime}),\\  
\mathcal{V}_{\lambda,n}(x,p):=&
\int_{\R}dp^{\prime} \big| 2^{\frac{1}{2}}H(x,p^{\prime})^{\frac{1}{2}}- 2^{\frac{1}{2}}H(x,p)^{\frac{1}{2}}\big|^{2n}      \mathcal{J}_{\lambda}(p,p^{\prime}).
\end{align*}
The function $\mathcal{V}_{\lambda}:=\mathcal{V}_{\lambda,1}$ is related to the predictable quadratic variation of the martingale $\mathbf{M}$ through~(\ref{ZipZap}).  We also denote the escape rates by $\mathcal{E}_{\lambda}(p):=\int_{\R}dp'\mathcal{J}_{\lambda}(p,p')$.  We define $\mathcal{A}_{\lambda}^{\pm}(s)= \max(\pm \mathcal{A}_{\lambda}(s),0)$ to be the positive and negative parts of  $\mathcal{A}_{\lambda}$. Proposition~\ref{AMinus}  contains some useful inequalities regarding the functions $\mathcal{A}^{\pm}_{\lambda}$, $\mathcal{V}_{\lambda,n}$, and we do not include the  proof, which is based on elementary inequalities and calculus.

\begin{proposition}\label{AMinus} There are constants $c,C, C_n>0$ such that for $\lambda<1$  the following inequalities hold:  
\begin{enumerate}

\item For all $(x,p)\in \Sigma$,  $\mathcal{V}_{\lambda,n}(x,p)\leq C_n(1+\lambda |p|)^{2n+1}$.

\item For all $(x,p)\in \Sigma$, $\mathcal{A}_\lambda^+(x,p)\leq \frac{C}{1+p^2}$.

\item  For $\lambda^{-\frac{3}{8}}\leq |p|\leq \lambda^{-\frac{3}{4} }$,
$$ \Big|\mathcal{A}_{\lambda}^{-}(x,p)-\frac{1}{2}\lambda |p| \Big| \leq C\lambda^{\frac{5}{4} }|p| \quad \text{and}\quad  \Big| \mathcal{V}_{\lambda}(x,p)-1       \Big|\leq C\lambda^{\frac{1}{2}}   .  $$

\item  For $|p|\leq \lambda^{-1}$,  $\mathcal{A}_\lambda^-(x,p)\leq C\lambda|p|$.

\item For all $(x,p)\in \Sigma$, $\left|\frac{\mathcal{A}_\lambda(x,p)}{\mathcal{E}_\lambda(p)} + \frac{2\lambda|p|}{1+\lambda}\right| \leq C$.

\item For all $p\in \R$, $ \mathcal{E}_\lambda(p)\leq \frac{1}{8(\lambda +1)}\left(1+C\lambda |p|\right)$ and $\lambda |p| \leq C\mathcal{E}_\lambda(p)$. 

\item  As $\lambda\rightarrow 0$, we have $\int_{\Sigma}ds  \mathcal{A}^{+}_{\lambda}(s)=1+\mathit{O}(\lambda^{\frac{1}{2}})$.

\end{enumerate}
\end{proposition}

Lemmas~\ref{FirstEnergyLem} and~\ref{LowEnergyLemma} below  characterize the typical energy behavior over the time interval $[0,\frac{T}{\lambda}]$ for $\lambda\ll 1$.  In particular, Lem.~\ref{FirstEnergyLem} states that the energy $H(X_{t},P_{t})=\frac{1}{2}P_{t}^{2}+V(X_{t}):=H_{t}$ does not typically go above the scale $\lambda^{-1}$, and Lem.~\ref{LowEnergyLemma} states that the energy typically does not spend much time smaller than $\lambda^{-\varrho}$ for any $0\leq \varrho<1$.  The proof for Lem.~\ref{LowEnergyLemma} is contained in Sect.~\ref{SecMiscProofs} and  Lem.~\ref{FirstEnergyLem} is from~\cite[Lem. 3.2]{Previous}.

\begin{lemma}\label{FirstEnergyLem}   For any $n\in \mathbb{N}$, there exists a $C>0$ such that
$$  \mathbb{E}^{(\lambda)}\Big[\sup_{0\leq r\leq \frac{T}{\lambda} } (H_{r})^{\frac{n}{2}} \Big] \leq  C\Big(\frac{ T}{\lambda}\Big)^{\frac{n}{2}}    $$
for all $T>0$ and $\lambda<1$.  

\end{lemma}

\begin{lemma}\label{LowEnergyLemma}
Define $\mathbf{T}_{t}=\lambda \int_{0}^{t}dr \chi(H_{r}\leq \epsilon \lambda^{-\varrho } )$ for  $0\leq   \varrho\leq 1$.  For any fixed $T>0$, there is a $C>0$ such that for small enough    $\lambda$ and all $ \epsilon \in [ \lambda^{\varrho},1]$,
$$ \mathbb{E}^{(\lambda)}\big[\mathbf{T}_{\frac{T}{\lambda}} \big] \leq C \epsilon^{\frac{1}{2}} \lambda^{\frac{1-\varrho}{2}}.       $$
\end{lemma}

The following lemma is reminiscent of ratio limit theorems for additive functionals of null-recurrent Markov processes since $L_{t}$ and $\mathbf{A}_{t}^{+}$ are time integrals of $S_{r}$ evaluating the velocity functions $U^{-1}\chi(H(s)\leq l   )$ and $\mathcal{A}^{+}(s)$, respectively.  To support this intuition, recall that the invariant measure for the Markov process $S_{t}$ ``approaches" Lebesgue measure on $\Sigma$ for small $\lambda\in \R^{+}$ and observe that $$U^{-1}\int_{\Sigma}ds\chi(H(s)\leq l   )=1= \int_{\Sigma}ds\mathcal{A}_{\lambda}^{+}(s)+\mathit{O}(\lambda^{\frac{1}{2}}),$$ where the second equality is by Part (7) of Prop.~\ref{AMinus}.  The proof of Lem.~\ref{LocalTimeBnd} is placed in Sect.~\ref{SecMiscProofs}.

\begin{lemma}\label{LocalTimeBnd} As $\lambda\rightarrow 0$,
$$ \mathbb{E}^{(\lambda)}\Big[\sup_{0\leq t\leq T}\Big|\lambda^{\frac{1}{2}} L_{\frac{t}{\lambda}}    -\lambda^{\frac{1}{2} }\mathbf{A}_{\frac{t}{\lambda}}^{+} \Big|      \Big]=\mathit{O}(\lambda^{\frac{1}{4}}).     $$
Moreover, there is a $C>0$ such that for $\lambda<1$, 
$$\mathbb{E}^{(\lambda)}\big[ \lambda^{\frac{1}{2}} L_{\frac{T}{\lambda}}\big] \leq C .$$
\end{lemma}

\vspace{.5cm}

\begin{proof}[Proof of Thm.~\ref{ThmLocalTime}]
By~\cite[Thm. 1.3]{Previous} the process $\lambda^{\frac{1}{2}} P_{\frac{\cdot}{\lambda}}$ converges in law to the Ornstein-Uhlenbeck process $\frak{p}$ with respect to the uniform metric.  As a consequence of limiting scheme in Lem.~\ref{LocalEstimate}, the trajectories for the first component  of the limiting pair $(\frak{p},\frak{l})$ determine the trajectories of the second component through the absolute value $|\frak{p}|$.
It is sufficient for us to show that $( |\lambda^{\frac{1}{2}} P_{\frac{\cdot}{\lambda}}|  ,\,\lambda^{\frac{1}{2}} L_{\frac{\cdot}{\lambda}})$ converges in law to the pair $(|\frak{p}|,\frak{l})$.  Our approach will be to approximate the pair $( |\lambda^{\frac{1}{2}} P_{\frac{t}{\lambda}}|  ,\,\lambda^{\frac{1}{2}} L_{\frac{t}{\lambda}})$ by the pair $(\lambda^{\frac{1}{2}}\mathbf{Q}_{\frac{t}{\lambda}},\, \lambda^{\frac{1}{2}} \mathbf{A}_{\frac{t}{\lambda}}^{+} )$ in Part (i) below,  and then to apply an argument based on the Tanaka-Meyer formula to analyze $(\lambda^{\frac{1}{2}}\mathbf{Q}_{\frac{t}{\lambda}},\, \lambda^{\frac{1}{2}} \mathbf{A}_{\frac{t}{\lambda}}^{+} )$ in  Part (ii).   All convergences in law in this proof are with respect to the uniform metric.

  \vspace{.5cm}

\noindent (i).  Showing that $|\lambda^{\frac{1}{2}} P_{\frac{t}{\lambda}}|$ is close to  $\lambda^{\frac{1}{2}}\mathbf{Q}_{\frac{t}{\lambda}}$ is easy since
$$\big| \big(p^{2}+2V(x) \big)^{\frac{1}{2}}-|p|   \big|\leq \big(2\sup_{x}V(x)\big)^{\frac{1}{2}}  \quad \text{and thus}\quad \big|\lambda^{\frac{1}{2}}\mathbf{Q}_{\frac{t}{\lambda}}-\lambda^{\frac{1}{2}}|P_{\frac{t}{\lambda}}|\big|\leq \lambda^{\frac{1}{2}}\big(2\sup_{x}V(x)\big)^{\frac{1}{2}}. $$ 
By Lem.~\ref{LocalTimeBnd}, $ \lambda^{\frac{1}{2}} L_{\frac{\cdot}{\lambda}}$ can be approximated by  $\lambda^{\frac{1}{2} }\mathbf{A}_{\frac{\cdot}{\lambda}}^{+}$  for small $\lambda$, and the expectation $\mathbb{E}^{(\lambda)}\big[ \lambda^{\frac{1}{2} }L_{\frac{t}{\lambda}}\big] $ is uniformly bounded for $\lambda<1$.  A consequence  of Part (ii) will be that  $\lambda^{\frac{1}{2} }L_{\frac{\cdot}{\lambda}}$ converges in law to $\frak{l}$ as $\lambda\rightarrow 0$.  This implies convergence of the first moment.

\vspace{.5cm}

\noindent (ii).    The process $\lambda^{\frac{1}{2}}|P_{\frac{\cdot}{\lambda}}|$ converges in law to $|\frak{p}|$ because $|\cdot|$ is a continuous map on functions in $L^{\infty}([0,T])$ with respect to the supremum norm and  $\lambda^{\frac{1}{2}}P_{\frac{\cdot}{\lambda}}$ converges in law to  $\frak{p} $  by~\cite[Thm. 1.3]{Previous}. 
 With Part (i) it follows that $\lambda^{\frac{1}{2}}\mathbf{Q}_{\frac{\cdot}{\lambda}}$ converges in law  to  $|\frak{p}|$.  Our main work is to incorporate the component $ \lambda^{\frac{1}{2}} \mathbf{A}_{\frac{\cdot}{\lambda}}^{+}$ for the convergence in law of the pair  $(\lambda^{\frac{1}{2}}\mathbf{Q}_{\frac{\cdot}{\lambda}},\, \lambda^{\frac{1}{2}} \mathbf{A}_{\frac{\cdot}{\lambda}}^{+} )$.

For the process $ \mathbf{A}_{t}^{+}$, we may write
\begin{align} \label{LikeItoLocal}
\mathbf{A}_{t}^{+}=  \mathbf{Q}_{t}- \mathbf{Q}_{0}-\mathbf{M}_{t}+\mathbf{A}_{t}^{-}.
\end{align}
Now, we will begin the analysis of $\lambda^{\frac{1}{2}}\mathbf{A}_{\frac{t}{\lambda}}^{+}$ through a study of the terms on the right side of the above equation.   By our assumptions on the  initial distribution $\mu$ for $(X_{0},P_{0})$, the random variable $\lambda^{\frac{1}{2}}\mathbf{Q}_{0}$ converges to zero in probability.   We will show that there is convergence in law
\begin{align}\label{Triple}
\mathbf{Y}^{(\lambda)}_{t}=\big(\lambda^{\frac{1}{2}}\mathbf{Q}_{\frac{t}{\lambda}},\,\lambda^{\frac{1}{2}} \mathbf{M}_{\frac{t}{\lambda}} ,\, \lambda^{\frac{1}{2}}\mathbf{A}_{\frac{t}{\lambda}}^{-}\big) \stackrel{\frak{L}}{\Longrightarrow} \Big(|\frak{p}_{t}|,\,\mathbf{\tilde{B}}_{t},\,\frac{1}{2}\int_{0}^{t}dr|\frak{p}_{r}| \Big), 
\end{align}
where $\mathbf{\tilde{B}}$ is the copy of Brownian motion in the Tanaka-Meyer formula~(\ref{ItoLocal}).  With the identities~(\ref{ItoLocal}) and~(\ref{LikeItoLocal}), the above convergence implies that  $( \lambda^{\frac{1}{2}}\mathbf{Q}_{\frac{\cdot}{\lambda}},\lambda^{\frac{1}{2}}\mathbf{A}^{+}_{\frac{\cdot}{\lambda}})$ converges in law to $(|\frak{p}|,\frak{l})$.  To prove the convergence~(\ref{Triple}), we will first show that $\lambda^{\frac{1}{2}}\mathbf{A}_{\frac{t}{\lambda}}^{-}$ can be approximated by $\frac{1}{2}\int_{0}^{t}dr \lambda^{\frac{1}{2}}\mathbf{Q}_{\frac{r}{\lambda}}$; see (I) below. 
 It is then enough to show functional convergence of the pair $\big(\lambda^{\frac{1}{2}}\mathbf{Q}_{\frac{\cdot}{\lambda}},\,\lambda^{\frac{1}{2}} \mathbf{M}_{\frac{\cdot}{\lambda}}\big) $ because the map sending $q\in L^{\infty}([0,T])$ to the element $\frac{1}{2}\int_{0}^{\cdot}dr q_{r}\in  L^{\infty}([0,T])$ is continuous with respect to the supremum norm. 
 A similar idea applies in the proof of the  convergence in law of  $\big(\lambda^{\frac{1}{2}}\mathbf{Q}_{\frac{\cdot}{\lambda}},\,\lambda^{\frac{1}{2}} \mathbf{M}_{\frac{\cdot}{\lambda}}\big) $.  It is clear from the statement of Lem.~\ref{LocalEstimate} that the trajectories for  $|\frak{p}|$ determine the trajectories for  $\mathbf{\tilde{B}}$, and the same relation emerges between $ \lambda^{\frac{1}{2}}\mathbf{Q}_{\frac{\cdot}{\lambda}} $ and $\lambda^{\frac{1}{2}} \mathbf{M}_{\frac{\cdot}{\lambda}} $ in the limit $\lambda\rightarrow 0$.  The main idea of the proof is to reduce everything to the functional convergence of $ \lambda^{\frac{1}{2}}\mathbf{Q}_{\frac{\cdot}{\lambda}} $ to the absolute value of the Ornstein-Uhlenbeck process $|\frak{p}|$, which we know to occur by the observation following (ii) above.


The analysis below will be split into the proof of statements (I)-(III) below.  The proofs of (II) and (III) work toward the convergence of the pair $\big(\lambda^{\frac{1}{2}}\mathbf{Q}_{\frac{\cdot}{\lambda}},\,\lambda^{\frac{1}{2}} \mathbf{M}_{\frac{\cdot}{\lambda}} \big) $.    
\begin{enumerate}[(I).]

\item There is a $C>0$ such that for all $\lambda< 1$,
$$\mathbb{E}^{(\lambda)}\Big[ \sup_{0\leq t\leq T}  \Big|  \lambda^{\frac{1}{2}}\mathbf{A}_{\frac{t}{\lambda}}^{-}-\frac{1}{2}\int_{0}^{t}dr \lambda^{\frac{1}{2}}\mathbf{Q}_{\frac{r}{\lambda}}  \Big| \Big]\leq C\lambda^{\frac{1}{8}}.$$

\item The martingales $ \mathbf{m}^{(\lambda)}_{t,\epsilon}$ defined as     
 $$  \frak{m}_{t,\epsilon}^{(\lambda)}:= \lambda^{\frac{1}{2}} \int_{0}^{\frac{t}{\lambda} }d\mathbf{M}_{r}\big(1-e^{-\epsilon^{-1}\lambda^{\frac{1}{2}}\mathbf{Q}_{r^{-}}}     \big)   $$
 are close to $\lambda^{\frac{1}{2}} \mathbf{M}_{\frac{t}{\lambda}}$ for small $\lambda$ and $\epsilon$ in the sense 
\begin{align}\label{Calcutta}
\mathbb{E}^{(\lambda)}\big[ \sup_{0\leq t\leq T}  \big| \lambda^{\frac{1}{2}} \mathbf{M}_{\frac{t}{\lambda} }- \frak{m}_{t,\epsilon}^{(\lambda)} \big|^{2} \big]\leq C \big(\textup{max}(\epsilon, \lambda)\big) ^{\frac{1}{2}} 
\end{align}
for some $C$ and all $\lambda, \epsilon <1$.

\item  For each fixed $\epsilon \in \R^{+}$,  there is convergence in law as $\lambda\rightarrow 0$
$$  \big(\lambda^{\frac{1}{2}}\mathbf{Q}_{\frac{t}{\lambda}} ,\, \frak{m}_{t,\epsilon}^{(\lambda)}\big)\stackrel{\frak{L}}{\Longrightarrow}  \big(|\frak{p}_{t}|,\,\mathbf{m}_{t,\epsilon}\big),    $$
for $\mathbf{m}_{t,\epsilon}= \int_{0}^{t }d\mathbf{\tilde{B}}_{r}\big(1-e^{-\frac{|\frak{p}_{r}|}{\epsilon}  }\big)$.

\end{enumerate}

The $\textup{max}(\epsilon , \lambda) $ on the right side of the inequality~(\ref{Calcutta}) can be replaced with $\epsilon$ by having a slightly more refined version of Lem.~\ref{LowEnergyLemma}, which we do not require here.  By combining the results (II) and (III) with Lem.~\ref{LocalEstimate}, which gives the convergence as $\epsilon\rightarrow 0$ of $(\frak{p},  \mathbf{m}_{\cdot,\epsilon}  )$ to $(\frak{p},  \mathbf{\tilde{B}} )$ in the norm $\|\cdot\|_{\frak{s}}= \mathbb{E}\big[\sup_{0\leq t\leq T}|\cdot|\big]$, then a standard argument which we sketch below shows that $\big(\lambda^{\frac{1}{2}}\mathbf{Q}_{\frac{\cdot}{\lambda}},\,\lambda^{\frac{1}{2}} \mathbf{M}_{\frac{\cdot}{\lambda}} \big) $ converges in law to $ (\frak{p},  \mathbf{\tilde{B}} )$.  These statements can be summarized by the marked arrows in the diagram below
        $$  \renewcommand{\arraystretch}{1.5} \begin{array}[c]{ccc} (  \lambda^{\frac{1}{2}}\mathbf{Q}_{\frac{t}{\lambda}},\, \mathbf{m}^{(\lambda)}_{t,\epsilon} )& \stackrel{\frak{L}}{\Longrightarrow   }  &  (\frak{p}_{t},\, \mathbf{m}_{t,\epsilon})\\   \Big \downarrow\scriptstyle{ \|\cdot\|_{\frak{s}} }&  &  \Big \downarrow\scriptstyle{  \|\cdot\|_{\frak{s}} }\\   (  \lambda^{\frac{1}{2}}\mathbf{Q}_{\frac{t}{\lambda}},\,\lambda^{\frac{1}{2}}\mathbf{M}_{\frac{t}{\lambda}}) &\Longrightarrow  &  (\frak{p}_{t},\,\mathbf{\tilde{B}}_{t})
\end{array},
$$
where the convergence on the right side of the diagram is by Lem.~\ref{LocalEstimate}, the top of the diagram is by (III), and the convergence on the left side of the diagram is from (II) and requires both $\epsilon$ and $\lambda$ to be small.  Let us sketch the proof of the convergence in law at the bottom line of the diagram.  By~\cite[Cor. IV.2.9]{Pollard} it is enough to show the convergence as $\lambda\rightarrow 0$ of
\begin{align}\label{Thorns}
\big|  \mathbb{E}^{(\lambda)}\big[F(  \lambda^{\frac{1}{2}}\mathbf{Q}_{\frac{\cdot}{\lambda}},\,\lambda^{\frac{1}{2}}\mathbf{M}_{\frac{\cdot}{\lambda}})   \big]- \mathbb{E}^{(\lambda)}\big[F(\frak{p},\,\mathbf{\tilde{B}})  \big]\big|\
\end{align}
to zero for functionals $F:L^{\infty}([0,T],\R^{2})\rightarrow \R$ that are bounded and uniformly continuous with respect to the supremum norm.  By the triangle inequality,~(\ref{Thorns}) is smaller than
\begin{align}
  \big| & \mathbb{E}^{(\lambda)}\big[F(  \lambda^{\frac{1}{2}}\mathbf{Q}_{\frac{\cdot}{\lambda}},\,\lambda^{\frac{1}{2}}\mathbf{M}_{\frac{\cdot}{\lambda}})   \big]- \mathbb{E}^{(\lambda)}\big[F (  \lambda^{\frac{1}{2}}\mathbf{Q}_{\frac{\cdot}{\lambda}},\, \mathbf{m}^{(\lambda)}_{\cdot,\epsilon} )  \big]\big|+\big|  \mathbb{E}^{(\lambda)}\big[F (  \lambda^{\frac{1}{2}}\mathbf{Q}_{\frac{\cdot}{\lambda}},\, \mathbf{m}^{(\lambda)}_{\cdot,\epsilon} )  \big]-\mathbb{E}^{(\lambda)}\big[F (\frak{p},\, \mathbf{m}_{\cdot,\epsilon})   \big]\big|\nonumber \\ &  +\big|  \mathbb{E}^{(\lambda)}\big[F (\frak{p},\, \mathbf{m}_{\cdot,\epsilon})   \big]-\mathbb{E}^{(\lambda)}\big[F (\frak{p},\,\mathbf{\tilde{B}}) \big]\big|.
\end{align}
Since $F$ is bounded and uniformly continuous, we can choose $\epsilon\vee \lambda$ and $\epsilon$ to make both the first and third terms small by (III) and Lem.~\ref{LocalEstimate}, respectively.  We can then choose $\lambda \in (0,\epsilon]$ to make the second term arbitrarily small by the convergence (II).

\vspace{.5cm}

Next, we prove statements (I)-(III).  The definition of constants $C_{n},C_{n}'>0$, $n\in \mathbb{N}$ will reset in different parts of the analysis. \vspace{.2cm}

\noindent (I).  \hspace{.25cm}  By the remark (ii), it is sufficient to bound the difference between $\lambda^{\frac{1}{2}}\mathbf{A}_{\frac{t}{\lambda}}^{-} $ and  $\frac{1}{2} \int_{0}^{t}dr|\lambda^{\frac{1}{2}} P_{\frac{r}{\lambda} }|$ for small $\lambda$. Conditioned on the event that $ \sup_{0\leq r\leq \frac{T}{\lambda}}|P_{r }|\leq \lambda^{-\frac{3}{4}}$ for $r\in [0,\frac{T}{\lambda}]$, then 
\begin{align*}
\sup_{0\leq t\leq T}& \Big| \lambda^{\frac{1}{2}}\mathbf{A}_{\frac{t}{\lambda}}^{-}-\frac{1}{2} \int_{0}^{t}dr |\lambda^{\frac{1}{2}} P_{\frac{r}{\lambda} }| \Big|\\ &\leq  \lambda^{\frac{1}{8}}C_{1} \int_{0}^{T}dr\chi\big(|P_{\frac{r}{\lambda}}|\leq \lambda^{-\frac{3}{8}}   \big)+ \lambda^{\frac{1}{2}} \int_{0}^{T}dr\chi\big(|P_{\frac{r}{\lambda}}|\geq \lambda^{-\frac{3}{8}}   \big)\, \Big| \lambda^{-1}\mathcal{A}_{\lambda}^{-}\big( X_{\frac{r}{\lambda}},P_{\frac{r}{\lambda}}    \big) -  \frac{1}{2}|P_{\frac{r}{\lambda}}|  \Big|  \\ &\leq  C_{1}T\lambda^{\frac{1}{8}}+C_{2} T\lambda^{\frac{3}{4}}\sup_{0\leq r\leq \frac{T}{\lambda} }|P_{r}| , 
\end{align*}
where $C_{1}:=\frac{1}{2}+  \sup_{ |p|\leq \lambda^{-\frac{3}{8}} } \lambda^{-\frac{5}{8} }\mathcal{A}_{\lambda}^{-}(x,p)   $, and $C_{1}$ is finite by Part (4) of Prop.~\ref{AMinus}.  The $C_{2}>0$ in the second inequality is from Part (3) of Prop.~\ref{AMinus}.

The above implies the first inequality below:
\begin{align*}
\mathbb{E}^{(\lambda)}&\Big[ \chi\big(\sup_{0\leq r\leq \frac{T}{\lambda}}   |P_{r}|\leq \lambda^{-\frac{3}{4}}\big) \sup_{0\leq t\leq T} \Big| \lambda^{\frac{1}{2}}\mathbf{A}_{\frac{t}{\lambda}}^{-}-\frac{1}{2} \int_{0}^{t}dr |\lambda^{\frac{1}{2}} P_{\frac{r}{\lambda} }| \Big|   \Big] \leq C_{1}T\lambda^{\frac{1}{8}}+C_{2}\lambda^{\frac{3}{4}}\mathbb{E}^{(\lambda)}\big[  \sup_{0\leq r\leq \frac{T}{\lambda} }\big| P_{r}\big|\big]\\ &\leq  C_{1}T\lambda^{\frac{1}{8}}+C_{2}2^{-\frac{1}{2}}T\lambda^{\frac{3}{4}}\mathbb{E}^{(\lambda)}\big[  \sup_{0\leq r\leq \frac{T}{\lambda} }\mathbf{Q}_{r}\big]\leq C_{1}T\lambda^{\frac{1}{8}}+C_{2}'T\lambda^{\frac{1}{4}},    
\end{align*}
where the second and third inequalities follows from $P_{r}^{2}\leq \mathbf{Q}_{r}^{2}=2H_{r}$ and  by Lem.~\ref{FirstEnergyLem}, respectively.  Moreover, for the event   $\sup_{0\leq r\leq \frac{T}{\lambda}}|P_{r}|> \lambda^{-\frac{3}{4}}$, then 
\begin{align*}
 \mathbb{E}^{(\lambda)}&\Big[\chi\big(\sup_{0\leq r\leq \frac{T}{\lambda}}   |P_{r}|> \lambda^{-\frac{3}{4}}\big) \sup_{0\leq t\leq T} \Big| \lambda^{\frac{1}{2}}\mathbf{A}_{\frac{t}{\lambda}}^{-}-\frac{1}{2} \int_{0}^{t}dr |\lambda^{\frac{1}{2}} P_{\frac{r}{\lambda} }| \Big|   \Big]& \nonumber \\& \leq   \mathbb{P}^{(\lambda)}\Big[ \sup_{0\leq r\leq \frac{T}{\lambda}}   |P_{r}|> \lambda^{-\frac{3}{4}} \Big]^{\frac{1}{2}}\mathbb{E}^{(\lambda)}\Big[ \Big|  \int_{0}^{T}dr\Big( |\lambda^{\frac{1}{2}} P_{\frac{r}{\lambda} }| +\lambda^{\frac{1}{2}}\mathcal{A}_{\lambda}^{-}(X_{\frac{r}{\lambda}},P_{\frac{r}{\lambda}}\big) \Big)\Big|^{2}   \Big]^{\frac{1}{2}}\nonumber   \\ &\leq  C_{1}'\lambda^{\frac{1}{2}}T^{\frac{1}{2}}\mathbb{E}^{(\lambda)}\Big[\Big( \sup_{0\leq r\leq \frac{T}{\lambda}}  \lambda^{\frac{1}{2}}  |P_{r}|  \Big)^{2}\Big]^{\frac{1}{2}}\mathbb{E}^{(\lambda)}\Big[\sup_{0\leq r\leq \frac{T}{\lambda} }\big(  \lambda^{\frac{1}{2}}+\lambda^{\frac{1}{2}}|P_{r}|+\lambda^{\frac{5}{2}} |P_{r}|^{2}\big)^{2}  \Big]^{\frac{1}{2}}\nonumber =\mathit{O}(\lambda^{\frac{1}{4}}) .
 \end{align*}
 The first inequality is Cauchy-Schwarz, and the second is Chebyshev's for the first term.   For the second term in the second inequality, Parts (5) and (6) of Prop.~\ref{AMinus} imply that there are $C_{1},C_{1}'>0$ such that  
 $$|p|+\mathcal{A}_{\lambda}^{-}(x,p)\leq |p|+4\lambda |p| \mathcal{E}_{\lambda}(p)+ C_{1}\mathcal{E}_{\lambda}(p)\leq C_{1}'(1+|p|+\lambda^{2}|p|^{2}).         $$
 The expectations on the last line above are finite by Lem.~\ref{FirstEnergyLem} since $|P_{r}|\leq (2 H_{r})^{\frac{1}{2}}$.

 \vspace{.5cm}
  
\noindent (II). \hspace{.25cm}  
The difference  between $\lambda^{\frac{1}{2}} \mathbf{M}_{\frac{t}{\lambda}}$ and $\frak{m}_{t,\epsilon}^{(\lambda)}$  can be bounded by
\begin{align*}
\mathbb{E}^{(\lambda)}\big[ \sup_{0\leq t\leq T}  \big| \lambda^{\frac{1}{2}} \mathbf{M}_{\frac{t}{\lambda} }- \frak{m}_{t,\epsilon}^{(\lambda)} \big|^{2} \big]&\leq 4 \mathbb{E}^{(\lambda)}\big[ \big| \lambda^{\frac{1}{2}}\mathbf{M}_{\frac{T}{\lambda}} -\frak{m}_{T,\epsilon}^{(\lambda)} \big|^{2} \big]= 4\lambda\mathbb{E}^{(\lambda)}\Big[ \Big|  \int_{0}^{\frac{T}{\lambda}}d\mathbf{M}_{r}e^{-\epsilon^{-1}\lambda^{\frac{1}{2}}\mathbf{Q}_{r^{-}}}  \Big|^{2} \Big]\\ & =  4\mathbb{E}^{(\lambda)}\Big[  \lambda\int_{0}^{\frac{T}{\lambda}}dr\mathcal{V}_{\lambda}(S_{r}) e^{-2\epsilon^{-1}\lambda^{\frac{1}{2}}\mathbf{Q}_{r^{-}}}  \Big].
\end{align*}
The first inequality is Doob's, and the second equality uses that $\frac{d}{dt}\langle \mathbf{M}\rangle_{t}=\mathcal{V}_{\lambda}(S_{t})$.  For $ \epsilon \in [\lambda,1]$ the right side is smaller than 
\begin{align*}
\mathbb{E}^{(\lambda)}\Big[ \lambda  \int_{0}^{\frac{T}{\lambda}}dr\mathcal{V}_{\lambda}(S_{r}) e^{-2\epsilon^{-1}\lambda^{\frac{1}{2}}\mathbf{Q}_{r^{-}}}   \Big]& \leq C_{1}\mathbb{E}^{(\lambda)}\big[\mathbf{T}_{\frac{T}{\lambda}}  \big]+ T\sup_{|p|>\epsilon^{\frac{1}{2}}\lambda^{-\frac{1}{2}} }  \mathcal{V}_{\lambda}(x,p) e^{-2^{\frac{3}{2}}\epsilon^{-1} \lambda^{\frac{1}{2}}H^{\frac{1}{2}}(x,p) }\\ &\leq  C_{1} \mathbb{E}^{(\lambda)}\big[\mathbf{T}_{\frac{T}{\lambda}}  \big]+ C_{2}T\sup_{|p|>\epsilon^{\frac{1}{2}}\lambda^{-\frac{1}{2}} } (1+\lambda|p|) e^{-2\epsilon^{-1}\lambda^{\frac{1}{2}}|p|} \\ &\leq C_{1}'\big(\textup{max}(\epsilon , \lambda )\big)^{\frac{1}{2}}+2C_{2}T  e^{-2\epsilon^{-\frac{1}{2} } }=\mathit{O}\big(\textup{max}(\epsilon^{\frac{1}{2}}, \lambda^{\frac{1}{2}} )  \big),
\end{align*}
where $C_{1}:=\sup_{\lambda<1}\sup_{|p|\leq \lambda^{-1} } \mathcal{V}_{\lambda}(x,p) $ and  $\mathbf{T}_{t}=\lambda\int_{0}^{t}dr\chi( H_{r}\leq \epsilon \lambda^{-1}  ) $. The value $C_{1}$ is finite by Part (1) of Prop.~\ref{AMinus}.  The second inequality uses Part (1) of Prop.~\ref{AMinus} again  and  that $|p|\leq 2^{\frac{1}{2}}H^{\frac{1}{2}}(x,p)$ in the exponent.  The  $C_{1}'$ in the third inequality is from Lem.~\ref{LowEnergyLemma}. \vspace{.5cm}

\noindent (III). \hspace{.25cm} We will show that $\frak{m}_{t,\epsilon}^{(\lambda)}$ becomes close  in the norm  $\|\cdot\|_{s} $  to $F_{t}(\lambda^{\frac{1}{2}}\mathbf{Q}_{\frac{\cdot}{\lambda}} )$ as $\lambda\rightarrow 0$   for a  function $F:L^{\infty}([0,T])\rightarrow L^{\infty}([0,T])$ that is continuous with respect to the supremum norm.  The convergence in law of the pair $\big(\lambda^{\frac{1}{2}}\mathbf{Q}_{\frac{t}{\lambda}} ,\, F_{t}(\lambda^{\frac{1}{2}}\mathbf{Q}_{\frac{\cdot}{\lambda}} )\big)$ is then determined be the convergence of the first component. 

For $q\in L^{\infty}([0,T])$ we define $F_{t}(q)$ as
\begin{align}\label{Misery}
F_{t}(q):=q_{t} +\epsilon e^{-\epsilon^{-1}q_{t}} +\frac{1}{2}\int_{0}^{t}dr q_{r}   (1-e^{-\epsilon^{-1}q_{r} })- \frac{1}{2\epsilon}\int_{0}^{t}dr e^{-\epsilon^{-1}q_{r}}. 
\end{align}
$F:L^{\infty}([0,T])$ is Lipschitz continuous with respect the supremum norm for a constant that scales as $\propto \epsilon^{-1}$ for small $\epsilon$.   
Let  $\mathbf{m}_{t,\epsilon}^{(\lambda),\prime}:=F_{t}(\lambda^{\frac{1}{2}}\mathbf{Q}_{\frac{\cdot}{\lambda}})$.  
Notice that since $\frak{p}_{0}=0$
\begin{align*}
F_{t}(|\frak{p}|) &=
 |\frak{p}_{t}|+\epsilon e^{-\frac{|\frak{p}_{t}| }{\epsilon} } + \frac{1}{2} \int_{0}^{t}dr|\frak{p}_{r}| \big(1-e^{-\frac{|\frak{p}_{r}| }{\epsilon} }\big)-\frac{1}{2\epsilon} \int_{0}^{t}dr e^{-\frac{|p_{r}|}{\epsilon} }\\ &= \int_{0}^{t}d\mathbf{\tilde{B}}_{r}\big(1-e^{-\frac{|\frak{p}_{r}| }{\epsilon} }\big)=\mathbf{m}_{t,\epsilon}
\end{align*}
where the second equality is from $d\mathbf{\tilde{B}}_{t}=d|\frak{p}_{t}|+\frac{1}{2}|\frak{p}_{t}|dt-d\frak{l}_{t}$, the chain rule, and that $(d|\frak{p}_{t}|)^{2}=dt$.  
By (i) and the convergence in law of $\lambda^{\frac{1}{2}}P_{\frac{t}{\lambda}}$ to $\frak{p}_{t}$ by~\cite[Thm. 1.3]{Previous}, there is  convergence in law as $\lambda\rightarrow 0$,
\begin{eqnarray*}
\big( \lambda^{\frac{1}{2}}\mathbf{Q}_{\frac{t}{\lambda}},  \mathbf{m}_{t,\epsilon}^{(\lambda),\prime}) \stackrel{\frak{L}}{\Longrightarrow} ( |\frak{p}_{t}|,\mathbf{m}_{t,\epsilon}).   
 \end{eqnarray*}
The remainder of the proof will focus on showing that the difference between $\mathbf{m}_{t,\epsilon}^{(\lambda)}$ and $ \mathbf{m}_{t,\epsilon}^{(\lambda),\prime}$ converges to zero in the norm $\|\cdot\|_{\frak{s}}$ as $\lambda \rightarrow 0$.  More precisely, we show that $\|\mathbf{m}_{t,\epsilon}^{(\lambda)}- \mathbf{m}_{t,\epsilon}^{(\lambda),\prime}\|_{\frak{s}}$ is $\mathit{O}(\lambda^{\frac{1}{8}})$ for small $\lambda$.

By substituting $d\mathbf{M}_{r}= d\mathbf{Q}_{r}-d\mathbf{A}_{r}^{+}+d\mathbf{A}_{r}^{-}       $, the martingale $\mathbf{m}_{t,\epsilon}^{(\lambda)}$ can be written as
$$\frak{m}_{t,\epsilon}^{(\lambda)}= \lambda^{\frac{1}{2}} \int_{0}^{\frac{t}{\lambda} }\big(d\mathbf{Q}_{r}-d\mathbf{A}_{r}^{+}+d\mathbf{A}_{r}^{-}\big)  \big(1-e^{-\epsilon^{-1}\lambda^{\frac{1}{2}}\mathbf{Q}_{r^{-}}}     \big) .  $$  
It is sufficient to show that  
\begin{eqnarray} \label{Plus}
& & -\lambda^{\frac{1}{2}} \int_{0}^{\frac{t}{\lambda} }d\mathbf{A}_{r}^{+} \big(1-e^{-\epsilon^{-1}\lambda^{\frac{1}{2}}\mathbf{Q}_{r^{-}}}     \big) \longrightarrow 0 ,\\ \label{Minus}
& &\lambda^{\frac{1}{2}} \int_{0}^{\frac{t}{\lambda} }d\mathbf{A}_{r}^{-} \big(1-e^{-\epsilon^{-1}\lambda^{\frac{1}{2}}\mathbf{Q}_{r^{-}}}     \big) -\frac{1}{2}\int_{0}^{t}dr \lambda^{\frac{1}{2}}\mathbf{Q}_{\frac{r}{\lambda}}   (1-e^{-\epsilon^{-1}\lambda^{\frac{1}{2}}\mathbf{Q}_{\frac{r}{\lambda}} })\longrightarrow 0,  \\ \label{SREnergy}
& & \lambda^{\frac{1}{2}} \int_{0}^{\frac{t}{\lambda} } d\mathbf{Q}_{r}  \big(1-e^{-\epsilon^{-1}\lambda^{\frac{1}{2}}\mathbf{Q}_{r^{-}}}     \big)- \lambda^{\frac{1}{2}}\mathbf{Q}_{\frac{t}{\lambda}}  -\epsilon e^{-\epsilon^{-1}\lambda^{\frac{1}{2}}\mathbf{Q}_{\frac{t}{\lambda}}}+ \frac{1}{2\epsilon}\int_{0}^{t}dr e^{-\epsilon^{-1}\lambda^{\frac{1}{2}}\mathbf{Q}_{\frac{r}{\lambda}}}\longrightarrow  0,\hspace{1.5cm}
\end{eqnarray}
since the expressions sum up to $\mathbf{m}_{t,\epsilon}^{(\lambda)}- \mathbf{m}_{t,\epsilon}^{(\lambda),\prime}$.

 Since $d\mathbf{A}^{+}_{t}=dt\mathcal{A}^{+}_{\lambda}(X_{t},P_{t})$ the value~(\ref{Plus}) is bounded by
\begin{multline*}
 \mathbb{E}^{(\lambda)}\Big[ \sup_{0\leq t\leq T }  \Big| \lambda^{\frac{1}{2}} \int_{0}^{\frac{t}{\lambda} }d\mathbf{A}_{r}^{+} \big(1-e^{-\epsilon^{-1}\lambda^{\frac{1}{2}}\mathbf{Q}_{r^{-}}}     \big)  \Big|\Big]=  \mathbb{E}^{(\lambda)}\Big[\lambda\int_{0}^{\frac{T}{\lambda} }dr\mathcal{A}_{\lambda}^{+}(X_{r},P_{r})  \big(1-e^{-\epsilon^{-1}\lambda^{\frac{1}{2}}\mathbf{Q}_{r}}   \big)\Big]\\ \leq C\lambda \mathbb{E}^{(\lambda)}\Big[\int_{0}^{\frac{T}{\lambda} }dr \frac{1}{1+|P_{r}|^{2}} \big(1-e^{-\epsilon^{-1}\lambda^{\frac{1}{2}} \mathbf{Q}_{r}}   \big)\Big]\leq  C\mathbb{E}^{(\lambda)}\big[\mathbf{T}_{\frac{T}{\lambda}}  \big]+CT\sup_{|p|>\epsilon^{\frac{1}{2}}\lambda^{-\frac{1}{2}}}\frac{ e^{-\epsilon^{-1}\lambda^{\frac{1}{2}}|p| }  }{1+p^{2}} =\mathit{O}(\epsilon^{\frac{1}{2}}), 
\end{multline*}
where $\mathbf{T}_{t}$ is defined as above.  The first inequality is from Part (2) of Prop.~\ref{AMinus}, and the second inequality is similar to the analysis in Part (I).  For the convergence~(\ref{Minus}),  $d\mathbf{A}^{-}_{t}=dt\mathcal{A}^{-}_{\lambda}(X_{t},P_{t})$ and     
\begin{align*}
\mathbb{E}^{(\lambda)}\Big[ &\sup_{0\leq t\leq T }  \Big| \lambda^{\frac{1}{2}} \int_{0}^{\frac{t}{\lambda} }d\mathbf{A}_{r}^{-} \big(1-e^{-\epsilon^{-1}\lambda^{\frac{1}{2}}\mathbf{Q}_{r^{-}}}     \big) - \lambda^{\frac{1}{2}}\frac{1}{2} \int_{0}^{t }dr \mathbf{Q}_{\frac{r}{\lambda} }  \big(1-e^{-\epsilon^{-1}\lambda^{\frac{1}{2}}\mathbf{Q}_{\frac{r}{\lambda} }}     \big)  \Big|\Big]\\ &\leq \mathbb{E}^{(\lambda)}\Big[ \sup_{0\leq t\leq T }   \int_{0}^{t }dr\big| \lambda^{-\frac{1}{2} }\mathcal{A}^{-}_{\lambda}(X_{\frac{r}{\lambda}},P_{\frac{r}{\lambda}})-\frac{1}{2}\lambda^{\frac{1}{2}}\mathbf{Q}_{\frac{r}{\lambda} }       \big|  \Big]  .   
\end{align*}
By adding and subtracting $\frac{1}{2}\lambda^{\frac{1}{2}} |P_{\frac{r}{\lambda}}|$ in the integrand and applying the triangle inequality, we are left with the terms
$$\Big| \lambda^{-\frac{1}{2} }\mathcal{A}^{-}_{\lambda}(X_{\frac{r}{\lambda}},P_{\frac{r}{\lambda}}) -\frac{1}{2}\lambda^{\frac{1}{2}} |P_{\frac{r}{\lambda}}|\Big|\quad  \text{and} \quad \Big|\frac{1}{2}\lambda^{\frac{1}{2}} |P_{\frac{r}{\lambda}}|  -\frac{1}{2}\lambda^{\frac{1}{2}}\mathbf{Q}_{\frac{r}{\lambda} }       \Big|,$$ which are bounded by the analysis in Part (II) and at the beginning of Part (i), respectively.

The convergence~(\ref{SREnergy}) requires more work.  The terms $\lambda^{\frac{1}{2}} \int_{0}^{\frac{t}{\lambda} } d\mathbf{Q}_{r}$ and $\lambda^{\frac{1}{2}}\mathbf{Q}_{\frac{t}{\lambda}}-\lambda^{\frac{1}{2}}\mathbf{Q}_{0}$  are equal, and $\lambda^{\frac{1}{2}}\mathbf{Q}_{0}$ is small, so we must bound
\begin{align}\label{Forlan}
\mathbb{E}^{(\lambda)}\Big[ \sup_{0\leq t\leq T}\Big| \epsilon e^{-\epsilon^{-1}\lambda^{\frac{1}{2}}\mathbf{Q}_{\frac{t}{\lambda}}}+\lambda^{\frac{1}{2}} \int_{0}^{\frac{t}{\lambda} } d\mathbf{Q}_{r}  e^{-\epsilon^{-1}\lambda^{\frac{1}{2}}\mathbf{Q}_{r^{-}}}- \frac{1}{2\epsilon}\int_{0}^{\frac{t}{\lambda} }dr e^{-\epsilon^{-1}\lambda^{\frac{1}{2}}\mathbf{Q}_{\frac{r}{\lambda}}}\Big|\Big].    
\end{align}
 The difference would be zero by the Ito chain rule if $\lambda^{\frac{1}{2}}\mathbf{Q}_{\frac{t}{\lambda}}$ were replaced by $|\frak{p}_{r}|$, and the norm of the difference is essentially a measure of how close the chain rule is to holding.  We start with a Taylor expansion around each collision time $t_{n}$.  Let $\Delta \mathbf{Q}_{r}=\mathbf{Q}_{r}-\mathbf{Q}_{r^{-}}$, then  $ \epsilon e^{-\epsilon^{-1}\lambda^{\frac{1}{2}}\mathbf{Q}_{\frac{t}{\lambda}}}$ can be written as
\begin{align*}
 \epsilon e^{-\epsilon^{-1}\lambda^{\frac{1}{2}}\mathbf{Q}_{\frac{t}{\lambda}}}-&\epsilon e^{-\epsilon^{-1}\lambda^{\frac{1}{2}}\mathbf{Q}_{0} } =\epsilon \sum_{n=1}^{\mathcal{N}_{\frac{t}{\lambda}}} \Big(e^{-\epsilon^{-1}\lambda^{\frac{1}{2}}\mathbf{Q}_{t_{n}}}-e^{-\epsilon^{-1}\lambda^{\frac{1}{2}}\mathbf{Q}_{t_{n}^{-}} }\Big)  =  -\lambda^{\frac{1}{2}}\sum_{n=1}^{\mathcal{N}_{\frac{t}{\lambda}}}\Delta \mathbf{Q}_{t_{n}} e^{-\epsilon^{-1}\lambda^{\frac{1}{2}}\mathbf{Q}_{t_{n}^{-}} }\\ &+ \frac{\lambda }{2\epsilon}
 \sum_{n=1}^{\mathcal{N}_{\frac{t}{\lambda}}}\big(\Delta \mathbf{Q}_{t_{n}} \big)^{2}e^{-\epsilon^{-1}\lambda^{\frac{1}{2}}\mathbf{Q}_{t_{n}^{-}} }  -\frac{\lambda^{\frac{3}{2}} }{2\epsilon^{2}} \sum_{n=1}^{\mathcal{N}_{\frac{t}{\lambda}}}\int_{0}^{  \Delta \mathbf{Q}_{t_{n}} }dw \big(\Delta \mathbf{Q}_{t_{n}}-w \big)^{2}e^{-\epsilon^{-1}\lambda^{\frac{1}{2}}(\mathbf{Q}_{t_{n}^{-}}+w) } \\ =&\lambda^{\frac{1}{2}} \int_{0}^{\frac{t}{\lambda} } d\mathbf{Q}_{r}  e^{-\epsilon^{-1}\lambda^{\frac{1}{2}}\mathbf{Q}_{r^{-}}}+ \frac{\lambda}{2\epsilon}\int_{0}^{\frac{t}{\lambda}}(d\mathbf{Q}_{r})^{2} e^{-\epsilon^{-1}\lambda^{\frac{1}{2}}\mathbf{Q}_{r^{-}}}+  \mathbf{R}_{\lambda,\epsilon,t},
\end{align*}
where $\mathcal{N}_{t}$ is the number of collisions up to time $t$, and $\mathbf{R}_{\lambda,\epsilon,t}$  denotes the third term between the two equalities.  By the triangle inequality, the expectation (\ref{Forlan}) is smaller than  
\begin{align}\label{Largesse}
\epsilon+
\mathbb{E}^{(\lambda)}\big[ \sup_{0\leq t\leq \frac{T}{\lambda} }\big|   \mathbf{R}_{\lambda,\epsilon,t}\big|\big]+\mathbb{E}^{(\lambda)}\Big[ \sup_{0\leq t\leq \frac{T}{\lambda} }\Big| \frac{\lambda}{2\epsilon}\int_{0}^{t}\big(dr-(d\mathbf{Q}_{r})^{2}\big)  e^{-\epsilon^{-1}\lambda^{\frac{1}{2}}\mathbf{Q}_{r^{-}}}\Big|\Big],
\end{align}
where $\epsilon\in \R^{+}$ bounds  $\mathbb{E}^{(\lambda)}\big[\epsilon e^{-\epsilon^{-1}\lambda^{\frac{1}{2}}\mathbf{Q}_{0} }\big]$.

To bound the remainder term $\mathbf{R}_{\lambda,\epsilon,t}$ in~(\ref{Largesse}), we may write
\begin{align*}
\mathbb{E}^{(\lambda)}\big[ \sup_{0\leq t\leq \frac{T}{\lambda} }\big|   \mathbf{R}_{\lambda,\epsilon,t}\big|\big]&\leq \frac{\lambda^{\frac{3}{2}} }{6\epsilon^{2}} \mathbb{E}^{(\lambda)}\Big[  \sum_{n=1}^{\mathcal{N}_{\frac{T}{\lambda}}} \big|\Delta \mathbf{Q}_{t_{n}}\big|^{3}\Big]   =  \frac{\lambda^{\frac{3}{2}} }{6\epsilon^{2}} \mathbb{E}^{(\lambda)}\Big[  \int_{0}^{\frac{T}{\lambda}}dr\mathcal{V}_{\lambda,\frac{3}{2}}(X_{r},P_{r})      \Big] \\ & \leq C_{1}\frac{\lambda^{\frac{3}{2}} }{6\epsilon^{2}} \mathbb{E}^{(\lambda)}\Big[  \int_{0}^{\frac{T}{\lambda}}dr\big(1+\lambda \mathbf{Q}_{r}\big)^{4} \Big]  \leq C_{1}'T\frac{\lambda^{\frac{1}{2}} }{\epsilon^{2}}=\mathit{O}(\lambda^{\frac{1}{2}}),
\end{align*}
where the first inequality is by Part (1) of Prop.~\ref{AMinus}, and the $C_{1}'>0$ in the second inequality exists by bounding the moments of $\mathbf{Q}_{r}=(2H_{r})^{\frac{1}{2}}$ over $0\leq r\leq \frac{T}{\lambda}$ using Lem.~\ref{FirstEnergyLem}.

By adding and subtracting  $\int_{0}^{t}dr\mathcal{V}_{\lambda}(X_{r},P_{r})   $ in the expression for the last term in~(\ref{Largesse}) and using the triangle inequality,
\begin{align*}
\mathbb{E}^{(\lambda)}\Big[& \sup_{0\leq t\leq \frac{T}{\lambda} }\Big| \frac{\lambda}{2\epsilon}\int_{0}^{t}\big(dr-(d\mathbf{Q}_{r})^{2}\big)  e^{-\epsilon^{-1}\lambda^{\frac{1}{2}}\mathbf{Q}_{r^{-}}}\Big|\Big] \leq \mathbb{E}^{(\lambda)}\Big[ \frac{\lambda}{2\epsilon}\int_{0}^{\frac{T}{\lambda} }dr\big| 1-\mathcal{V}_{\lambda}(X_{r},P_{r})    \big| \Big]\\ &+\mathbb{E}^{(\lambda)}\Big[ \sup_{0\leq t\leq \frac{T}{\lambda} }\Big| \frac{\lambda}{2\epsilon}\int_{0}^{t}\big(dr\mathcal{V}_{\lambda}(X_{r},P_{r}) -(d\mathbf{Q}_{r})^{2}\big)  e^{-\epsilon^{-1}\lambda^{\frac{1}{2}}\mathbf{Q}_{r^{-}}}\Big|\Big].
\end{align*}
The first term on the right side is smaller than 
\begin{align}\label{LondonBurning}
\mathbb{E}^{(\lambda)}\Big[ \frac{\lambda}{2\epsilon}\int_{0}^{\frac{T}{\lambda} }dr\big| 1-\mathcal{V}_{\lambda}(X_{r},P_{r})    \big| \Big]\nonumber \leq & C_{1}\frac{1}{\epsilon}\mathbb{P}^{(\lambda)}\big[ \mathbf{T}_{\frac{T}{\lambda}}'    \big]+C_{2}\frac{\lambda^{\frac{1}{2}} }{\epsilon}\\ &+C_{3}\frac{\lambda}{\epsilon}\mathbb{E}^{(\lambda)}\Big[ \int_{0}^{\frac{T}{\lambda} }dr\chi\big(\mathbf{Q}_{r}\geq \lambda^{-\frac{3}{4}} \big)\big(1+\lambda \mathbf{Q}_{r}\big)^{3}  \Big]
\end{align}
for some $C_{1},C_{2},C_{3}>0$, where $\mathbf{T}_{t}':=\lambda\int_{0}^{t}dr\chi\big( \mathbf{Q}_{r}\leq \lambda^{-\frac{3}{8}}   \big)$, and the three terms on the right correspond to the parts of the trajectory such that $\mathbf{Q}_{r}\leq \lambda^{-\frac{3}{8}}$, $\lambda^{-\frac{3}{8}} \leq \mathbf{Q}_{r}\leq \lambda^{-\frac{3}{4}}$,  and $ \lambda^{-\frac{3}{4}}\leq \mathbf{Q}_{r}$.  For the first and third terms on the right side of~(\ref{LondonBurning}), we have applied Part (1) of Prop.~\ref{AMinus}.  For the second term, we applied Part (3) of Prop.~\ref{AMinus}.   The first term is $\mathit{O}(\lambda^{\frac{1}{8} })$ by Lem.~\ref{LowEnergyLemma}.  For the last term on the right side of~(\ref{LondonBurning}), we can apply Cauchy-Schwarz and an analogous argument 
to that at the end of Part (I).  

Moreover, the expression $\int_{0}^{t}\big(dr\mathcal{V}_{\lambda}(X_{r},P_{r}) -(d\mathbf{Q}_{r})^{2}\big)  \,e^{-\epsilon^{-1}\lambda^{\frac{1}{2}}\mathbf{Q}_{r^{-}}}$ is a martingale with predictable quadratic variation
$$\int_{0}^{t}dr\mathcal{V}_{\lambda,2}(X_{r},P_{r})e^{-2\epsilon^{-1}\lambda^{\frac{1}{2}}\mathbf{Q}_{r}}. $$
Hence, by Doob's maximal inequality
\begin{align*}
\mathbb{E}^{(\lambda)}&\Big[ \sup_{0\leq t\leq \frac{T}{\lambda} }\Big| \frac{\lambda}{2\epsilon}\int_{0}^{t}\big(dr\mathcal{V}_{\lambda}(X_{r},P_{r}) -(d\mathbf{Q}_{r})^{2}\big)  \,e^{-\epsilon^{-1}\lambda^{\frac{1}{2}}\mathbf{Q}_{r^{-}}}\Big|^{2}\Big]^{\frac{1}{2}}& \\ &\leq  \frac{\lambda}{\epsilon}\mathbb{E}^{(\lambda)}\Big[ \int_{0}^{\frac{T}{\lambda} }dr\mathcal{V}_{\lambda,2}(X_{r},P_{r})e^{-2\epsilon^{-1}\lambda^{\frac{1}{2}}\mathbf{Q}_{r}}\Big]^{\frac{1}{2}}&\\  & \leq  C_{1}\frac{\lambda}{\epsilon}\mathbb{E}^{(\lambda)}\Big[ \int_{0}^{\frac{T}{\lambda} }dr(1+\lambda\mathbf{Q}_{r})^{5}     \Big]^{\frac{1}{2}}\leq C_{1}'\frac{T\lambda^{\frac{1}{2}} }{\epsilon}.
\end{align*}
The second inequality holds for some $C_{1}>0$ by Part (1) of Prop.~\ref{AMinus} and $|p|\leq 2^{\frac{1}{2}}H^{\frac{1}{2}}(x,p)$. Lem.~\ref{FirstEnergyLem} yields the third inequality for some $C_{1}'>0$.

\end{proof}

\section{The martingale problem}\label{SecMartProb}

In the lemma below, we consider the class of process pairs  $(\frak{p},\frak{m})\in \R^{2}$ such that the first component is an Ornstein-Uhlenbeck process and the second component is a continuous martingale.  With the additional criterion that $\langle \frak{m}\rangle$ is the local time of the process $\frak{p}$ at zero, Lem.~\ref{MartProblem} states that the law for the pair $(\frak{p},\frak{m})$ is determined uniquely as  $(\frak{p},\mathbf{B}_{\frak{l}})$, where $\mathbf{B}$ is a standard Brownian motion independent of $\frak{p}$.  For the process inverse $\frak{s}$ of $\frak{l}$,  we can immediately observe that the process $\mathbf{B}_{t}:=\frak{m}_{\frak{s}_{t}}$ is a Brown motion since it is a continuous martingale with quadratic variation $t$.  Thus the question concerns the independence of $\mathbf{B}$ from $\frak{p}$.  Lemma~\ref{MartProblem}  is a formulation of the \textit{martingale problem} in the sense of~\cite{Jacod}.  For example, a standard Brownian motion is the unique continuous martingale $\frak{m}$ satisfying that $\frak{m}_{t}^{2}-t$ is a martingale.  Our criterion could be formulated analogously by demanding that
$$  \frak{m}_{t}^{2}-\frak{l}_{t}   $$
is a martingale.   The proof of the lemma makes  use of the fact that $\frak{l}$ almost surely makes all of its movement on a set of times having measure zero.  If we only needed to show that $\big(\frak{l},\frak{m}\big)$ with the condition above necessarily has the law of $(\frak{l},\mathbf{B}_{\frak{l}})$ for $\mathbf{B}$ independent of $\frak{l}$, then we could apply the argument in~\cite[Thm. 4.21]{Hopfner} since $\frak{l}$ is the process inverse of the one-sided Levy process $\frak{s}$.   However, $\frak{p}$ contains information that $\frak{l}$ does not so there is the logical possibility that  $\frak{p}$ and $\mathbf{B}$ are still dependent.

\begin{lemma}\label{MartProblem}
Consider a process $(\frak{p},\frak{m})\in \R^{2}$ and let $\mathbb{F}_{t}$ be the filtration generated by it. Let $\frak{p}$ be a copy of the  Ornstein-Uhlenbeck process satisfying the Markov property with respect to $\mathbb{F}_{t}$ and $\frak{l}$ be the local time of $\frak{p}$ at zero.  Moreover, let $\frak{m}$ be continuous, a martingale with respect to $\mathbb{F}_{t}$, and have predictable quadratic variation satisfying $\langle \frak{m}\rangle=\frak{l}$.  It follows that   
$(\frak{p},\frak{m})$ is equal in law to $(\frak{p},\mathbf{B}_{\frak{l}})$, where $\mathbf{B}$ is a standard Brownian motion independent of $\frak{p}$.

\end{lemma}

\begin{proof}

By definition the process $\frak{p}$ satisfies the Langevin equation $d\frak{p}_{t}=-\frac{1}{2}\frak{p}_{t}dt+ d\mathbf{B}_{t}'   $ for a standard Brownian motion $\mathbf{B}'$.  Since $\frak{p}$ satisfies the Markov property with respect $\mathbb{F}_{t}$, the Brownian motion $\mathbf{B}'$ must also.   We denote the right-continuous process inverse of $\frak{l}$ by  $\frak{s}$.  The time-changed martingale $\mathbf{B}_{t}=\frak{m}_{\frak{s}_{t}}$ is continuous and has quadratic variation $\langle \mathbf{B}\rangle_{t}=t$, and is thus a copy of Brownian motion.   We will construct a family of processes $\frak{p}^{(\epsilon)}$ such that

\begin{enumerate}[(I).]
\item $\frak{p}^{(\epsilon)}$ is independent of $\mathbf{B}$ for each $\epsilon>0$.

\item As $\epsilon\rightarrow 0$, $\mathbb{E}\big[ \sup_{0\leq t\leq T} \big|\frak{p}_{t}^{(\epsilon)}-\frak{p}_{t}\big|        \big]=\mathit{O}(\epsilon^{\frac{1}{2}-\delta })$ for any $\delta>0$.

\end{enumerate}

The above statements imply that the processes $\mathbf{B}$ and $\frak{p}$ are independent.  Since $\frak{l}$ is the process inverse of $\frak{s}$,  $\frak{m}_{t}=\mathbf{B}_{\frak{l}_{t}}$.  Thus (I) and (II) imply the result.\vspace{.5cm}

\noindent(I). \hspace{.1cm} First, we give definitions that are  prerequisite 
to defining $\frak{p}^{(\epsilon)}$.  If $|\frak{p}_{0}|< \epsilon$ let  the stopping times  $\varsigma_{n},\varsigma_{n}^{\prime}$ be defined such that 
 $\varsigma_{0}=\varsigma_{0}'=\varsigma^{\prime}_{1}=0$ and
\begin{align*}
\varsigma^{\prime}_{n}= \min \{r\in (\varsigma_{n-1},\infty)\,\big| \,|\frak{p}_{r}|\leq \frac{1}{2}\epsilon       \},\quad \quad
\varsigma_{n}= \min \{ r\in (\varsigma^{\prime}_{n},\infty)\,\big| \, |\frak{p}_{r}|\geq  \epsilon           \}.
\end{align*}
Also let  $\mathbf{n}_{t}$ be the number of $\varsigma_{n}$ up to time $t$. If $|\frak{p}_{0}|\geq \epsilon $ then we use the same recursive definition with $\varsigma_{0}=\varsigma^{\prime}_{0}=0$.  The intervals  $[\varsigma_{n}',\varsigma_{n})$, $n\geq 0$   and  $[\varsigma_{n},\varsigma_{n+1}')$, $n\geq 1$  will be referred to as the incursions and excursions respectively.   Let $\tau_{t}$ be the hitting time that 
$$ t= \tau_{t} -\varsigma_{\mathbf{n}_{\tau_{t}}}+ \sum_{n=0}^{\mathbf{n}_{\tau_{t}}-1}\varsigma_{n+1}'-\varsigma_{n}.   $$
In other terms, $\tau_{t}$ is the first time that the total excursion time sums up to $t$.

   Define another copy of Brownian motion 
$\mathbf{B}^{(\epsilon)}$
$$\mathbf{B}_{t}^{(\epsilon)}= \mathbf{B}_{\tau_{t}}'-\mathbf{B}_{\varsigma_{\mathbf{n}_{\tau_{t} }}}' + \sum_{n=0}^{\mathbf{n}_{\tau_{t} }-1}\mathbf{B}_{\varsigma_{n+1}'}'-\mathbf{B}_{\varsigma_{n}}'. $$
 Define $\frak{p}^{(\epsilon)}$ and $\tilde{\frak{p}}^{(\epsilon)}$  to be the solutions of the Langevin equations
 \begin{align*}
  d\frak{p}_{t}^{(\epsilon)}=&-\frac{1}{2}\frak{p}_{t}^{(\epsilon)}dt+d \mathbf{B}_{t}^{(\epsilon)},  \\    d\tilde{\frak{p}}_{t}^{(\epsilon)}=&\chi\big(t\in \cup_{n=0}^{\infty}[\varsigma_{n},\,\varsigma_{n+1}'     ]  \big)\big(-\frac{1}{2}\tilde{\frak{p}}_{t}^{(\epsilon)}dt+  d \mathbf{B}_{t}'\big),  
  \end{align*}
with $\frak{p}_{0}^{(\epsilon)}=\tilde{\frak{p}}_{0}^{(\epsilon)}=\frak{p}_{0}$.  We will use the process $\tilde{\frak{p}}^{(\epsilon)}$ as an intermediary between $\frak{p}^{(\epsilon)}$ and $\frak{p}$ in (II).   

We claim that our construction makes the Brownian motion $ \mathbf{B}^{(\epsilon)}$ independent of $\mathbf{B}$ and thus $\frak{p}^{(\epsilon)}$ is also independent of $\mathbf{B}$.  Construct the stopping time $\gamma_{t}$ and the martingale  $\frak{m}^{(\epsilon)} $  such that 

$$ t= \gamma_{t} -\varsigma_{\mathbf{n}_{\gamma_{t}}}+ \sum_{n=1}^{\mathbf{n}_{\gamma_{t}}-1}\varsigma_{n}-\varsigma_{n}' \quad\text{and}\quad \frak{m}_{t}^{(\epsilon)}= \frak{m}_{\gamma_{t}}-\frak{m}_{\varsigma_{\mathbf{n}_{\gamma_{t} }}} + \sum_{n=1}^{\mathbf{n}_{\gamma_{t} }-1}\frak{m}_{\varsigma_{n}}-\frak{m}_{\varsigma_{n}'}.  $$
Analogously to $\tau_{t}$  the above means that $\gamma_{t}$ is the first time that the duration of all the incursions sums up to $t$.  The martingale $\frak{m}^{(\epsilon)}  $ is a time-change of $\frak{m} $ with  $\frak{m}_{\gamma_{t}}=\frak{m}_{t}^{(\epsilon)}$ in which a portion of the pauses during which $\langle m\rangle=\frak{l}$ remains constant have been cut out.  Since only pauses have been cut out,  $\sigma( \frak{m}^{(\epsilon)})$  contains all of the information regarding $\mathbf{B}$.   However, the $\sigma$-algebras $\sigma(\mathbf{B}^{(\epsilon)})$ and  $\sigma( \frak{m}^{(\epsilon)})$ are independent.  This follows since $\sigma(\mathbf{B}^{(\epsilon)})$ has no information about the incursions--including their durations, and vice versa for $\sigma(\frak{m}^{(\epsilon)})$.

\vspace{.5cm}

\noindent (II).\hspace{.2cm} By the triangle inequality,
\begin{align}\label{Asfalto}
\mathbb{E}\Big[ \sup_{0\leq t\leq T} \big|\frak{p}_{t}^{(\epsilon)}-\frak{p}_{t}\big|        \Big]\leq \mathbb{E}\Big[\sup_{0\leq t\leq T} \big| \frak{p}_{t}^{(\epsilon)}-\tilde{\frak{p}}_{t}^{(\epsilon)}    \big|\Big]+\mathbb{E}\Big[\sup_{0\leq t\leq T} \big|\tilde{\frak{p}}_{t}^{(\epsilon)}-\frak{p}_{t}   \big|   \Big].
\end{align}
  We bound the first and second terms on the right side of~(\ref{Asfalto}) in (i) and (ii) below.   First we show that $\mathbb{E}\big[ \tau_{T}-T   \big] =\mathit{O}(\epsilon) $, which is used in both parts.   A Riemann over-sum using that $4n\geq 2(n+1)$ for $n\geq 1 $ gives the first inequality below.
 \begin{align}
\mathbb{E}\big[ \tau_{T}-T   \big] \leq &  \nonumber \mathbb{E}\big[\tau_{T}\wedge (2T)-T    \big]+4T\sum_{n=1}^{\infty}\mathbb{P}\big[ \tau_{T}\geq 2nT \big] \nonumber    \\
   \leq & \nonumber \mathbb{E}\big[\tau_{T}\wedge (2T)-T    \big]+4T\sum_{n=1}^{\infty} \Big(\sup_{q\in \R}\,\mathbb{P}_{q}\big[ \tau_{T}\geq 2T \big]\Big)^{n}\nonumber \\
 = &   \mathbb{E}\big[\tau_{T}\wedge (2T)-T    \big]+4T \frac{ \mathbb{P}_{0}\big[ \tau_{T}>2T \big] }{1-   \mathbb{P}_{0}\big[ \tau_{T}>2T \big]  }=\mathit{O}(\epsilon) \label{Brenen}.
  \end{align}
In order for the event $\tau_{T}>2nT $ to occur, the random walker must fail to accumulate a duration $T$  of excursion time over $n$ disjoint intervals of length $2T$.  Thus $ \mathbb{P}\big[ \tau_{T}\geq 2nT \big] \leq \big(\sup_{q\in \R}\mathbb{P}_{q}\big[ \tau_{T}\geq 2T \big]\big)^{n}$, as we have used in the second inequality.  The equality in~(\ref{Brenen}) is from summing the geometric series,  and since  $ \mathbb{P}_{q}\big[ \tau_{T}\geq 2T \big]$ is maximized for $q=0$.  The starting point $q=0$ maximizes the probability that $\tau_{T}$ is large (e.g. $\geq 2T$) because the process must travel the furthest to attain a value $|\frak{p}_{t}|\geq \epsilon$ in which the excursion clock may begin to run.

 To show the order equality~(\ref{Brenen}), we  show that $ \mathbb{P}_{0}\big[ \tau_{T}>2T \big]$ and  $\mathbb{E}\big[\tau_{T}\wedge (2T)-T    \big]$ are $\mathit{O}(\epsilon)$.  We first note that 
\begin{align*}
\mathbb{P}_{0}\big[\tau_{T}\geq 2T     \big] & \leq \mathbb{P}_{0}\Big[\int_{0}^{2T}dr\chi\big(|\frak{p}_{r}|\leq \epsilon\big)\geq T   \Big]\\ &\leq \frac{1}{T}\mathbb{E}_{0}\Big[\int_{0}^{2T}dr\chi\big(|\frak{p}_{r}| \leq \epsilon\big)   \Big]=\frac{1}{T}\int_{0}^{2T}dt\int_{[-\epsilon, \epsilon]}dq \frac{ e^{-\frac{ q^{2}   }{2\omega_{t} } }  }{ (2\pi \omega_{t} )^{\frac{1}{2}}   }=\mathit{O}(\epsilon),
\end{align*}
where $\omega_{t}=1-e^{-\frac{1}{2}t}$.   The first inequality uses that the event  $ \tau_{T}\geq 2T$ implies the event $\int_{0}^{2T}dr\chi\big(|\frak{p}_{r}|\leq \epsilon\big)\geq T $ since the incursions have $|\frak{p}_{r}|\leq \epsilon$. The second inequality is Chebyshev's, and the first equality uses that the density $\frac{ e^{-\frac{ q^{2}   }{2\omega_{t} } }  }{ (2\pi \omega_{t} )^{\frac{1}{2}}   }  $ is the explicit solution to Ornstein-Uhlenbeck forward equation, i.e., Kramer's equation, starting from zero.  The other term is similar: 
\begin{align*}
\mathbb{E}\big[\tau_{T}\wedge (2T)-T    \big] & \leq\mathbb{E}\Big[\int_{0}^{2T}dr\chi\big(|\frak{p}_{r}| \leq \epsilon\big)   \Big]\\ &\leq \mathbb{E}_{0}\Big[\int_{0}^{2T}dr\chi\big(|\frak{p}_{r}| \leq \epsilon\big)   \Big]
=\int_{0}^{2T}dt\int_{[-\epsilon, \epsilon]}dq \frac{ e^{-\frac{ q^{2}   }{2\omega_{t} } }  }{ (2\pi \omega_{t} )^{\frac{1}{2}}   }=\mathit{O}(\epsilon).
\end{align*}

  \vspace{.5cm}

\noindent (i).\hspace{.2cm} Notice that $ \frak{p}^{(\epsilon)}$ is a stochastic time-change of    $\tilde{\frak{p}}^{(\epsilon)}   $ with $\frak{p}_{t}^{(\epsilon)}=\tilde{\frak{p}}_{\tau_{t}}^{(\epsilon)} $.  Thus the first term on right side of~(\ref{Asfalto}) is smaller than
\begin{align}
\mathbb{E}\Big[&\sup_{0\leq t\leq T} \big| \frak{p}_{t}^{(\epsilon)}-\tilde{\frak{p}}_{t}^{(\epsilon)}    \big|\Big]\leq 
\mathbb{E}\Big[ \sup_{ \substack{ 0\leq r\leq \tau_{T}-T  \\ 0\leq t\leq T   }   }\big| \frak{p}_{t+r}^{(\epsilon)}-\frak{p}_{t}^{(\epsilon)}  \big|  \Big] \nonumber \\
 &= \mathbb{E}\Big[\mathbb{E}\Big[ \sup_{ \substack{ 0\leq r\leq \tau_{T}-T  \\  0\leq t\leq T   }   }\big| \frak{p}_{t+r}^{(\epsilon)}-\frak{p}_{t}^{(\epsilon)}  \big|      \,\Big|\,\tau_{T}-T \Big]\Big]=\mathbb{E}\Big[\delta_{\tau_{T}-T}(v) \mathbb{E}\Big[ \sup_{ \substack{ 0\leq r\leq v  \\ 0\leq t\leq T   }   }\big| \frak{p}_{t+r}^{(\epsilon)}-\frak{p}_{t}^{(\epsilon)}  \big|     \Big]\Big]
\nonumber \\ &\leq \mathbb{E}\Big[ (1-e^{-\frac{1}{2}(\tau_{T}-T)  })\sup_{  0\leq t\leq \tau_{T}  }\big| \frak{p}_{t}^{(\epsilon)} \big|     \Big]+\mathbb{E}\Big[ \delta_{\tau_{T}-T}(v) \mathbb{E}\Big[ \sup_{ \substack{ 0\leq r\leq v  \\ 0\leq t\leq T   }   }\Big|  \int_{t}^{t+r}d\mathbf{B}_{s}^{(\epsilon)}e^{-\frac{1}{2}(t+r-s)}  \Big|     \Big]\Big] \label{Minkowski}
\end{align}
 The second equality follows from the independence of the process $ \frak{p}^{(\epsilon)} $ and the difference $\tau_{T}-T$.  
For the last inequality, we have used the triangle inequality with the explicit form in the first equality below:
\begin{align}
\frak{p}_{t+r}^{(\epsilon)}-\frak{p}_{t}^{(\epsilon)}&=(e^{-\frac{1}{2}r}-1)\frak{p}_{t}^{(\epsilon)}+ \int_{t}^{t+r}d\mathbf{B}_{s}^{(\epsilon)}e^{-\frac{1}{2}(r+t-s)} \nonumber \nonumber \\ &=(e^{-\frac{1}{2}r}-1)\frak{p}_{t}^{(\epsilon)}+\mathbf{B}_{t+r}^{(\epsilon)}-\mathbf{B}_{t}^{(\epsilon)} -\frac{1}{2}\int_{t}^{t+r}ds \big(\mathbf{B}_{s+t}^{(\epsilon)}-\mathbf{B}_{t}^{(\epsilon)}\big)e^{-\frac{1}{2}(r+t-s)} \label{Tray}.
\end{align}
The second equality is Ito's product rule.  Note that for $m\geq 1$
\begin{align}\label{CheapBnd}
\mathbb{E}\Big[  \sup_{ 0\leq v\leq r   }\Big|  \int_{t}^{t+v}d\mathbf{B}_{s}^{(\epsilon)}e^{-\frac{1}{2}(t+r-s)}  \Big|^{2m}  \Big]\leq &  2^{m}\mathbb{E}\Big[  \sup_{ 0\leq v\leq r   }\Big| \mathbf{B}_{t+v}^{(\epsilon)}-\mathbf{B}_{t}^{(\epsilon)}  \Big|^{2m}  \Big] \nonumber \\ \leq & 2^{m}\big(\frac{2m}{2m-1}\big)^{2m}\mathbb{E}\Big[ \Big| \mathbf{B}_{t+r}^{(\epsilon)}-\mathbf{B}_{t}^{(\epsilon)}  \Big|^{2m}  \Big]\leq m! \big(\frac{4m}{2m-1}\big)^{2m}r^{m}.
\end{align}
The first inequality comes from rewriting $\int_{t}^{t+v}d\mathbf{B}_{s}^{(\epsilon)}e^{-\frac{1}{2}(t+r-s)}$ as in~(\ref{Tray}), applying the triangle inequality, and using that $\int_{t}^{t+r}ds e^{-\frac{1}{2} (t+r-s)}\leq 2$. The second inequality is Doob's, and the last is a computation of the Gaussian moment.

For the first term on the right side of~(\ref{Minkowski}), we have following routine inequalities using that $\mathbb{E}[ \tau_{T}-T]\leq C\epsilon$ for some $C>0$:
\begin{align*}
\mathbb{E}\Big[ (1-e^{-\frac{1}{2}(\tau_{T}-T)  })\sup_{  0\leq t\leq \tau_{T}  }\big| \frak{p}_{t}^{(\epsilon)}  \big|     \Big] & \leq \mathbb{E}\Big[ \big(1-e^{-\frac{1}{2}(\tau_{T}-T)  }\big)^{2}\Big]^{\frac{1}{2}} \mathbb{E}\Big[  \sup_{  0\leq t\leq \tau_{T}  }\big| \frak{p}_{t}^{(\epsilon)}  \big|^{2}     \Big]^{\frac{1}{2}} \\ &\leq \mathbb{E}\big[ \big(\tau_{T}-T)\wedge 1 \big]^{\frac{1}{2}} \mathbb{E}\Big[ \sup_{  0\leq t\leq \tau_{T}   }\big| \frak{p}_{t}^{(\epsilon)}  \big|^{2}   \Big]^{\frac{1}{2}} \\ &  \leq C\epsilon^{\frac{1}{2}} \mathbb{E}\big[ \big| \frak{p}_{0}  \big|^{2}     \big]^{\frac{1}{2}}+C\epsilon^{\frac{1}{2}} \mathbb{E}\Big[ \sup_{ 0\leq t\leq \tau_{T}   }\Big|  \int_{0}^{t}d\mathbf{B}_{r}^{(\epsilon)}e^{-\frac{1}{2}(t-r)}  \Big|^{2}     \Big]^{\frac{1}{2}} \\ &\leq C\epsilon^{\frac{1}{2}} \mathbb{E}\big[ \big| \frak{p}_{0}  \big|^{2}     \big]^{\frac{1}{2}}+C\epsilon^{\frac{1}{2}}2\mathbb{E}\big[\tau_{T} \big]^{\frac{1}{2}}=\mathit{O}(\epsilon^{\frac{1}{2}}).
\end{align*}
The last inequality follows from the independence of $\tau_{T}$ and the Brownian motion $\mathbf{B}^{(\epsilon)} $ and   ~(\ref{CheapBnd}).

Now we bound  the second  term on the right side of~(\ref{Minkowski}).   We have the following relations for $m\geq 1$:
\begin{align*}
\mathbb{E}\Big[ \sup_{ \substack{ 0\leq r\leq v  \\ 0\leq t\leq T   }   }\Big|  \int_{0}^{r}d\mathbf{B}_{t+s}^{(\epsilon)}e^{-\frac{1}{2}(r-s)}  \Big|     \Big] & =\mathbb{E}\Big[ \sup_{ \substack{ 0\leq r\leq v  \\ 0\leq z+r\leq T+v   }   }\Big|  \int_{z}^{z+r}d\mathbf{B}_{s}^{(\epsilon)}e^{-\frac{1}{2}(z+r-s) }  \Big|     \Big] \\& \leq 2 \mathbb{E}\Big[ \sup_{0\leq n \leq \lfloor \frac{T+v}{v}\rfloor } \sup_{  0\leq r\leq v     }\Big|  \int_{nv}^{nv+r}d\mathbf{B}_{s}^{(\epsilon)}e^{-\frac{1}{2}(nv+r-s )}  \Big|        \Big]\\ &\leq 2\mathbb{E}\Big[ \sum_{n=0}^{\lfloor \frac{T+v}{v} \rfloor }  \sup_{  0\leq r\leq v    }\Big|  \int_{z}^{z+r}d\mathbf{B}_{s}^{(\epsilon)}e^{-\frac{1}{2}(z+r-s) }  \Big|^{2m}     \Big]^{\frac{1}{2m}}\\  &= 2\left\lfloor \frac{T+v}{v} \right\rfloor^{\frac{1}{2m}}\mathbb{E}\Big[   \sup_{  0\leq r\leq v    }\Big|  \int_{0}^{r}d\mathbf{B}_{s}^{(\epsilon)}e^{-\frac{1}{2}(r-s) }  \Big|^{2m}     \Big]^{\frac{1}{2m}} \\ &\leq 2(m!)^{\frac{1}{2m}}\frac{4m}{2m-1}   \left\lfloor \frac{T+v}{v} \right\rfloor^{\frac{1}{2m}} v^{\frac{1}{2}}< 8 m^{\frac{1}{2}} |T+v|^{\frac{1}{2m}}|v|^{\frac{m-1}{2m}}.
\end{align*}  
The  second inequality is $ (\sup_{n} a_{n})^{2m}\leq \sum_{n}a_{n}^{2m}$ followed by Jensen's inequality,  the second equality is from the stationarity of the increments for $\mathbf{B}^{(\epsilon)}$, and the third inequality is from~(\ref{CheapBnd}).  With the above 
\begin{align*}
\mathbb{E}\Big[ \delta_{\tau_{T}-T}(v) \mathbb{E}\Big[ \sup_{ \substack{ 0\leq r\leq v  \\ 0\leq t\leq T   }   }\Big|  \int_{t}^{t+r}d\mathbf{B}_{s}^{(\epsilon)}e^{-\frac{1}{2}(t+r-s)}  \Big|     \Big]\Big] & \leq 8m^{\frac{1}{2}}\mathbb{E}\Big[ |\tau_{T}|^{\frac{1}{2m}}|\tau_{T}-T|^{\frac{m-1}{2m}}\Big] \\ &\leq  8m^{\frac{1}{2}}\mathbb{E}\big[ \tau_{T}^{\frac{1}{m+1}}\big]^{\frac{m+1}{2m}}\mathbb{E}\big[  \tau_{T}-T\big]^{\frac{m-1}{2m}}=\mathit{O}(\epsilon^{\frac{m-1}{2m}}),
\end{align*}
where the second inequality is Holder's.  The value $m$ can be picked to make the power of $\epsilon>0$ arbitrarily close to $\frac{1}{2}$.

\vspace{.5cm}

\noindent (ii). \hspace{.2cm} Notice that $\frak{p}$ and $\tilde{\frak{p}}^{(\epsilon)}$  satisfy the equations
\begin{align}
\frak{p}_{t}=&e^{-\frac{1}{2}t}\frak{p}_{0}+\int_{0}^{t}d\mathbf{B}_{r}'e^{-\frac{1}{2}(t-r)} \label{Orig}\\
\tilde{\frak{p}}_{t}^{(\epsilon)}=&e^{-\frac{1}{2}t}\frak{p}_{0}+\int_{0}^{t}d\mathbf{B}_{r}'\chi_{r}^{(\epsilon)}e^{-\frac{1}{2}(t-r)} +\frac{1}{2}\int_{0}^{t}dr\tilde{\frak{p}}_{r}^{(\epsilon)}e^{-\frac{1}{2}(t-r)}(1-\chi_{r}^{(\epsilon)}) \label{Epsified} ,
\end{align}
where $\chi_{r}^{(\epsilon)}=\chi\big(r\in \cup_{n=0}^{\infty}[\varsigma_{n},\,\varsigma_{n+1}'     ]  \big)$.  The Ito product rule for the martingale $\int_{0}^{t}d\mathbf{B}_{r}'\big(1-\chi_{r}^{(\epsilon)}\big)$ gives  
\begin{align}\label{Arab}
\int_{0}^{t}d\mathbf{B}_{r}'\big(1-\chi_{r}^{(\epsilon)}\big)e^{-\frac{1}{2}(t-r)}= \int_{0}^{t}d\mathbf{B}_{r}'\big(1-\chi_{r}^{(\epsilon)}\big)-\frac{1}{2}\int_{0}^{t}dr e^{-\frac{1}{2}(t-r)}\int_{0}^{r}d\mathbf{B}_{s}'\big(1-\chi_{s}^{(\epsilon)}\big).
\end{align}
Similarly to~(\ref{CheapBnd}),  
\begin{align}
\mathbb{E}\Big[&  \sup_{ 0\leq t\leq T  }\Big|  \int_{0}^{t}d\mathbf{B}_{r}'\big(1-\chi_{r}^{(\epsilon)}\big)e^{-\frac{1}{2}(t-r)}  \Big|^{2}  \Big]\leq  4\mathbb{E}\Big[  \sup_{ 0\leq t\leq T }\Big|\int_{0}^{t}d\mathbf{B}_{t}'\big(1-\chi_{t}^{(\epsilon)}\big) \Big|^{2}  \Big] \nonumber \\  & \leq 16\mathbb{E}\Big[ \Big|\int_{0}^{T}d\mathbf{B}_{t}'\big(1-\chi_{t}^{(\epsilon)}\big)    \Big|^{2}  \Big]=16\mathbb{E}\Big[ \int_{0}^{T}dt\big(1-\chi_{t}^{(\epsilon)}\big)      \Big]\leq 16\mathbb{E}\Big[ \int_{0}^{T}dt\chi\big(|\frak{p}_{t}|<\epsilon\big)     \Big] \nonumber \\ &\leq 16\mathbb{E}_{0}\Big[ \int_{0}^{T}dt\chi\big(|\frak{p}_{t}|<\epsilon\big)     \Big]=\int_{0}^{T}dt\int_{[-\epsilon, \epsilon]}dq\frac{ e^{-\frac{ q^{2}   }{2\omega_{t} } }  }{ (2\pi \omega_{t} )^{\frac{1}{2}}   }=\mathit{O}(\epsilon).\label{Punjabi}
\end{align}
The first inequality is from~(\ref{Arab}) with the triangle inequality, and the second inequality is Doob's.  The fourth inequality uses that the initial value $\frak{p}_{0}=0$ will maximize the expectation of the quantity $\int_{0}^{T}dt\chi_{t}\big(|\frak{p}_{t}|<\epsilon\big)$.

Using~(\ref{Orig}) and~(\ref{Epsified}) with the triangle inequality, we have the first inequality below:
\begin{align}
 \mathbb{E}\big[& \sup_{0\leq t\leq T} \big| \tilde{\frak{p}}_{t}^{(\epsilon)}-\frak{p}_{t}   \big|  \Big] \nonumber \\ &\leq \mathbb{E}\Big[ \sup_{0\leq t\leq T} \Big| \int_{0}^{t}d\mathbf{B}_{r}'e^{-\frac{1}{2}(t-r)}(1-\chi_{r}^{(\epsilon)} )  \Big|  \Big] + \mathbb{E}\Big[ \sup_{0\leq t\leq T} \Big| \int_{0}^{t}dr\tilde{\frak{p}}_{r}^{(\epsilon)}e^{-\frac{1}{2}(t-r)}(1-\chi_{r}^{(\epsilon)})  \Big|  \Big] \nonumber  \\  &\leq \mathit{O}(\epsilon^{\frac{1}{2}})+  \mathbb{E}\Big[ \sup_{0\leq t\leq T} \big|\tilde{\frak{p}}_{t}^{(\epsilon)}\big|^{2}\Big]^{\frac{1}{2}}  \mathbb{E}\Big[ \Big(\int_{0}^{T}dt\big(1-\chi_{t}^{(\epsilon)}\big) \Big)^{2}  \Big]^{\frac{1}{2}} \nonumber  \\ &\leq \mathit{O}(\epsilon^{\frac{1}{2}})+  T^{\frac{1}{2}}\mathbb{E}\Big[ \sup_{0\leq t\leq T} \big|\frak{p}_{r}\big|^{2}\Big]^{\frac{1}{2}}  \mathbb{E}\Big[ \int_{0}^{T}dt\big(1-\chi_{t}^{(\epsilon)}\big)  \Big]^{\frac{1}{2}} =\mathit{O}(\epsilon^{\frac{1}{2}}).\label{Fellon}
\end{align}
The second inequality uses~(\ref{Punjabi}) for the first term and  Holder's equality twice  for the second term.  The second inequality follows from the fact that $ \tilde{\frak{p}}^{(\epsilon)}_{\tau_{t}}$ has the same law as $\frak{p}_{t}$ and  $\tau_{t}\geq t$.  In other words, $\frak{p}$ has the same law as a sped-up version of   $ \tilde{\frak{p}}^{(\epsilon)}$.   Finally,  $\mathbb{E}\Big[ \int_{0}^{T}dt\big(1-\chi_{t}^{(\epsilon)}\big)  \Big] =\mathit{O}(\epsilon)$ by~(\ref{Punjabi}).

\end{proof}

\section{Nummelin splitting}\label{SecNumSplit}

We will now summarize the particular splitting structure defined in~\cite[Sect.  2]{Previous}, which extends the state space of the process.  This construction is contained in the first two components of the split process introduced in~\cite{Loch}.   The resulting process behaves nearly as though the state space contains a recurrent atom.  This has the advantage that the life cycles between returns to the ``atom" are nearly uncorrelated.  To do this we first introduce a resolvent chain embedded in the original process.  We then split the chain using the standard technique~\cite{Athreya,Nummelin}, and we extend the resolvent chain to a non-Markovian process that contains an embedded version of the original process.

Let $e_{m}$, $m\in \mathbb{N}$ be mean one exponential random variables that  are independent of each other and of the process $S_{t}=(X_{t},P_{t})\in \Sigma$.  Define $\tau_{n}:=\sum_{m=1}^{n}e_{m}$, and  by convention, we set $\tau_{0}=0$.  The $\tau_{n}$ will be referred to as the \textit{partition times}.  Define $\mathbf{N}_{t}$ to be the number of non-zero $\tau_{n}$ less than $t$, and the Markov chain $\sigma_{n}:=(X_{\tau_{n}},P_{\tau_{n}})\in \Sigma$,  
 which is referred to as the \textit{resolvent chain}.  The  transition kernel $\mathcal{T}_{\lambda}$ for the resolvent chain, which acts on functions from the left and on  measures from the right, has the form   
$$\mathcal{T}_{\lambda} = \int_{0}^{\infty}dr e^{-r+r\mathcal{L}_{\lambda}}=(1-\mathcal{L}_{\lambda} )^{-1},       $$
 where $\mathcal{L}_{\lambda}$ is the backward Markov generator for the process.   The resolvent chain has the same invariant probability density~(\ref{MaxBolt}) as the original process.   By Nummelin splitting, which we outline presently,  the state space  $\Sigma$ is extended to   $\tilde{\Sigma}=\Sigma\times \{0,1\}$  in order to construct a chain $(\tilde{\sigma}_{n})\in \tilde{\Sigma}$ with a \textit{recurrent atom} and such that the statistics for $(\sigma_{n})$ are embedded in the first component of $(\tilde{\sigma}_{n})$.  For a Markov chain, an \textit{atom} is a subset of the state space such that the transition measure is independent of the element within the subset.  The atom is \textit{recurrent} if the event of returning to the atom  in the future has probability one.

  A probability measure $\nu$ on $\Sigma$ paired with a non-zero function $h:\Sigma\rightarrow [0,1] $ are said to satisfy the \textit{minorization condition} with respect to $\mathcal{T}_{\lambda}$  if
\begin{align}\label{NummelinCrit}
 \mathcal{T}_{\lambda}(s_{1},ds_{2})\geq h(s_{1}) \nu(ds_{2}).   
 \end{align}
By Part (1) of \cite[Prop. 2.3]{Previous}, there exists a $\mathbf{u}>0$ such that
\begin{align}\label{Philobuster}
 h(s)= \mathbf{u}  \frac{ \chi\big(H(s)\leq l   \big)   }{ U } \hspace{1cm} \text{and} \hspace{1cm}  \nu(ds)=   ds\frac{ \chi\big( H(s)\leq l   \big)   }{U},   \end{align}
satisfy the minorization condition,
where  $l=1+2\sup_{x}V(x)$ and $U>0$ is the normalization constant of $\nu$.  The specific choice of $h$ and $\nu$ satisfying~(\ref{NummelinCrit}) is not important in this section although we will take them to be defined as in~(\ref{Philobuster}) for future sections.

     We define the following forward transition operator $\tilde{\mathcal{T}}_{\lambda}$, which sends  the state $(s_{1},z_{1})\in \tilde{\Sigma}$ to the infinitesimal region $(ds_{2}, z_{2})$ with measure:
\begin{align*}
\tilde{\mathcal{T}}_{\lambda}( s_{1},z_{1};ds_{2},z_{2})=\left\{  \begin{array}{ccc} \frac{1-h(s_{2}) }{1-h(s_{1})}  \big( \mathcal{T}_{\lambda}-  h \otimes \nu\big) (s_{1},ds_{2})            & \hspace{.5cm} & z_{1}=z_{2}=0 , \\  \frac{h(s_{2}) }{1-h(s_{1})}  \big( \mathcal{T}_{\lambda}-  h\otimes \nu\big) (s_{1},ds_{2})         &  \hspace{.5cm}& z_{1}=1-z_{2}=0 ,   \\ \big(1- h(s_{2})\big)  \nu (ds_{2})             & \hspace{.5cm} & z_{1}=1- z_{2}=1 , \\ h(s_{2}) \nu(ds_{2})    & \hspace{.5cm}& z_{1}=z_{2}=1. 
    \end{array} \right.  
\end{align*}
Given a measure $\mu$ on $\Sigma$, we refer to its \textit{splitting} $\tilde{\mu}$ as the measure on $\tilde{\Sigma}$ given by 
\begin{align}\label{MeasureSplitting}
\tilde{\mu}(ds,z)=  \chi(z=0)\big(1-h(s)\big)\mu(ds)+\chi(z=1)h(s)\mu(ds). 
\end{align}
In particular, the split chain is taken to have initial distribution given by the splitting of the initial distribution for the original (pre-split) chain.    The invariant measure for the chain $(\tilde{\sigma}_{n})$  is the splitting $\tilde{\Psi}_{\infty,\lambda}$   of the Maxwell-Boltzmann distribution   defined in~(\ref{MaxBolt}).  The split chain is positive-recurrent for any $\lambda>0$ since the original process is positive-recurrent and, in fact, exponentially ergodic to $\Psi_{\infty,\lambda}$.   The jump rates from $(s_{1},1)$ are independent of $s_{1}\in \Sigma$, and thus the set $\Sigma\times 1 \subset \tilde{\Sigma} $ is an atom.  The atom is recurrent since the original chain is positive-recurrent with stationary state $\Psi_{\infty,\lambda}$  and  $\tilde{\Psi}_{\infty,\lambda}(\Sigma\times 1)=\Psi_{\infty,\lambda}(h)>0$.    Notice that according to the above transition rates, the probability that $z_{2}=1$ is $h(s_{2})$ when given $s_{1}$, $s_{2}$, and $z_{1}$.

  Using the law for the split chain $\tilde{\sigma}_{n}\in \tilde{\Sigma}$ determined by the transition rates $\tilde{\mathcal{T}}_{\lambda}$ above, we may construct a split process $(\tilde{S}_{t})\in \tilde{\Sigma}$ and a sequence of times $\tilde{\tau}_{n}$ with the recipe below.  The $\tilde{\tau}_{n}$ should be thought of as the partition times $\tau_{n}$ embedded in the split statistics although we temporarily denote them with the tilde to emphasize their axiomatic role in the construction of the split process.  Let $\tilde{\tau}_{n}$ and $\tilde{S}_{t}=(S_{t},Z_{t})$ be such that 
\begin{enumerate}
\item  $0=\tilde{\tau}_{0}$, $\tilde{\tau}_{n}\leq \tilde{\tau}_{n+1}$, and $\tilde{\tau}_{n}\rightarrow \infty$ almost surely.  

\item The chain $(\tilde{S}_{\tilde{\tau}_{n}})$ has the same law as $(\tilde{\sigma}_{n})$.  

\item  For $t\in [\tilde{\tau}_{n},\tilde{\tau}_{n+1})$, then $Z_{t}=Z_{\tilde{\tau}_{n}}$.

\item  Conditioned on the information known up to time $\tilde{\tau}_{n}$ for $\tilde{S}_{t}$, $t\in [0,\tilde{\tau}_{n}]$ and $\tilde{\tau}_{m}$, $m\leq n$, and also the value $\tilde{S}_{\tilde{\tau}_{n+1} }$, the law for the trajectories $S_{t}$, $t\in[\tilde{\tau}_{n},\tilde{\tau}_{n+1}]$ (which includes the length $\tilde{\tau}_{n+1}-\tilde{\tau}_{n}$) agrees with the law for the original dynamics conditioned on knowing the values $S_{\tilde{\tau}_{n}}$ and $S_{\tilde{\tau}_{n+1}}$.   
\end{enumerate}
The marginal distribution for the first component $S_{t}$ agrees with the original process and the times $\tilde{\tau}_{n}$ are independent mean one  exponential random variables which are independent of $S_{t}$. Of course the times $\tilde{\tau}_{n}$ are not independent of the process $\tilde{S}_{t}$, and we emphasize that the increment $ \tilde{\tau}_{n+1}-\tilde{\tau}_{n}$ is not necessarily exponential given the state $\tilde{S}_{\tilde{\tau}_{n} }$.  The process $\tilde{S}_{t}$ is not Markovian due to the conditioning in (4), although the chain $(\tilde{S}_{\tilde{\tau}_{n}})$ is  Markovian.   By~\cite{Loch} we can construct a Markov process by including an extra component to the process: the triple $(S_{t}, Z_{t}, S_{\tau(t)})\in \Sigma\times \{0,1\}\times \Sigma$ is Markovian, where $\tau(t)$ is the first partition time to occur after time $t$.  We refer to the statistics of the split process by $\tilde{\mathbb{E}}^{(\lambda)}$ and $\tilde{\mathbb{P}}^{(\lambda)}$ for expectations and probabilities, respectively.  We will neglect the tilde from the symbol $\tilde{\tau}_{n}$ for the remainder of the text.

Note that once we have defined the split process $\tilde{S}_{t}$, we can proceed to define the life cycles.  Let $R_{m}'$, $m\geq 1$ be the time $\tau_{\tilde{n}_{m}}$ for $\tilde{n}_{m}=\textup{min} \{ n\in \mathbb{N}\,\big|\,\sum_{r=0}^{n}\chi(Z_{\tau_{r}}= 1)  =m  \}   $.   The random variable $R_{m}'$ is the $m$th partition time corresponding to a visit of the atom set $\Sigma\times 1$, and we set $R_{0}'=0$ by convention.  We define $R_{m}$ to be the partition time following $R_{m}'$.  The $m$th life cycle is the time interval $[R_{m},R_{m+1})$.  
Successive life cycle trajectories over $[R_{n-1},R_{n})$ and $[R_{n},R_{n+1})$ are obviously not independent since  a.s. $S_{R_{n}^{-} }=S_{R_{n} }$.  However, non-successive life cycles are pairwise independent.   When considering the random variables $\int_{R_{n}}^{R_{n+1}}dr\frac{dV}{dx}(X_{r})$, the correlations between successive terms can be removed by adding and subtracting certain resolvent terms as seen in the summand in the lemma below.      

Let  $\tilde{N}_{t}$ be the number of $R_{n}'$ to have occurred up to time $t$.  Define $\tilde{\mathcal{F}}_{t}'$ to be the filtration containing all information for the partition times $\tau_{n}$ and the split process $\tilde{S}_{t}$ before time  $R_{n+1}$ where $t\in [R_{n}',R_{n+1}')$.  Also define $\frak{R}^{(\lambda)}  $ as the reduced resolvent of the backward generator $\mathcal{L}_{\lambda}$ corresponding to the master equation~(\ref{TheModel}).  The reduced resolvent formally satisfies $\frak{R}^{(\lambda)}=\int_{0}^{\infty}dr e^{r\mathcal{L}_{\lambda}}$ on elements $g\in L^{\infty}(\Sigma)$ with $\Psi_{\infty,\lambda}(g)=0$.  Notice that the martingale defined in the lemma below resembles~(\ref{MartConst}).

\begin{lemma}\label{LemKeyMart}
Let the process  $\tilde{M}_{t}$ be defined as
$$ \tilde{M}_{t}:= \sum_{n=1}^{ \tilde{N}_{t} }\Big(\int_{R_{n}}^{R_{n+1}}dr\frac{dV}{dx}(X_{r})-\big(\frak{R}^{(\lambda)}\frac{dV}{dx}\big)( S_{R_{n}} ) +\big(\frak{R}^{(\lambda)}\,\frac{dV}{dx}\big)( S_{R_{n+1}} ) \Big). $$

 The process  $\tilde{M}_{t}$ is a martingale with respect to the filtration $\tilde{\mathcal{F}}_{t}'$.    Moreover, the predictable quadratic variation $\langle \tilde{M}\rangle_{t}$ has the form 
$$
   \langle\tilde{M}\rangle_{t}= \sum_{n=1}^{\tilde{N}_{t}}\check{\upsilon}_{\lambda}\big(S_{R_{n}}   \big),
$$
where $\check{\upsilon}_{\lambda}:\Sigma\rightarrow \R^{+}$ is defined as
  \begin{align*} \check{\upsilon}_{\lambda}\big(s  \big):= & 2\tilde{\mathbb{E}}^{(\lambda)}_{\tilde{\delta}_{s} }\Big[\int_{0}^{R_{1}}dr\frac{dV}{dx}(X_{r})\big(\frak{R}^{(\lambda)}\,\frac{dV}{dx}\big)(S_{r} ) \Big]\\ &+\int_{\Sigma}d\nu (s')\Big(\big(\frak{R}^{(\lambda)}\,\frac{dV}{dx}\big)( s' )\Big)^{2}
  -\Big( \big(\frak{R}^{(\lambda)}\,\frac{dV}{dx}\big)( s )    \Big)^{2}.
 \end{align*}  
In the above, $\tilde{\delta}_{s}$ is the splitting of the $\delta$-distribution at $s\in \Sigma$; see~(\ref{MeasureSplitting}). 

\end{lemma}

\section{Proof of Thm.~\ref{ThmMain}}\label{SecMainProof}

Let us define (or recall)  the following notations:
\begin{eqnarray*}
&\tilde{S}_{t}=(S_{t},Z_{t})       & \text{State of the split process at time $t\in \R^{+}$      }\\
&  \tau_{m}\in \R^{+}     &   \text{$m$th partition time} \\
 &\mathbf{N}_{t}\in \mathbb{N} \text{\,}    & \text{Number of non-zero partition times up to time $t\in \R^{+}$ }\\
& R_{m} \in \R^{+}         &  \text{Beginning time of the $m$th life cycle}\\
&  \tilde{N}_{t} \in \mathbb{N}         &    \text{Number of returns to the atom up to time $t\in \R^{+}$} \\
&\mathcal{F}_{t}    &    \text{Information up to time $t\in \R^{+}$ for the original process $S_{r}$ and the $\tau_{m}$ } \\ 
&\tilde{\mathcal{F}}_{t}    &    \text{Information up to time $t\in \R^{+}$ for the split process  $\tilde{S}_{t}$ and the $\tau_{m}$ }\\
&\tilde{\mathcal{F}}_{t}'    &    \text{Information for  $\tilde{S}_{t}$ and the $\tau_{m}$  before time $R_{n+1}$, where $R_{n}'\leq t<R_{n+1}'$   }
\end{eqnarray*}

 Let the constant $\mathbf{u}\in \R^{+}$, the function $h:\Sigma\rightarrow [0,1]$, and measure $\tilde{\nu}$ on $\tilde{\Sigma}$ be defined as in Sect.~\ref{SecNumSplit}. Define  $\upsilon_{\lambda}>0$  as
\begin{align*}
\upsilon_{\lambda}:=& 2\tilde{\mathbb{E}}^{(\lambda)}_{\tilde{\nu}}\Big[\int_{0}^{R_{1}}dr\frac{dV}{dx}(X_{r})\int_{r}^{R_{2}}dr'\frac{dV}{dx}(X_{r'} )\Big]\\
=&\frac{ 2 \int_{\Sigma}dxdp e^{-\lambda H(x,p)} \frac{dV}{dx}(x )\big(\frak{R}^{(\lambda)}\,\frac{dV}{dx}\big)(x,p)  }{ \int_{\Sigma}dxdp e^{-\lambda H(x,p)}     h(x,p) }
,  
\end{align*}
where the equality holds by~\cite[Prop. 2.4]{Previous}. Notice that $\upsilon_{\lambda}$ is  formally equal to $\frac{\kappa}{\mathbf{u}}$ for $\lambda=0$ since the numerator is the formal Green-Kubo expression~(\ref{DiffConst}) and the denominator is $\mathbf{u}=\int_{\Sigma}ds h(s)$. The value $\upsilon_{0}>0$ is well-defined by Lem.~\ref{LemUpsilon}, which is from~\cite[Lem. 5.2]{Previous}.  Thus we can give a rigorous definition for the diffusion constant $\kappa\in \R^{+}$ as $ \kappa:= \mathbf{u}\, \upsilon_{0}$.
\begin{lemma}\label{LemUpsilon}
The value $\upsilon_{\lambda}\in \R^{+}$ is uniformly bounded for $\lambda<1$, and $\upsilon_{\lambda}$ depends continuously on the parameter $\lambda$.
\end{lemma}

The following proposition is from~\cite[Prop. 2.5]{Previous} and~\cite[Lem. 5.3]{Previous}.  The martingale $\tilde{M}_{t}$ was defined in Lem.~\ref{LemKeyMart}. 
  
\begin{proposition}\label{TrivialMart} \text{  }
\begin{enumerate}
\item For the split statistics,  $ \tilde{N}_{t} - \sum_{n=1}^{\mathbf{N}_{t}}h(S_{\tau_{n}})$ is a martingale with respect to the filtration $\tilde{\mathcal{F}}_{t}$.   For the original statistics, 
$\sum_{n=1}^{\mathbf{N}_{t}}h(S_{\tau_{n}}) - \int_{0}^{t}dr h(S_{r})$ is a martingale with respect to  $\mathcal{F}_{t}$.  
In particular,
$$\tilde{\mathbb{E}}^{(\lambda)}\big[ \tilde{N}_{t}    \big]=\mathbb{E}^{(\lambda)}\Big[ \int_{0}^{t}dr h(S_{r})    \Big].  $$

\item As $\lambda\rightarrow 0$, the following asymptotics hold:
$$ \tilde{\mathbb{E}}^{(\lambda)}\Big[\sup_{0\leq t\leq T}\Big|\lambda^{\frac{1}{2}} \langle \tilde{M}\rangle_{\frac{t}{\lambda}}    -\lambda^{\frac{1}{2} }\upsilon_{\lambda} \tilde{N}_{\frac{t}{\lambda}}\Big|      \Big]=\mathit{O}(\lambda^{\frac{1}{4}}).     $$
Also, for any $t\geq 0$, the expectations are equal
 $\tilde{\mathbb{E}}^{(\lambda)}\big[\langle \tilde{M}\rangle_{t}  \big] =\upsilon_{\lambda}\tilde{\mathbb{E}}^{(\lambda)}\big[ \tilde{N}_{t}\big]$.
\end{enumerate}

\end{proposition}

The equality in Prop.~\ref{BasicsOfNum} is from~\cite[Prop. 2.3]{Previous} and is of a standard type for splitting constructions~\cite{Nummelin}.  It states that the probability of the process being at the atom at time $r\in \R^{+}$, conditioned on $r$ being a partition time (i.e. $\mathbf{N}_{r}=\mathbf{N}_{r^{-}}+1$) and the entire past  $\tilde{\mathcal{F}}_{r^{-}}$, is given by the value $h(S_{r})$.  Note that the value $S_{r}$ is a.s. contained in $\tilde{\mathcal{F}}_{r^{-}}$ since a collision will a.s. not occur at the partition time $r\in \R^{+}$ and thus $\lim_{v\nearrow r}S_{v}=S_{r}$.

\begin{proposition}\label{BasicsOfNum} \text{  } 
  $$\tilde{\mathbb{P}}^{(\lambda)}\big[ Z_{r }= 1 \,\big| \,\tilde{\mathcal{F}}_{r^{-}},\,\mathbf{N}_{r}-\mathbf{N}_{r^{-}}=1 \big]=h(S_{r})$$
\end{proposition}

Our proof of Thm.~\ref{ThmMain} takes some inspiration from~the proof of~\cite[Thm. 4.12]{Hopfner} and relies heavily on~\cite{Jacod}. \vspace{.5cm}

\noindent[Proof of Thm.~\ref{ThmMain}]\vspace{.4cm}

For the study of the pair $(\lambda^{\frac{1}{2}}P_{\frac{\cdot}{\lambda}},\lambda^{\frac{1}{4}} D_{\frac{\cdot}{\lambda} }) $, we will begin by embedding the processes in the split statistics defined in Sect.~\ref{SecNumSplit}.  Let the martingale $\tilde{M}$ be defined as in Lem.~\ref{LemKeyMart}. In this proof, all convergences in law refer to the Skorokhod metric.  The following points hold regarding the processes $\lambda^{\frac{1}{4}} D_{\frac{\cdot}{\lambda} }$ and $\lambda^{\frac{1}{4}}\tilde{M}_{\frac{\cdot}{\lambda} }   $:  
\begin{enumerate}[(I).]
\item   As $\lambda\rightarrow 0$
$$\tilde{\mathbb{E}}^{(\lambda)}\Big[\sup_{0\leq t\leq T}\Big| \lambda^{\frac{1}{4}} D_{\frac{t}{\lambda} }-\lambda^{\frac{1}{4}}\tilde{M}_{\frac{t}{\lambda} }    \Big|     \Big]\longrightarrow 0. $$

\item As $\lambda\rightarrow 0$ the bracket process $\langle \tilde{M} \rangle_{t}$ satisfies
$$ \tilde{ \mathbb{E}}^{(\lambda)}\Big[\sup_{0\leq t\leq T}\Big| \lambda^{\frac{1}{2}} \langle \tilde{M} \rangle_{\frac{t}{\lambda} } - \kappa \lambda^{\frac{1}{2}}L_{\frac{t}{\lambda}}    \Big|\Big]\longrightarrow 0, $$
where  $L_{t}=   \mathbf{u}^{-1}\int_{0}^{t}dr h(X_{r},P_{r})$.
\item The martingale $\lambda^{\frac{1}{4}}\tilde{M}_{\frac{t}{\lambda} }   $ satisfies the Lindeberg condition
$$ 
 \sup_{0<\lambda \leq 1}\tilde{\mathbb{P}}^{(\lambda)}\Big[ \sup_{1\leq r\leq \tilde{N}_{\frac{T}{\lambda}}} \Big| \tilde{M}_{r }-   \tilde{M}_{r^{-} } \Big|^{2}  > \frac{\epsilon}{\lambda}  \Big]\longrightarrow 0,   \hspace{1cm} \text{as}  \hspace{1cm} \epsilon\rightarrow 0.   $$
\end{enumerate}
   Statements (I) and (III) have already been shown in the proof of
   \cite[Thm. 1.2]{Previous}. 

We will temporarily assume statement (II) and proceed with the main part of the proof.   By (I) we may work with the pair  $\big(\lambda^{\frac{1}{2}}P_{\frac{\cdot}{\lambda}} ,\lambda^{\frac{1}{4}} \tilde{M}_{\frac{\cdot}{\lambda} }   \big) $  rather than $\big(\lambda^{\frac{1}{2}}P_{\frac{\cdot}{\lambda}} ,\lambda^{\frac{1}{4}} D_{\frac{\cdot}{\lambda} }   \big) $.  By Thm.~\ref{ThmLocalTime} and (II), there is convergence in law as $\lambda\rightarrow 0$ 
\begin{align}\label{Jesus}
 \hspace{1cm} \big(\lambda^{\frac{1}{2}}P_{\frac{t}{\lambda}} ,\lambda^{\frac{1}{2}}\langle \tilde{M}\rangle_{\frac{t}{\lambda}}   \big) \stackrel{\frak{L}}{\Longrightarrow} (\frak{p}_{t},\kappa\frak{l}_{t} ). 
\end{align}
  It follows that the components $\lambda^{\frac{1}{2}}P_{\frac{\cdot}{\lambda}}$ and $\lambda^{\frac{1}{2}}\langle \tilde{M}\rangle_{\frac{\cdot}{\lambda}} $ are $C$-tight for $\lambda<1$.  By~\cite[Thm. VI.4.13]{Jacod} the family of martingales $\lambda^{\frac{1}{4}} \tilde{M}_{\frac{\cdot}{\lambda} }    $ must also be tight for $ \lambda < 1$.  The Lindeberg condition (III) and~\cite[Prop. VI.3.26]{Jacod} guarantee that the family of martingales must be $C$-tight.

 The  triple $T^{(\lambda)}=\big(\lambda^{\frac{1}{2}}P_{\frac{\cdot}{\lambda}} ,\lambda^{\frac{1}{2}} \langle\tilde{M}\rangle_{\frac{\cdot}{\lambda} }    ,\lambda^{\frac{1}{4}} \tilde{M}_{\frac{\cdot}{\lambda} }   \big)$ is $C$-tight  for $\lambda<1$ by~\cite[Cor. VI.3.33]{Jacod} since all of the components are $C$-tight.  By tightness, we may consider a subsequence $\lambda_{n}\rightarrow 0$ such that  $T^{(\lambda_{n})}$  converges in law to a limit $(\frak{p}, \frak{v},\frak{m})$.  The first two components  $\frak{p}$, $\frak{v}$ are respectively the Ornstein-Uhlenbeck process and $\kappa$ multiplied its the local time, i.e., $\frak{v}=\kappa \frak{l}$,  by~(\ref{Jesus}).  We will argue that the third component $\frak{m}$ must be a continuous martingale with respect to the filtration $\sigma(\frak{p}_{r},\frak{m}_{r};\, 0\leq r\leq t)$  such that $\langle \frak{ m}\rangle= \kappa \frak{l}$. The continuity of $\frak{m}$ follows by the $C$-tightness of $\lambda^{\frac{1}{4}} \tilde{M}_{\frac{\cdot}{\lambda} }$.  
  The process $\frak{m}$ is a martingale with respect to $\sigma(\frak{p}_{r},\frak{m}_{r};\, 0\leq r\leq t)$ by~\cite[Prop. IX.1.17]{Jacod}  since $\big(\lambda_{n}^{\frac{1}{2}}P_{\frac{\cdot}{\lambda_{n}}} ,\lambda_{n}^{\frac{1}{4}} \tilde{M}_{\frac{\cdot}{\lambda_{n}} }   \big)$ is adapted to the filtration  $\tilde{\mathcal{F}}_{t}^{(\lambda_{n})}:= \tilde{\mathcal{F}}_{\frac{t}{\lambda_{n}}}'$, the process $\lambda_{n}^{\frac{1}{4}}\tilde{M}_{\frac{\cdot}{\lambda_{n}}}$ is a martingale with respect to $\tilde{\mathcal{F}}_{t}^{(\lambda_{n})}$ by Lem.~\ref{LemKeyMart}, and the family of random variables $\lambda^{\frac{1}{4}}\tilde{M}_{\frac{t}{\lambda}}$ for $\lambda<1$ and $t\in[0,T]$ is uniformly square integrable.  To see the uniform square integrability notice
 \begin{align}\label{Bore} \sup_{0\leq t\leq T} \tilde{\mathbb{E}}^{(\lambda)}\big[\big(  \lambda^{\frac{1}{4}}\tilde{M}_{\frac{t}{\lambda}} \big)^{2}    \big]=  \tilde{\mathbb{E}}^{(\lambda)}\big[\lambda^{\frac{1}{2}}\langle \tilde{M}\rangle_{\frac{T}{\lambda}}   \big]=  \upsilon_{\lambda}\tilde{\mathbb{E}}^{(\lambda)}\big[\lambda^{\frac{1}{2}}\tilde{N}_{\frac{T}{\lambda}}   \big] = \upsilon_{\lambda}\mathbb{E}^{(\lambda)}\Big[\lambda^{\frac{1}{2}}\int_{0}^{\frac{T}{\lambda}}dr h(S_{r})  \Big]. 
 \end{align}
The second and third equalities are by Part (2) and Part (1) of   Prop.~\ref{TrivialMart}, respectively.
   The right side of~(\ref{Bore}) is uniformly bounded for $\lambda<1$ by Thm.~\ref{ThmLocalTime}, and thus    $\sup_{t\in[0,T]} \sup_{\lambda<1}\tilde{\mathbb{E}}^{(\lambda)}\big[\big(  \lambda^{\frac{1}{4}}\tilde{M}_{\frac{t}{\lambda}} \big)^{2}    \big]$ is finite.  By~\cite[Cor. VI.6.7]{Jacod}  the convergence $\lambda_{n}^{\frac{1}{4}} \tilde{M}_{\frac{\cdot}{\lambda_{n}} }\stackrel{\frak{L}}{\Longrightarrow} \frak{m} $ with the Lindeberg condition~(III)     implies the joint convergence of the pair
    $$\big( \lambda_{n}^{\frac{1}{2}}\langle \tilde{M}\rangle_{\frac{t}{\lambda_{n}}}, \lambda_{n}^{\frac{1}{4}} \tilde{M}_{\frac{t}{\lambda_{n}}}\big) \stackrel{\frak{L}}{\Longrightarrow} (\langle\frak{m}\rangle_{t},\frak{m}_{t}).$$ 
For the above, we have used that the difference between  $\lambda_{n}^{\frac{1}{2}}[\tilde{M}]_{\frac{t}{\lambda_{n}}}$ and $\lambda_{n}^{\frac{1}{2}}\langle \tilde{M}\rangle_{\frac{t}{\lambda_{n}}}$ is $\mathit{O}(\lambda_{n}^{\frac{1}{4}})$.  Thus $\langle \frak{m}\rangle=\kappa \frak{l}$.


We have now learned what we could from the martingale $\tilde{M}$.  By (I) we have shown that  $\big(\lambda_{n}^{\frac{1}{2}}P_{\frac{\cdot}{\lambda_{n}}} ,\lambda_{n}^{\frac{1}{4}} D_{\frac{\cdot}{\lambda_{n}} }   \big) $, interpreted with respect to the original statistics, converges in law to  a pair $(\frak{p},\frak{m})$ as $n\rightarrow \infty$, where  $\frak{m}$ is a continuous martingale with respect to the filtration $\sigma(\frak{p}_{r},\frak{m}_{r};\, 0\leq r\leq t)   $ and  $\langle \frak{m}\rangle=\kappa\frak{l}$.  If we establish that $\frak{p}$ satisfies the Markov property with respect to the filtration $\sigma(\frak{p}_{r},\frak{m}_{r};\, 0\leq r\leq t )$, then Lem.~\ref{MartProblem} states that the pair $(\frak{p},\frak{m})$ must have the law of the process $(\frak{p},\sqrt{\kappa} \mathbf{B}_{\frak{l}})$ for a copy of Brownian motion $\mathbf{B}$ independent of $\frak{p}$. Since the pair  $\big(\lambda^{\frac{1}{2}}P_{\frac{\cdot}{\lambda}} ,\lambda^{\frac{1}{4}} D_{\frac{\cdot}{\lambda} }   \big) $ is tight for $\lambda<1$,  establishing the law  $(\frak{p}, \sqrt{\kappa}\mathbf{B}_{\frak{l}})$ as the unique possible subsequential limit would imply  the convergence in law of $\big(\lambda^{\frac{1}{2}}P_{\frac{\cdot}{\lambda}} ,\lambda^{\frac{1}{4}} D_{\frac{\cdot}{\lambda} }   \big) $  as $\lambda\rightarrow 0$ to the process $(\frak{p},\mathbf{B}_{\frak{l}})$.  

To show that $\frak{p}$ satisfies the Markov property with respect to the  filtration  $\sigma(\frak{p}_{r},\frak{m}_{r}; 0\leq r\leq t)$, it is enough to show that the trajectory $\frak{p}_{s}$, $s>t$ is independent of $\sigma(\frak{m}_{r};\, 0\leq r\leq t)$ when given $\sigma(\frak{p}_{r};\, 0\leq r\leq t)$ since the process $\frak{p}$ satisfies the Markov property with respect to its own filtration.  The triple $\big(\lambda_{n}^{\frac{1}{2}}X_{\frac{\cdot}{\lambda_{n}}},  \lambda_{n}^{\frac{1}{2}}P_{\frac{\cdot}{\lambda_{n}}} ,\lambda_{n}^{\frac{1}{4}} D_{\frac{\cdot}{\lambda_{n}} }   \big) $ converges  to $(0,\frak{p},\frak{m})$ since the variable $X\in \mathbb{T}=[0,1)$ is bounded.  Moreover, $\sigma \big(\lambda_{n}^{\frac{1}{2}}X_{\frac{r}{\lambda_{n}}},  \lambda_{n}^{\frac{1}{2}}P_{\frac{r}{\lambda_{n}}};\, 0\leq r\leq t  \big)$ contains the information in  $\sigma \big(\lambda_{n}^{\frac{1}{4}} D_{\frac{r}{\lambda_{n}} }  ;\, 0\leq r\leq t  \big)$ since $D_{t}$ is defined   as a function of the  Markov process $(X_{r},P_{r})$ over $0\leq r\leq t$.   Thus the path $\lambda_{n}^{\frac{1}{2}}P_{\frac{s}{\lambda_{n}}}$, $s>t$ is independent of   $\sigma \big(\lambda_{n}^{\frac{1}{4}} D_{\frac{r}{\lambda_{n}} }  ;\, 0\leq r\leq t  \big)$ when given $\sigma \big(\lambda_{n}^{\frac{1}{2}}X_{\frac{r}{\lambda_{n}}},  \lambda_{n}^{\frac{1}{2}}P_{\frac{r}{\lambda_{n}}};\, 0\leq r\leq t  \big)$.  This independence carries over into the limit $n\rightarrow \infty$, and thus $\frak{p}_{s}$ for $s>t$ is independent of $\sigma(\frak{m}_{r};\, 0\leq r\leq t)$ when given the information $\sigma(\frak{p}_{r};\, 0\leq r\leq t)$.

 \vspace{.3cm}

The remainder of the proof is concerned with showing (II). \vspace{.5cm}

 \noindent (II)\hspace{.25cm}  By the triangle inequality,
 \begin{align}\label{Toad}
 \tilde{\mathbb{E}}^{(\lambda)}\Big[&\sup_{0\leq t\leq T}\Big| \lambda^{\frac{1}{2}} \langle \tilde{M} \rangle_{\frac{t}{\lambda} } -\kappa \lambda^{\frac{1}{2}}L_{\frac{t}{\lambda}}    \Big|\Big] \leq  \tilde{\mathbb{E}}^{(\lambda)}\Big[\sup_{0\leq t\leq T}\Big| \lambda^{\frac{1}{2}} \langle \tilde{M} \rangle_{\frac{t}{\lambda} } -\upsilon_{\lambda} \lambda^{\frac{1}{2}}\tilde{N}_{\frac{t}{\lambda}}    \Big|\Big]+ |\upsilon_{\lambda}-\frac{\kappa}{\mathbf{u}}|\,\tilde{\mathbb{E}}^{(\lambda)}\big[  \lambda^{\frac{1}{2}}\tilde{N}_{\frac{T}{\lambda}}\big]\nonumber  \\  &+\frac{\kappa}{\mathbf{u}}\tilde{\mathbb{E}}^{(\lambda)}\Big[\sup_{0\leq t\leq T}\Big|  \lambda^{\frac{1}{2}}\tilde{N}_{\frac{t}{\lambda}} -  \lambda^{\frac{1}{2}}\sum_{n=1}^{\mathbf{N}_{\frac{t}{\lambda}} }h(S_{\tau_{n}}) \Big|\Big]+\kappa\tilde{\mathbb{E}}^{(\lambda)}\Big[\sup_{0\leq t\leq T}\Big|  \mathbf{u}^{-1}\lambda^{\frac{1}{2}}\sum_{n=1}^{\mathbf{N}_{\frac{t}{\lambda}} }h(S_{\tau_{n}}) -  \lambda^{\frac{1}{2}}L_{\frac{t}{\lambda}} \Big|\Big],
 \end{align}
where $\mathbf{N}_{t}$ is the number of partition times up to time $t\in \R^{+}$.  The first term on the right is $\mathit{O}(\lambda^{\frac{1}{4}})$ by Part (2) of Lem.~\ref{TrivialMart}. The second term is bounded through
$$|\upsilon_{\lambda}-\frac{\kappa}{\mathbf{u}}|\,\tilde{\mathbb{E}}^{(\lambda)}\big[  \lambda^{\frac{1}{2}}\tilde{N}_{\frac{T}{\lambda}}\big]=   |\upsilon_{\lambda}-\frac{\kappa}{\mathbf{u}}|\,\mathbb{E}^{(\lambda)}\Big[  \lambda^{\frac{1}{2}}\int_{0}^{\frac{T}{\lambda}}dr h(S_{r})\Big]\longrightarrow 0,  $$
where we have used Part (1) of Prop.~\ref{TrivialMart} for the equality.   For convergence to zero, we have used Thm.~\ref{ThmLocalTime} to get a uniform constant bound for the expectation of $\lambda^{\frac{1}{2}}\int_{0}^{\frac{T}{\lambda}}dr h(S_{r})$ over $\lambda<1$ and  Lem~\ref{LemUpsilon}, which gives that  $\upsilon_{\lambda}\rightarrow \frac{\kappa}{\mathbf{u}}$  as $\lambda\rightarrow 0$. 

For the third term in~(\ref{Toad}),
\begin{align}\label{Elijiah}
&\tilde{\mathbb{E}}^{(\lambda)}\Big[\sup_{0\leq t\leq T}\Big|  \lambda^{\frac{1}{2}}\tilde{N}_{\frac{t}{\lambda}} -  \lambda^{\frac{1}{2}}\sum_{n=1}^{\mathbf{N}_{\frac{t}{\lambda}} }h(S_{\tau_{n}}) \Big|\Big]\leq 2\tilde{\mathbb{E}}^{(\lambda)}\Big[ \Big| \lambda^{\frac{1}{2}}\tilde{N}_{\frac{T}{\lambda}} -  \lambda^{\frac{1}{2}}\sum_{n=1}^{\mathbf{N}_{\frac{T}{\lambda}} }h(S_{\tau_{n}}) \Big|^{2}\Big]^{\frac{1}{2}}\nonumber \\ & = 2\lambda^{\frac{1}{2}}\tilde{\mathbb{E}}^{(\lambda)}\Big[  \sum_{n=1}^{\mathbf{N}_{\frac{T}{\lambda}} }h(S_{\tau_{n}})-h^{2}(S_{\tau_{n}}) \Big]^{\frac{1}{2}}\leq   2\lambda^{\frac{1}{2}}\mathbb{E}^{(\lambda)}\Big[  \sum_{n=1}^{\mathbf{N}_{\frac{T}{\lambda}} }h(S_{\tau_{n}}) \Big]^{\frac{1}{2}}=2\lambda^{\frac{1}{2}}\mathbb{E}^{(\lambda)}\Big[  \int_{0}^{\frac{T}{\lambda} }dr h(S_{r}) \Big]^{\frac{1}{2}}.
\end{align}
The first inequality uses Jensen's inequality and  Doob's inequality since
\begin{align*}
\tilde{N}_{t} -  \sum_{n=1}^{\mathbf{N}_{t} }h(S_{\tau_{n}})=  \sum_{n=1}^{\mathbf{N}_{t} } \chi( Z_{\tau_{n}}=1)-     h(S_{\tau_{n}})
\end{align*}
is a martingale with respect $\tilde{\mathcal{F}}_{t}$ by 
Prop.~\ref{TrivialMart}.   The first equality in~(\ref{Elijiah}) follows because the quadratic variation of the martingale is
 $  \sum_{n=1}^{t }\big(\chi(Z_{\tau_{n}}=1)-  h(S_{\tau_{n}})\big)^{2} $, and 
$$\tilde{\mathbb{E}}\big[ \big(\chi(Z_{r}=1)-  h(S_{r})\big)^{2} \,\big|\,\tilde{\mathcal{F}}_{r^{-}},\,\mathbf{N}_{r}-\mathbf{N}_{r^{-}}=1\big]=h(S_{r})-h^{2}(S_{r}),   $$
by  Prop.~\ref{BasicsOfNum}.   For the second inequality, we discard $h^{2}(S_{\tau_{n}})$, and go from the split to the original statistics since the argument of the expectation is well-defined for the original statistics.   Finally, the last equality holds since the partition times $\tau_{n}$ occur with Poisson rate one  independently of the process $S_{t}$.

The fourth term in~(\ref{Toad}) is similar to the third.  The process $\mathbf{u}^{-1}\sum_{n=1}^{\mathbf{N}_{t} }h(S_{\tau_{n}}) -  L_{t} $ is  well-defined in the original statistics and is a martingale with respect to the  filtration $\mathcal{F}_{t}$ by Prop.~\ref{TrivialMart}.  With routine arguments
\begin{align*}
\tilde{\mathbb{E}}^{(\lambda)}\Big[&\sup_{0\leq t\leq T}\Big|  \frac{\lambda^{\frac{1}{2}}}{ \mathbf{u}  }\sum_{n=1}^{\mathbf{N}_{\frac{t}{\lambda}} }h(S_{\tau_{n}}) -  \lambda^{\frac{1}{2}}L_{\frac{t}{\lambda}} \Big|\Big]= \mathbb{E}^{(\lambda)}\Big[\sup_{0\leq t\leq T}\Big|  \frac{ \lambda^{\frac{1}{2}}}{ \mathbf{u}  }\sum_{n=1}^{\mathbf{N}_{\frac{t}{\lambda}} }h(S_{\tau_{n}}) -  \lambda^{\frac{1}{2}}L_{\frac{t}{\lambda}} \Big|\Big] \\ &\leq  2\mathbb{E}^{(\lambda)}\Big[\Big|  \frac{\lambda^{\frac{1}{2}}}{ \mathbf{u} }\sum_{n=1}^{\mathbf{N}_{\frac{T}{\lambda}} }h(S_{\tau_{n}}) -  \lambda^{\frac{1}{2}}L_{\frac{T}{\lambda}} \Big|^{2}\Big]^{\frac{1}{2}}= 
\frac{\lambda^{\frac{1}{2}} }{\mathbf{u}}\mathbb{E}^{(\lambda)}\Big[  \int_{0}^{\frac{T}{\lambda}}dr  h^{2}(S_{r})     \Big]^{\frac{1}{2}} =\mathit{O}(\lambda^{\frac{1}{4}}).
\end{align*}
The  inequality uses Jensen's and Doob's inequalities.  The second equality uses that  the predictable quadratic variation of $\mathbf{u}^{-1}\sum_{n=1}^{\mathbf{N}_{t} }h(S_{\tau_{n}}) -  L_{t} $ is $\mathbf{u}^{-2}\int_{0}^{t}dr h^{2}(S_{r}) $ since the terms $h(S_{\tau_{n}})$ occur with Poisson rate one independently of the process $ S_{t}$.

\section{Miscellaneous proofs}\label{SecMiscProofs}

\begin{proof}[Proof of Lem.~\ref{LocalEstimate}]
Define the martingale $\frak{m}_{ t,\epsilon}:=\int_{0}^{t}d\mathbf{\tilde{B}}_{r}\big(1-e^{-\frac{|\frak{p}_{r}| }{\epsilon} }\big)$.  The difference between $\frak{m}_{ t,\epsilon} $ and $\mathbf{\tilde{B}}_{t}$ tends to zero as $\epsilon\rightarrow 0$ in the norm  $\mathbb{E}\big[ \sup_{0\leq t\leq T}\big| \cdot \big| \big]$ since
\begin{align}\label{CappedQuad}
\mathbb{E}\big[\sup_{0\leq t\leq T}\big| \mathbf{\tilde{B}}_{t}-\frak{m}_{ t,\epsilon}   \big|\big] & \leq  \mathbb{E}\big[\sup_{0\leq t\leq T}\big| \mathbf{\tilde{B}}_{t}-\frak{m}_{ t,\epsilon}   \big|^{2}\big]^{\frac{1}{2}} \leq  2\mathbb{E}\big[\big| \mathbf{\tilde{B}}_{T}-\frak{m}_{ T,\epsilon}   \big|^{2}\big]^{\frac{1}{2}} \nonumber \\ &=2\mathbb{E}\Big[\int_{0}^{T}dr e^{-2\frac{|\frak{p}_{r}| }{\epsilon} } \Big]^{\frac{1}{2}}= 2\Big( \int_{0}^{T}dr \mathbb{E}\big[e^{-2\frac{|\frak{p}_{r}| }{\epsilon} } \big] \Big)^{\frac{1}{2}}\leq 2\Big( \int_{0}^{T}dr \mathbb{E}_{0}\big[e^{-2\frac{|\frak{p}_{r}| }{\epsilon} } \big] \Big)^{\frac{1}{2}} \nonumber \\  &= 2\Big( \int_{0}^{T}dr \int_{\R}dq \frac{e^{-\frac{1}{2\omega_{r}}q^{2}-2\frac{|q|}{\epsilon}   }  }{(2\pi\omega_{r}  )^{\frac{1}{2}}   } \ \Big)^{\frac{1}{2}} =\mathit{O}(\epsilon^{\frac{1}{2}}),
\end{align}
  where $\omega_{r}:=1-e^{-r}$.  The first inequality is Jensen's, the second is Doob's, and the first equality uses that $ e^{-2\frac{|\frak{p}_{r}| }{\epsilon} }$ is derivative of the  quadratic variation of the martingale $ \mathbf{\tilde{B}}_{r}-\frak{m}_{ r,\epsilon}   $.  The third inequality uses that $\mathbb{E}\big[e^{-2\frac{|\frak{p}_{r}| }{\epsilon} } \big]$ is largest when $\frak{p}_{0}$ is initially zero.   The third equality holds since $\frac{e^{-\frac{1}{2\omega_{t}}q^{2}   }  }{(2\pi\omega_{t}  )^{\frac{1}{2}}   }$ is the density for $\frak{p}_{t}$ starting with $\frak{p}_{0}=0$.

Moreover, $\mathbf{m}_{t,\epsilon}$ can be rewritten
\begin{align*}
\mathbf{m}_{t,\epsilon}= \int_{0}^{t}d\mathbf{\tilde{B}}_{r}\big(1-e^{-\frac{|\frak{p}_{r}| }{\epsilon} }\big)=& \int_{0}^{t}\big(d|\frak{p}_{t}|+\frac{1}{2}dr|\frak{p}_{r}|  \big)\big(1-e^{-\frac{|\frak{p}_{r}| }{\epsilon} }\big) \nonumber \\   =& |\frak{p}_{t}|-|\frak{p}_{0}|+\epsilon e^{-\frac{|\frak{p}_{t}| }{\epsilon} } + \frac{1}{2} \int_{0}^{t}dr |\frak{p}_{r}|  \big(1-e^{-\frac{|\frak{p}_{r}| }{\epsilon} }\big)-\frac{1}{2\epsilon} \int_{0}^{t}dr e^{-\frac{|p_{r}|}{\epsilon} }.
\end{align*}
 The second equality follows by the substitution $d\mathbf{\tilde{B}}_{t}=d|\frak{p}_{t}|-\frac{1}{2}dt|\frak{p}_{t}|-d\frak{l}_{t}$ (from the Tanaka-Meyer formula~(\ref{ItoLocal})) and since $d\frak{l}_{t}$ multiplied by $(1-e^{-\frac{|p_{r}|}{\epsilon} })$ is zero.  The chain rule and the fact that $(d|p_{r}|)^{2}=dr$   give the third equality.  From the convergence~(\ref{CappedQuad}) it follows that the right side converges to $\mathbf{\tilde{B}}$ in the norm $\|\cdot\|_{\frak{s}}=\mathbb{E}\big[ \sup_{0\leq t\leq T}\big| \cdot \big| \big]$. 
 
 As $\epsilon\rightarrow 0$
$$ \Big\|   \epsilon e^{-\frac{|\frak{p}_{t}| }{\epsilon} }   \Big \|_{\frak{s}}=\mathit{O}(\epsilon) \quad \text{and} \quad  \Big \|  \int_{0}^{t}dr|\frak{p}_{r}| e^{-\frac{|\frak{p}_{r}| }{\epsilon} } \Big\|_{\frak{s}}=\mathit{O}(\epsilon), $$
where the later term follows by the same argument as in the right side of~(\ref{CappedQuad}).  In conclusion, $$ \mathbf{\tilde{B}}_{t}=|\frak{p}_{t}| -|\frak{p}_{0}|+ \frac{1}{2} \int_{0}^{t}dr|\frak{p}_{r}|- \frac{1}{2\epsilon}\int_{0}^{t}dr e^{-\frac{|p_{r}|}{\epsilon} }+\mathit{O}(\epsilon^{\frac{1}{2}} ) , $$
where $\mathit{O}(\epsilon^{\frac{1}{2}})$ refers to the norm  $\|\cdot\|_{\frak{s}}$.  By the Tanaka-Meyer formula, we have that  $$\frak{l}_{t}=\lim_{\epsilon \rightarrow 0}\frac{1}{2\epsilon} \int_{0}^{t}dr e^{-\frac{|p_{r}|}{\epsilon} },$$ where the error in the limit is $\mathit{O}(\epsilon^{\frac{1}{2}})$ in $\|\cdot\|_{\frak{s}}$.

\end{proof}
 Lemma~\ref{ReStuff} is a small technical point regarding the distribution for momentum jumps conditioned to exit sets $\{p\in \R\,|\,|p|\geq b \}$ for some $b\in [0,\lambda^{-1}]$, and it is a consequence of the exponential decay found in the jump rates $\mathcal{J}_{\lambda}\big(p,p'\big)$.  
 We will apply Lem.~\ref{ReStuff} in the proof of Lem.~\ref{LowEnergyLemma} below.

\begin{lemma}\label{ReStuff}
For each $m\in \mathbb{N}$, the following inequality holds:
$$
 \sup_{\lambda<1}\sup_{|p|\leq b \leq \lambda^{-1}}  \frac{ \int_{|p'|\geq  b}dp'(|p'|-b)^{m}\,\mathcal{J}_{\lambda}\big(p,p'\big)       }{   \int_{|p'|\geq b}\,dp'\mathcal{J}_{\lambda}\big(p,p'\big)   }<\infty .
$$ 

\end{lemma}
\vspace{.4cm}

The proof of Lem.~\ref{LowEnergyLemma} relies on an application of the submartingale up-crossing inequality to bound the number of the times that the process $H_{r}^{\frac{1}{2}}$, $r\in [0,\frac{T}{\lambda}]$ wanders from below $\epsilon^{\frac{1}{2}}\lambda^{-\frac{\varrho}{2}}$ to above $2\epsilon^{\frac{1}{2}}\lambda^{-\frac{\varrho}{2}}$, which is closely related to the total time such that  $H_{r}\leq \epsilon \lambda^{-\varrho}$ over the interval $r\in [0,\frac{T}{\lambda}]$.  The process $H_{r}^{\frac{1}{2}}$ behaves nearly as a submartingale at low energies in the sense that a manageable perturbation $H_{r}^{\frac{1}{2}}+c\lambda H_{r}^{\frac{3}{2}}$, for large enough $c>0$, is a submartingale at low energies.  This contrivance is not necessary for the $\lambda=0$ case of the dynamics for which  $H_{r}^{\frac{1}{2}}$ is a submartingale with the desired properties.

\begin{proof}[Proof of Lem.~\ref{LowEnergyLemma}]\text{  }\\
For $b>0$ let $\gamma$ be the minimum of the hitting time that $H_{t}$ jumps above $b\lambda^{-2}$ and the final time $\frac{T}{\lambda}$. We have the following inequalities:
\begin{align*}
 \mathbb{E}^{(\lambda)}\big[\mathbf{T}_{\frac{T}{\lambda} } \big] &  \leq T\mathbb{P}^{(\lambda)}\big[ \sup_{0\leq r\leq \frac{T}{\lambda} }H_{r}>b\lambda^{-2}\big]+ \mathbb{E}^{(\lambda)}\big[\mathbf{T}_{\gamma} \big]\\ & \leq T\frac{\lambda^{4}}{b^{2}} T\mathbb{E}^{(\lambda)}\Big[ \Big(\sup_{0\leq r\leq  \frac{T}{\lambda} }H_{r}\Big)^{2}\Big]+ \mathbb{E}^{(\lambda)}\big[\mathbf{T}_{\gamma} \big]   \leq CT^{3}\frac{\lambda^{2}}{ b^{2}  }+ \mathbb{E}^{(\lambda)}\big[\mathbf{T}_{\gamma} \big],  
 \end{align*}
where the second inequality is Chebyshev's and the $C>0$ in the third is from Lem.~\ref{FirstEnergyLem}.   With the restriction $ \epsilon \geq \lambda^{\varrho}$, the term following the last inequality decays faster than   $\epsilon^{\frac{1}{2}} \lambda^{\frac{1-\varrho}{2}}$ as $\lambda\rightarrow 0$, so we can can focus our study to $\mathbb{E}^{(\lambda)}\big[\mathbf{T}_{\gamma} \big]$. The energy process $H_{t}=H(X_{t}, P_{t})$ behaves as a submartingale for time periods in which $H_{t}\leq b'\lambda^{-2}$ for small enough $b'>0$.  More precisely there exists $0\leq b'\leq 1$, $\underline{\sigma}>0$ such that for all $\lambda<1$ and all $(x,p)$ with $H(x,p)\leq b'\lambda^{-2}$,
\begin{align}
\hspace{1cm}\mathbf{q}^{(\lambda)}_{1}(x,p):=\frac{d}{dt} \mathbb{E}_{(x,p)}^{(\lambda)}\big[ H_{t}   \big]\big|_{t=0}&=  \frac{1}{2}\int_{\R}dp^{\prime}\big( (p^{\prime})^{2}+V(x) -p^{2}-V(x)   \big)\mathcal{J}_{\lambda}\big(p,p^{\prime})\nonumber  \\ &=  \frac{1}{2}\int_{\R}dp^{\prime}\big( (p^{\prime})^{2} -p^{2}   \big)\mathcal{J}_{\lambda}\big(p,p^{\prime})\geq \underline{\sigma} \label{SubMart} .   
\end{align}
 From~(\ref{SubMart}) we  have that for all $m\geq 1$, $\lambda< 1$, and   $H(x,p)\leq b'\lambda^{-2}$,
\begin{align}\label{SubMartAgain}
\hspace{1cm}\mathbf{q}^{(\lambda)}_{m}(x,p) & =  \int_{\R}dp^{\prime}\Big(\Big( \frac{1}{2}(p^{\prime})^{2}+V(x)\Big)^{m} -\Big(\frac{1}{2}p^{2}+V(x)\Big)^{m} \Big)\mathcal{J}_{\lambda}\big(p,p^{\prime})\nonumber \\  &  \geq m\underline{\sigma}H^{m-1}(x,p),  
\end{align}
where $\mathbf{q}^{(\lambda)}_{m}(x,p):=\frac{d}{dt} \mathbb{E}_{(x,p)}^{(\lambda)}\big[ H_{t}^{m}  \big]\big|_{t=0}$.  The inequality~(\ref{SubMartAgain}) follows from~(\ref{SubMart}) since $f(y)=|y|^{m}$ is convex, and thus 
$$ f(Y_{1})-f(Y_{0})\geq (Y_{1}-Y_{0})f'(Y_{0})  $$
for $Y_{1}:=\frac{1}{2}(p^{\prime})^{2}+V(x)$ and $Y_{0}:=\frac{1}{2}p^{2}+V(x)$.

Notice that the value $\mathbf{q}^{(\lambda)}_m(S_{t})$ is the derivative of the predictable part of the semimartingale decomposition of the process $H_{t}^{m}$.  In other terms, the following is a martingale:
$$ H_{t}^{m}-\int_{0}^{t}dr\mathbf{q}^{(\lambda)}_m(S_{r}).   $$
  In addition to the lower bounds in~(\ref{SubMartAgain}), there are upper bounds
\begin{align}\label{UpperSigma}
\mathbf{q}^{(\lambda)}_{m}(x,p)\leq \overline{\sigma}_{m}\big(1+H^{m-1}(x,p) \big) \end{align}
 which hold for some constants $ \overline{\sigma}_{m}$ and all $\lambda<1$ and $(x,p)$ with $H(x,p)\leq \lambda^{-2}$.

Using the above observations  there is a useful submartingale that is ``close" to $H_{t}^{\frac{1}{2}}$.  Let the $b>0$ defining the stopping time $\gamma$ be the $b'$ chosen to ensure the condition~(\ref{SubMart}).   There exists a $c>0$ such that for all $\lambda$ small enough, 
 $$\kappa_{t}=H^{\frac{1}{2}}_{t}+c\lambda H_{t}^{\frac{3}{2}}     $$
is a submartingale over the time interval $t\in [0,\gamma]$.  To see that $\kappa_{t}$ is a submartingale up to time $\gamma$, first notice that  the predictable component $\int_{0}^{t}dr\mathbf{q}^{(\lambda)}_{\frac{3}{2}}(S_{r}) $ in the semimartingale decomposition of $H_{t}^{\frac{3}{2}}$ increases with rate greater than $\frac{3}{2}\underline{\sigma}H^{\frac{1}{2}}_{t}$ by~(\ref{SubMartAgain}). Moreover, the
 predictable component of the semimartingale decomposition of $H_{t}^{\frac{1}{2}}$ is $2^{-\frac{1}{2}}\int_{0}^{t}dr\mathcal{A}_{\lambda}(X_{r},P_{r})$, and   the negative part of the function  $\mathcal{A}_{\lambda}$  satisfies the inequality $ \mathcal{A}_{\lambda}^{-}(x,p)\leq C\lambda H^{\frac{1}{2}}(x,p)$ for some $C>0$ by Part (4) of Prop.~\ref{AMinus} and the elementary inequality $|p|\leq 2^{\frac{1}{2}}H^{\frac{1}{2}}(x,p)$.  Thus we can choose $c:=\frac{2}{3}\frac{C}{\underline{\sigma}}$ to ensure that $\kappa_{t} $  is a submartingale over the specified time interval.

  Set $\varsigma^{\prime}_{0}=\varsigma_{0}=\varsigma^{\prime}_{1}=0$, and define the stopping times $\varsigma_{n},\varsigma_{n}^{\prime}\leq \gamma$ such that for $n\geq 1$,
\begin{align*}
\varsigma^{\prime}_{n}= \min \{r\in (\varsigma_{n-1},\infty)\,\big| \,H_{r}\leq  \epsilon\lambda^{-\varrho }         \},\quad \quad
\varsigma_{n}= \min \{ r\in (\varsigma^{\prime}_{n},\infty)\,\big| \, H_{r}\geq 4 \epsilon \lambda^{-\varrho }            \}.
\end{align*}
The above definition assumes $H_{0}< 4 \epsilon \lambda^{-\varrho } $ and otherwise we should only take $\varsigma^{\prime}_{0}=\varsigma_{0}=0$.  
Let $\mathbf{n}_{\gamma}$ be the number $\varsigma^{\prime}_{n}$'s less than $\gamma$.  In other words, $\mathbf{n}_{\gamma}$ is one more than the number of up-crossings of $ H_{r}$ from $\lambda^{-\varrho}$ to $4\lambda^{-\varrho}$  that have been completed by time $\gamma$.  Let $\mathbf{n}_{\gamma}^{\prime}$ be defined similarly as one plus  the number of crossings of $\kappa_{t}$ from $\frac{3}{2}\epsilon^{\frac{1}{2}}\lambda^{-\frac{\varrho}{2}}$ to $2\epsilon^{\frac{1}{2}}\lambda^{-\frac{\varrho}{2} }$.  For  $\lambda<(\frac{1}{2c})^{\frac{1}{1-\varrho}}$  we have both of the implications 
$$ H_{t}\geq 4 \epsilon\lambda^{-\varrho }\implies  \kappa_{t}\geq 2 \epsilon^{\frac{1}{2}} \lambda^{-\frac{\varrho}{2}} \quad \text{and}\quad   H_{t}\leq  \epsilon \lambda^{-\varrho}\implies \kappa_{t}\leq \frac{3}{2}\epsilon^{\frac{1}{2}}\lambda^{-\frac{\varrho}{2}} ,  $$
and hence $\mathbf{n}^{\prime}_{\gamma}\geq \mathbf{n}_{\gamma}$.  The definitions give us the almost always inequality $$\mathbf{T}_{t}\leq \lambda\sum_{n=1}^{\mathbf{n}_{t}}\varsigma_{n}-\varsigma^{\prime}_{n}.    $$
Next observe that
\begin{align}\label{IncursionTime}
\lambda^{-1}\mathbb{E}^{(\lambda)}\big[\mathbf{T}_{\gamma}  \big]\leq \mathbb{E}^{(\lambda)}\Big[  \sum_{n=1}^{\mathbf{n}_{\gamma}}\varsigma_{n}-\varsigma^{\prime}_{n}\Big]\leq \mathbb{E}^{(\lambda)}\big[ \mathbf{n}_{\gamma}\big]\sup_{n\in \mathbb{N}}\mathbb{E}^{(\lambda)}\big[\varsigma_{n}-\varsigma^{\prime}_{n}\,\big|\, n\leq \mathbf{n}_{\gamma }   \big]  . 
\end{align}
With the above we have an upper bound in terms of the expectation for the number of up-crossings $\mathbf{n}_{t}$ and the expectation for the duration of a single up-crossing $ \varsigma_{n}-\varsigma_{n}^{\prime}$ conditioned on the event $n\leq \mathbf{n}_{t}$. By the observation above, $\mathbb{E}^{(\lambda)}\big[ \mathbf{n}_{\gamma}\big]\leq \mathbb{E}^{(\lambda)}\big[\mathbf{n}_{\gamma}^{\prime} \big]$.   By the submartingale up-crossing inequality~\cite{Chung}, we have the first inequality below: 
\begin{align}
 \mathbb{E}^{(\lambda)}\big[\mathbf{n}_{\gamma}^{\prime} \big] & \leq \frac{ \mathbb{E}^{(\lambda)}\big[\kappa_{\gamma}  \big]}{2 \epsilon^{\frac{1}{2}} \lambda^{-\frac{\varrho}{2}}-\epsilon^{\frac{1}{2}}\frac{3}{2}\lambda^{-\frac{\varrho}{2} }    } = 2\epsilon^{-\frac{1}{2}}\lambda^{\frac{\varrho}{2}} \mathbb{E}^{(\lambda)}\big[H_{\gamma }^{\frac{1}{2}}+c\lambda H_{\gamma}^{\frac{3}{2}}  \big]\leq 2\epsilon^{-\frac{1}{2}}\lambda^{\frac{\varrho}{2}}\Big( \mathbb{E}^{(\lambda)}\big[H_{\gamma }\big]^{\frac{1}{2}}+c\lambda \mathbb{E}\big[H_{\gamma}^{\frac{3}{2}}  \big] \Big)\nonumber \\ & \leq  2 \epsilon^{-\frac{1}{2}} \lambda^{\frac{\varrho}{2}}\big(\mathbb{E}^{(\lambda)}\big[H_{0}\big]  +\lambda^{-1}T\overline{\sigma }_{1} \big)^{\frac{1}{2}} +2c\epsilon^{-\frac{1}{2}} \lambda^{1+\frac{\varrho}{2}} \big(  \mathbb{E}^{(\lambda)}\big[H_{0}\big]  +\lambda^{-\frac{3}{2} +\frac{\varrho}{2} }T^{\frac{3}{2}}\overline{\sigma }_{\frac{3}{2}} \big) \nonumber \\  &<4\epsilon^{-\frac{1}{2}} \overline{\sigma}_{1}^{\frac{1}{2}}T^{\frac{1}{2}}
\lambda^{\frac{\varrho-1}{2}},\label{Numero}
\end{align}
where the last inequality is for $\lambda$ small enough.  The second inequality is Jensen's, and the third uses that $\gamma\leq \frac{T}{\lambda}$ and the bound~(\ref{UpperSigma}) for the derivatives of the  predictable components of the semimartingales $H_{t}$ and $H^{\frac{3}{2}}_{t}$.

We now focus on the expectation of the incursions $\varsigma_{n}-\varsigma_{n}^{\prime}$.  Whether or not the event $n\leq \mathbf{n}_{t}$ occurred will be known at time $\varsigma^{\prime}_{n}$, so    
$$\sup_{n\in \mathbb{N}}\mathbb{E}^{(\lambda)}\big[\varsigma_{n}-\varsigma_{n}^{\prime}\,\big|\, n\leq \mathbf{n}_{t}   \big]\leq \sup_{n\in \mathbb{N},\,\omega\in \mathcal{F}_{\varsigma_{n}'} }\mathbb{E}^{(\lambda)}\big[\varsigma_{n}-\varsigma_{n}^{\prime}\,\big|\,\mathcal{F}_{\varsigma_{n}^{\prime} } \big]=\sup_{H(s)\leq \lambda^{-\varrho }  }\mathbb{E}^{(\lambda)}_{s}\big[\varsigma_{1}  \big]. $$
By~(\ref{SubMart}) 
\begin{align}\label{Sharia}
 \underline{\sigma}\, \mathbb{E}^{(\lambda)}_{s}\big[\varsigma_{1}\big] \leq \mathbb{E}^{(\lambda)}_{s}\Big[ \int_{0 }^{\varsigma_{1} }dr\mathbf{q}^{(\lambda)}_{1}(X_{r},P_{r})  \Big]= \mathbb{E}_{s}^{(\lambda)}\big[H_{\varsigma_{1} }-H_{0}  \big], 
 \end{align}
where the equality holds  by the optional sampling theorem since $\int_{0}^{t}dr\mathbf{q}^{(\lambda)}_{1}(X_{r},P_{r})$ is the predictable part of the semimartingale decomposition for $H_{t}-H_{0}$.  This application of the optional sampling theorem is legal since $\varsigma_{1}$ is almost surely finite, and $$\mathbb{E}_{s}^{(\lambda)}[H_{ r}\chi(r<\varsigma_{1} )]\leq \lambda^{-\frac{1}{2}}\mathbb{P}^{(\lambda)}_{s}[ r<\varsigma_{1}]\longrightarrow 0\quad \text{as}\quad  r\longrightarrow \infty. $$    

Continuing with the right side of~(\ref{Sharia}),
  \begin{align*}  \mathbb{E}_{s}^{(\lambda)}\big[H_{\varsigma_{1} }-H_{0}  \big] &=  \mathbb{E}_{s}^{(\lambda)}\big[H_{\varsigma_{1}^{-} }-H_{0}  \big]+\frac{1}{2}\mathbb{E}_{s}^{(\lambda)}\big[P^{2}_{\varsigma_{1} }- P^{2}_{\varsigma_{1}^{-} } \big]\nonumber \\  &
\leq 4 \epsilon \lambda^{-\varrho}+\mathbb{E}_{s}^{(\lambda)}[\Delta^{2}]=4  \epsilon \lambda^{-\varrho}+\mathit{O}(\lambda^{0})\leq c\epsilon\lambda^{-\varrho}. 
\end{align*}
where $\Delta:=P_{\varsigma_{1} }- P_{\varsigma_{1}^{-}}$.  For the first inequality,  we have used that $P^{2}_{\varsigma_{1}^{-}}\leq H_{\varsigma_{1}^{-} }\leq  \epsilon\lambda^{-\varrho}$, and the inequality $(x+y)^{2}\leq 2(x^{2}+y^{2})$.  The last inequality holds for some $c>0$ by our restriction $\epsilon\geq \lambda^{\varrho}$.  The term $\mathbb{E}_{s}^{(\lambda)}[\Delta^{2}]$ is uniformly bounded for $\lambda< 1$ since  by nested conditional expectations $\mathbb{E}_{s}^{(\lambda)}[\Delta^{2}]=\mathbb{E}_{s}^{(\lambda)}[ \mathbb{E}^{(\lambda)}[\Delta^{2}\,|\,\mathcal{F}_{\varsigma_{1}^{-} },\,\varsigma_{1}   ]]$ and 
\begin{align*}
 \mathbb{E}^{(\lambda)}[\Delta^{2}\,|\,\mathcal{F}_{\varsigma_{1}^{-} } ,\,\varsigma_{1}  ] &=   \frac{ \int_{H(X_{\varsigma_{1}^{-}},p' )\geq \epsilon\lambda^{-\varrho} } dp^{\prime} \big(p^{\prime}-P_{\varsigma_{1}^{-}}\big)^{2} \mathcal{J}_{\lambda}\big(P_{\varsigma_{1}^{-}},p^{\prime}\big)}{\int_{H(X_{\varsigma_{1}^{-}},p^{\prime})\geq \epsilon\lambda^{-\varrho}  }dp^{\prime} \mathcal{J}_{\lambda}\big(P_{\varsigma_{1}^{-}},p^{\prime}\big) }
 \\ &  \leq \sup_{\lambda<1}\sup_{H(x,p)\leq \epsilon  \lambda^{-\varrho}}\frac{ \int_{H(x,p^{\prime})\geq \epsilon\lambda^{-\varrho} } dp^{\prime} (p^{\prime}-p)^{2}    \mathcal{J}_{\lambda}(p,p^{\prime})}{\int_{H(x,p^{\prime})\geq \epsilon \lambda^{-\varrho}  }dp^{\prime} \mathcal{J}_{\lambda}(p,p^{\prime}) }<\infty.
 \end{align*}
The equality relies on the strong Markov property since  the distribution for $\Delta$ is independent of  $\mathcal{F}_{\varsigma_{1}^{-} } $ when  given $\varsigma_{1}$ and $S_{\varsigma_{1}^{-}}$.  The last expression is finite by Lem.~\ref{ReStuff}.

Putting our results for $\mathbb{E}^{(\lambda)}\big[ \mathbf{n}_{\gamma}\big]$ and $\sup_{n\in \mathbb{N}}\mathbb{E}^{(\lambda)}\big[ \varsigma_{n}-\varsigma_{n}^{\prime}\,\big|\,n\leq \mathbf{n}_{\gamma}\big]$ together,   $$\mathbb{E}^{(\lambda)}\big[  \mathbf{T}_{\gamma}\big]\leq  \lambda \mathbb{E}^{(\lambda)}\big[ \mathbf{n}_{\gamma}\big] \sup_{n\in \mathbb{N}}\mathbb{E}^{(\lambda)}\big[ \varsigma_{n}-\varsigma_{n}^{\prime}\,\big|\,n\leq \mathbf{n}_{\gamma}\big]\leq 4 c\epsilon^{\frac{1}{2}}\overline{\sigma}^{\frac{1}{2}}T^{\frac{1}{2}}  \lambda^{\frac{1-\varrho}{2}} . $$
This completes the proof.

\end{proof}

\vspace{.5cm}

The proof of Lem.~\ref{LocalTimeBnd} follows by a fairly standard argument for bounding the difference between two additive functionals using the splitting structure from Sect.~\ref{SecNumSplit}.  Several results from~\cite{Previous} from will be used in the proof.

\begin{proof}[Proof of Lem.~\ref{LocalTimeBnd}]
By Part (1) of Prop.~\ref{TrivialMart}, we have the equality $ \mathbf{u}\tilde{\mathbb{E}}^{(\lambda)}\big[ \tilde{N}_{t}\big]= \tilde{\mathbb{E}}^{(\lambda)}\big[ L_{t}\big]$. Moreover, we have the uniform bound 
\begin{align}\label{SureFire}
\sup_{\lambda<1}\tilde{\mathbb{E}}^{(\lambda)}\big[ \lambda^{\frac{1}{2}}\tilde{N}_{\frac{T}{\lambda}}\big]<\infty
\end{align}  by~\cite[Lem. 3.3]{Previous}, and thus we also have that  $\tilde{\mathbb{E}}^{(\lambda)}\big[ \lambda^{\frac{1}{2}}L_{\frac{T}{\lambda}}\big]$ is uniformly bounded.

 To show that $ \lambda^{\frac{1}{2}} L_{\frac{t}{\lambda}}$ is close to $\lambda^{\frac{1}{2} }\mathbf{A}_{\frac{t}{\lambda}}^{+}$, we will consider the split dynamics.  Going to the split statistics in the first equality below, we have
\begin{multline}\label{Kela}
\mathbb{E}^{(\lambda)}\Big[\sup_{0\leq t\leq T}\Big|\lambda^{\frac{1}{2}} L_{\frac{t}{\lambda}}    -\lambda^{\frac{1}{2} }\mathbf{A}_{\frac{t}{\lambda}}^{+} \Big|      \Big]= \tilde{\mathbb{E}}^{(\lambda)}\Big[\sup_{0\leq t\leq T}\Big|\lambda^{\frac{1}{2}} L_{\frac{t}{\lambda}}    -\lambda^{\frac{1}{2} }\mathbf{A}_{\frac{t}{\lambda}}^{+} \Big|      \Big] \\ \leq \lambda^{\frac{1}{2}} \tilde{\mathbb{E}}^{(\lambda)}\Big[\Big(\int_{0}^{R_{1}}dr+\sup_{R_{1}\leq t\leq \frac{T}{\lambda} } \int_{t}^{ R_{\tilde{N}_{t}+1 } } dr  \Big)  g_{\lambda}(S_{r})      \Big] +\lambda^{\frac{1}{2}}\tilde{\mathbb{E}}^{(\lambda)}\Big[\sup_{0\leq t\leq \frac{T}{\lambda} } \Big| \sum_{n=1}^{\tilde{N}_{t}}  \int_{R_{n}}^{R_{n+1}}g_{\lambda}'(S_{r})  \Big| \Big],
\end{multline}
where $R_{n}$ is the beginning time of the $n$ life cycle, $ \tilde{N}_{t}$ is the number of $R_{n}'$  that have occurred up to time $t$, $g_{\lambda}:=\mathcal{A}_{\lambda}^{+}+\mathbf{u}^{-1}h$, and  $g_{\lambda}':=\mathcal{A}_{\lambda}^{+}-\mathbf{u}^{-1}h$.

The first term on the right side of~(\ref{Kela}) contains the boundary terms for the partition of the integrals over the interval $[0,\frac{T}{\lambda}]$ using the life cycle times $R_{n}$.  By Part (2) of Prop.~\ref{AMinus} and since $h$ has compact support,  $g_{\lambda}:=\mathcal{A}_{\lambda}^{+}+\mathbf{u}^{-1}h$ satisfies the conditions  of~\cite[Prop. 4.3]{Previous}.   By Part (2) of~\cite[Prop. 4.3]{Previous}, there is a $C>0$ such that the inequality below holds:  
\begin{align*}
\lambda^{\frac{1}{2}} \tilde{\mathbb{E}}^{(\lambda)}\Big[&\int_{0}^{R_{1}}dr g_{\lambda}(S_{r})    \Big]= \lambda^{\frac{1}{2}}\int_{\tilde{\Sigma}}d\tilde{\mu}(x,p,z) \tilde{\mathbb{E}}^{(\lambda)}_{(x,p,z)}\Big[\int_{0}^{R_{1}}dr g_{\lambda}(S_{r})     \Big] \\  &\leq C\lambda^{\frac{1}{2}} \int_{\tilde{\Sigma}}d\tilde{\mu}(x,p,z)\big(1+\log(1+|p|)    \big)= C\lambda^{\frac{1}{2}}\int_{\Sigma}d\mu(x,p)\big(1+\log(1+|p|)    \big)=\mathit{O}(\lambda^{\frac{1}{2}}),
\end{align*}
where the measure $\tilde{\mu}$ on $\tilde{\Sigma}$ is the splitting of the measure $\mu$. The integral above is finite by our assumptions on the initial measure $\mu$.  

The other part of the first term on the right side of~(\ref{Kela}) is bounded through
\begin{align*}
\lambda^{\frac{1}{2}} \tilde{\mathbb{E}}^{(\lambda)}\Big[\sup_{0\leq n\leq   \tilde{N}_{\frac{T}{\lambda}} } \int_{R_{n}}^{R_{n+1}} dr g_{\lambda}(S_{r})      \Big]& \leq  \lambda^{\frac{1}{2}} \tilde{\mathbb{E}}^{(\lambda)}\Big[ \sum_{n=1}^{ \tilde{N}_{\frac{T}{\lambda}}   }\Big(\int_{R_{n}}^{R_{n+1}} dr g_{\lambda}(S_{r}) \Big)^{2}      \Big]^{\frac{1}{2}}\\  &=  \lambda^{\frac{1}{2}} \tilde{\mathbb{E}}^{(\lambda)}\Big[ \sum_{n=1}^{ \tilde{N}_{\frac{T}{\lambda}}   }\tilde{\mathbb{E}}^{(\lambda)}\Big[\Big(\int_{R_{n}}^{R_{n+1}} dr g_{\lambda}(S_{r}) \Big)^{2} \,\Big|\,\tilde{\mathcal{F}}_{R_{n}'}  \Big]     \Big]^{\frac{1}{2}} \\ &\leq C^{\frac{1}{2}} \lambda^{\frac{1}{2}} \tilde{\mathbb{E}}^{(\lambda)}\big[ \tilde{N}_{\frac{T}{\lambda}} \big]^{\frac{1}{2}} =\mathit{O}(\lambda^{\frac{1}{4}}).
\end{align*}
The first inequality uses that $\sup_{n}a_{n}\leq (\sum_{n}a_{n}^2)^{\frac{1}{2}}$ for positive numbers $a_{n}>0$ followed by Jensen's inequality. The  second inequality uses the strong Markov property for the split chain $\tilde{\sigma}_{m}=S_{\tau_{m}}$ and that $\tilde{S}_{R_{n}}$ has distribution $\tilde{\nu}$ by $\tilde{\mathcal{F}}_{R_{n}'}$  by Part (1) of~\cite[Prop. 2.1]{Previous}.

For the second term on the right side of~(\ref{Kela}), the key observation is that 
\begin{align*}
b_{\lambda}:= & \tilde{\mathbb{E}}^{(\lambda)}\Big[ \int_{R_{n}}^{R_{n+1}}dr g_{\lambda}'(S_{r})\,\Big|\, \tilde{\mathcal{F}}_{R_{n}'}\Big]  =  \tilde{\mathbb{E}}_{\tilde{\nu}}^{(\lambda)}\Big[ \int_{0}^{R_{1}}dr g_{\lambda}'(S_{r})\Big] \\ = &\frac{ \int_{\Sigma}ds \Psi_{\infty,\lambda}(s) \big(\mathcal{A}^{+}_{\lambda}(s)-\mathbf{u}^{-1}h(s) \big)     }{  \int_{\Sigma}ds \Psi_{\infty,\lambda}(s)h(s)  }=\frac{ \int_{\Sigma}ds e^{-\lambda H(s)} \big(\mathcal{A}^{+}_{\lambda}(s)-\mathbf{u}^{-1}h(s) \big)     }{  \int_{\Sigma}ds e^{-\lambda H(s)} h(s)  }
\end{align*}
is $\mathit{O}(\lambda^{\frac{1}{2}})$ for small $\lambda$. The first equality is by the strong Markov property for the chain $\tilde{\sigma}_{n}=\tilde{S}_{\tau_{n}}$  since $\tilde{S}_{\tau_{n}}$ has distribution $\tilde{\nu}$ when conditioned on the information $\tilde{\mathcal{F}}_{R_{n}'}$  by Part (1) of~\cite[Prop. 2.1]{Previous}.  The second equality is by Part (2) of~\cite[Prop. 2.4]{Previous}.  The denominator of the rightmost expression approaches $\int_{\Sigma}ds h(s)=\mathbf{u}$, and the numerator is a difference of terms which are $1+\mathit{O}(\lambda^{\frac{1}{2}})$.  This follows since $\mathbf{u}^{-1}\int_{\Sigma}ds h(s)=1$, by the approximation $ \int_{\Sigma}ds \mathcal{A}^{+}_{\lambda}(s)=1+\mathit{O}(\lambda^{\frac{1}{2}})$  in Part (7) of Prop.~\ref{AMinus}, and  since inserting the factor $e^{-\lambda H(s)}$ in the integrals will perturb the values by $\mathit{O}(\lambda^{\frac{1}{2}})$.

With the above and the triangle inequality, the second term on the right side of~(\ref{Kela}) is smaller than
\begin{align*}
|b_{\lambda} |&\lambda^{\frac{1}{2}}\tilde{\mathbb{E}}^{(\lambda)}\big[ \tilde{N}_{\frac{T}{\lambda}}  \big]+\lambda^{\frac{1}{2}}\tilde{\mathbb{E}}^{(\lambda)}\Big[\sup_{0\leq t\leq \frac{T}{\lambda} }  \Big|\sum_{n=1}^{\lfloor \frac{1}{2}\tilde{N}_{t}+\frac{1}{2} \rfloor }  \Big(\int_{R_{2n-1}}^{R_{2n}}dr g_{\lambda}'(S_{r})-b_{\lambda}  \Big) \Big| \Big] \\  &+\lambda^{\frac{1}{2}}\tilde{\mathbb{E}}^{(\lambda)}\Big[\sup_{0\leq t\leq \frac{T}{\lambda} }  \Big|\sum_{n=1}^{\lfloor \frac{1}{2}\tilde{N}_{t} \rfloor} \Big( \int_{R_{2n}}^{R_{2n+1}}dr g_{\lambda}'(S_{r})-b_{\lambda} \Big)  \Big| \Big] .
\end{align*}
The first term is $\mathit{O}(\lambda^{\frac{1}{2}})$ since $|b_{\lambda} |=\mathit{O}(\lambda^{\frac{1}{2}})$ and  $\lambda^{\frac{1}{2}}\tilde{\mathbb{E}}^{(\lambda)}\big[ \tilde{N}_{\frac{T}{\lambda}}  \big] $ is bounded for $\lambda< 1$ by the remark~(\ref{SureFire}).  Moreover, the processes 
\begin{align}\label{Libby}
\sum_{n=1}^{\lfloor \frac{1}{2}\tilde{N}_{t}+\frac{1}{2} \rfloor }  \Big(\int_{R_{2n-1}}^{R_{2n}}dr g_{\lambda}'(S_{r})-b_{\lambda}  \Big) \hspace{1cm} \text{and}\hspace{1cm} \sum_{n=1}^{\lfloor \frac{1}{2}\tilde{N}_{t} \rfloor} \Big( \int_{R_{2n}}^{R_{2n+1}}dr g_{\lambda}'(S_{r})-b_{\lambda} \Big)  
\end{align}
are martingales with respect to the filtration $ \tilde{\mathcal{F}}'_{t}$ since non-sequential life cycles are independent, which is why we split the original sum into even and odd terms.   We can apply the standard arguments to bound the sums in~(\ref{Libby}), for instance:
\begin{align*}
\lambda^{\frac{1}{2}}\tilde{\mathbb{E}}^{(\lambda)}\Big[\sup_{0\leq t\leq \frac{T}{\lambda} }  \Big|\sum_{n=1}^{\lfloor \frac{1}{2}\tilde{N}_{t}+\frac{1}{2} \rfloor }  \int_{R_{2n-1}}^{R_{2n}}dr g_{\lambda}'(S_{r})-b_{\lambda}   \Big| \Big] &\leq 2\lambda^{\frac{1}{2}}\tilde{\mathbb{E}}^{(\lambda)}\Big[ \sum_{n=1}^{\lfloor \frac{1}{2}\tilde{N}_{t}+\frac{1}{2} \rfloor} \Big( \int_{R_{2n-1}}^{R_{2n}}dr g_{\lambda}'(S_{r})  -b_{\lambda} \Big)^{2} \Big]^{\frac{1}{2} }  \\ &\leq \lambda^{\frac{1}{2}} C^{\frac{1}{2}} \tilde{\mathbb{E}}^{(\lambda)}\big[ \tilde{N}_{\frac{T}{\lambda} }\big]^{\frac{1}{2}}=\mathit{O}(\lambda^{\frac{1}{4}}).
\end{align*}
The first inequality is Jensen's  with the square function followed by Doob's maximal inequality, and the second follows analogously to previous discussion.

\end{proof}

\section*{Acknowledgments}
The author thanks Lo\"ic Dubois for many discussions on the topic of this article.  This work is supported by the European Research Council grant No. 227772 and NSF grant DMS-08446325.

\begin{appendix}

\section{The limiting diffusion process}\label{SecIdealProc}

\subsection{Local time at the origin for an Ornstein-Uhlenbeck process}

Let  $\frak{p}$ be the Ornstein-Uhlenbeck process satisfying the Langevin equation~(\ref{TheLimit}) and $\frak{l}$ be the corresponding local time at zero.  For a discussion of local time for continuous semimartingales, we refer to~\cite[Sect. 3.7]{Karat}, and for a list of many formulae related to the local time of an Ornstein-Uhlenbeck process we refer to~\cite{Borodin}.  As mentioned before the local time is formally $\frak{l}_{t}=\int_{0}^{t}dr \delta_{0}(\frak{p}_{r})$,  and through a formal application of the Ito formula, it satisfies
$$ \frak{l}_{t}=|\frak{p}_{t}|-|\frak{p}_{0}| -\int_{0}^{t}dr\textup{sgn}(\frak{p}_{r})d\mathbf{B}_{r}+\frac{1}{2}\int_{0}^{t}dr|\frak{p}_{r}|dr,  $$
where $\textup{sgn}:\R\rightarrow \{\pm 1\}$ is the sign function.  The above is one of the Tanaka-Meyer  formulas.  The process $\frak{l}$ is a continuous increasing process  satisfying $\frak{l}_{t}\rightarrow \infty$ as $t\rightarrow \infty$ since $\frak{p}$ is a positive-recurrent process.  The process inverse $\frak{s}_{r}=\inf \{ t\in \R^{+}\,\big|\, \frak{l}_{t}\geq r   \}    $ has independent and stationary increments and is thus an increasing Levy processes.  The flats of $\frak{l}$ correspond to excursions from the origin for $\frak{p}$ and jumps for $\frak{s}$.  

We can give a closed expression for the Laplace transform $\mathbb{E}\big[e^{ -\gamma \frak{s}_{t}}\big]$.  The Laplace transform has the form  
\begin{align}\label{CornStarch}
\mathbb{E}\big[e^{ -\gamma \frak{s}_{t}}\big]=e^{-\frac{t}{G_{\gamma}(0,0)}}  .   
\end{align}
where $G_{\gamma}$ is the Green function for the Ornstein-Uhlenbeck process.  The densities $\textup{Q}_{t}:\R\rightarrow \R^{+}$ for $\frak{p}_{t}$ satisfy the forward equation
$$\hspace{3cm} \frac{d}{dt}\textup{Q}_{t}(p)=\frac{1}{2} \textup{Q}_{t}(p) +\frac{1}{2}p\frac{\partial}{\partial p}\textup{Q}_{t}(p)+\frac{1}{2}\frac{\partial^{2}}{\partial^{2} p}\textup{Q}_{t}(p). $$
When $Q_{0}(p)=\delta_{0}(p)$, then $\textup{Q}_{t}(p)$ has the explicit form
\begin{align}\label{RightHereChap}
 \hspace{1.5cm} \textup{Q}_{t}(p) =\frac{  e^{-\frac{ p^{2}   }{ 2\omega_{t} }  } }{  (2\pi \omega_{t})^{\frac{1}{2}}},\hspace{1cm} \omega_{t}=1-e^{-\frac{1}{2}t}. 
 \end{align}
Notice that there is convergence to a variance one Gaussian in the limit that $t\rightarrow \infty$.  The form~(\ref{RightHereChap}) allows the Green's  function value $G_{\gamma}(0,0)$ to be computed as the following:
\begin{align*}
G_{\gamma}(0,0)&=\int_{0}^{\infty}dt e^{-\gamma t}\textup{Q}_{t}(0)=(2\pi)^{-\frac{1}{2}}\int_{0}^{\infty}dt \frac{e^{-\gamma t}}{\big(1-e^{-\frac{1}{2}t}\big)^{\frac{1}{2}}}  = (\frac{2}{\pi})^{\frac{1}{2}}\int_{0}^{1}du\, u^{2\gamma-1}\big(1-u\big)^{-\frac{1}{2}}\\ &=(\frac{2}{\pi})^{\frac{1}{2}}\textup{B}\big(2\gamma,\frac{1}{2}\big)=2^{\frac{1}{2}}\frac{ \Gamma(2\gamma)   }{\Gamma(2\gamma+\frac{1}{2}) },
\end{align*}
where $\textup{B}$ and $\Gamma$ are respectively the $\beta$-function and $\gamma$-functions, and we have made the substitution $u=e^{-\frac{1}{2}t}$, $-2u^{-1}du=dt$ for the third equality. 
Plugging our results into~(\ref{CornStarch}) the moment-generating function of $\frak{s}_{t}$ is 
$$\mathbb{E}\big[e^{ -\gamma \frak{s}_{t}}\big]= e^{-t 2^{-\frac{1}{2}}\frac{ \Gamma(2\gamma+\frac{1}{2})   }{\Gamma(2\gamma) }}. $$

The Levy rate density $R:\R^{+}\rightarrow \R^{+} $ for $\frak{s}_{t}$ satisfies that
$$   \int_{0}^{\infty}d\tau \big(1-e^{-\gamma \tau}    \big)R(\tau)= 2^{-\frac{1}{2}}\frac{ \Gamma(2\gamma+\frac{1}{2})   }{\Gamma(2\gamma) } . $$
The rates $R(\tau)=4^{-1}(2\pi)^{-\frac{1}{2}}e^{-\frac{1}{4}\tau} \big(1-e^{-\frac{1}{2}\tau} \big)^{-\frac{3}{2}}  $ can be deduced by   similar operations as above in reverse order since 
\begin{align*}
2^{-\frac{1}{2}}\frac{ \Gamma(2\gamma+\frac{1}{2})   }{\Gamma(2\gamma) }&= \frac{2\gamma  }{(2\pi)^{\frac{1}{2}} }\textup{B}\big(2\gamma+\frac{1}{2},\frac{1}{2}\big) =\frac{\gamma  }{(2\pi)^{\frac{1}{2}} }\int_{0}^{\infty}d\tau e^{-\gamma \tau} \frac{ e^{-\frac{1}{4}\tau}  }{\big(1-e^{-\frac{1}{2}\tau} \big)^{\frac{1}{2}}   } \\ &=\frac{1 }{4(2\pi)^{\frac{1}{2}} }\int_{0}^{\infty}d\tau \big(1-e^{-\gamma \tau}\big) \frac{ e^{-\frac{1}{4}\tau}  }{\big(1-e^{-\frac{1}{2}\tau} \big)^{\frac{3}{2}}   } .
\end{align*}

\subsection{A diffusion time-changed by $\frak{l}_{t}$}

Now we consider the process $\mathbf{B}_{\frak{l}}$ where $\mathbf{B}$ is a Brownian motion with diffusion rate $\kappa$ which is independent of the process $\frak{l}$ discussed in the last section.  Although $\mathbf{B}_{\frak{l}}$ is non-Markovian, the triple $(\mathbf{B}_{\frak{l}},\,\tau  ,\,\eta  )$ is Markovian, where $\tau_{t}:= \frak{s}_{\ell_{t}}-\frak{s}_{\ell_{t}-}$ is the total duration of the current excursion (which requires some information from the future), and $\eta_{t}:=t- \frak{s}_{\ell_{t}-}$ is the amount of time that has passed since the beginning of the excursion.  

We can give a closed form for the joint density $\rho_{t}(x,\tau,\eta)$ for the triple $(\mathbf{B}_{\frak{l}_{t}},\,\tau_{t}  ,\,\eta_{t}  )$ assuming that  $\mathbf{B}_{0}$ has density $\rho(x)$ and $\eta_{0}=\tau_{0}=0$.  Let $\Psi_{r}(t)$ be the probability density at the value $t\in \R^{+}$ for the Levy process $\frak{s}$ at time $r$.  The joint density $\rho_{t}(x,\tau,\eta)$ for the triple $(\mathbf{B}_{\frak{l}_{t}},\,\tau_{t}  ,\,\eta_{t}  )$ has the closed form  
$$\rho_{t}(x,\tau,\eta)=\chi( \eta\leq \tau\wedge t  )\,R(\tau)\int_{0}^{\infty}dr \Psi_{r}(t-\eta) \, (g_{r}*\rho)(x), \quad \quad g_{r}(x)= \frac{e^{-\frac{x^{2}}{2r\kappa }}}{(2\pi r\kappa)^{\frac{1}{2}}}, $$
where $R: \R^{+}\rightarrow \R^{+}$ is the rate function for the Levy process $\frak{s}$.   By integrating out the $\tau, \eta$ variables, we obtain that the marginal density $\rho_{t}(x)$ satisfies the Volterra-type integro-differential equation of the form
\begin{align}\label{Shizz}
\rho_{t}(x) =\rho_{0}(x) +\frac{\kappa}{2}\int_{0}^{t}dr\frac{(2\pi)^{-\frac{1}{2}}}{\big( 1-e^{-\frac{1}{2}(t-r)}\big)^{\frac{1}{2}} }(\Delta \rho_{r})(x).  
\end{align}
In the above, we have used that $ \Psi_{s}*\Psi_{t}=\Psi_{s+t} $ and the explicit computation 
  $$ \int_{0}^{\infty}dr\Psi_{r}(t)= Q_{t}(0)= \frac{ (2\pi)^{-\frac{1}{2}}}{ \big( 1-e^{-\frac{1}{2}t}\big)^{\frac{1}{2}} }.   $$

The above is analogous to the master equation for a Brownian motion time-changed by a Mittag-Leffler process.  The  Mittag-Leffler process $\frak{m}^{(\alpha)}$  of index $0<\alpha<1$ is distributed as  the process inverse of the one-sided stable law of index $\alpha$.  The $\alpha=\frac{1}{2}$ case has the same law as the local time of a standard Brownian motion.  If $\mathbf{B}$ is a standard Brownian motion, then the densities for   $\sqrt{\kappa}\mathbf{B}_{\frak{m}_{ t}^{(\alpha)}}$  satisfy the equation
\begin{align*}
\rho_{t}(x) =\rho_{0}(x) +\frac{\kappa}{2\Gamma(\alpha)}\int_{0}^{t}dr(t-r)^{\alpha-1}(\Delta \rho_{r})(x),
\end{align*}
which is equivalent to the fractional diffusion equation
$$\partial_{t}^{\alpha}\rho_{ t}= \kappa \Delta_{q} \rho_{ t},    $$
where the fractional derivative $\partial_{t}^{\alpha}$ acts as $(\partial_{t}^{\alpha}f)(t)=\frac{1}{\Gamma(1-\alpha)}\frac{d}{dt}\int_{0}^{t}dr (t-r)^{-\alpha}f(r)$.  Processes satisfying these equations arise in the theory of continuous time random walks~\cite{Montroll,Meershaert} and the limit theory for martingales whose quadratic variations are driven by  additive functionals of null-recurrent Markov processes~\cite{TouatiUnpub,ChenII,Hopfner}.   The process $\mathbf{B}_{\frak{m}^{(\alpha)}}$ has the scale invariance in law
$$
\mathbf{B}_{\frak{m}_{t}^{(\alpha)} }\stackrel{\mathcal{L}}{ =  }\epsilon^{-\frac{\alpha}{2}}\mathbf{B}_{\frak{m}_{ \epsilon t}^{(\alpha)} }.
$$

\subsection{Long-term behavior}

Now we can look into the diffusive behavior for  $\mathbf{B}_{\frak{l}_{t}}$  in the limit of large times $t$.  Since the process is already a diffusion, this is just a question of the convergence in probability for  the normalized quadratic variation $t^{-1}\frak{l}_{st}$ for $s\in \R^{+}$ as $t\rightarrow \infty$.  However, we actually have a strong limit since      
$$    \lim_{t\rightarrow \infty}\frac{\frak{l}_{st}}{t}=s\lim_{r\rightarrow \infty}\frac{  r  }{ \frak{s}_{r}  }=s\Big( \int_{0}^{\infty}d\tau\, \tau R(\tau) \Big)^{-1}=s(2\pi)^{-\frac{1}{2}}. $$
The first equality holds since $\frak{l}$ and $\frak{s}$ are process inverses of one another and tend to infinity almost surely.  The second equality is the strong law of large numbers for the Levy process $\frak{s}_{r}$.  The computation for the third equality is based on the representation of the Laplace transform of $\frak{s}_{t}$ from the last section.  The above implies the convergence in law as $\lambda \rightarrow 0$ given by 
$$   t^{-\frac{1}{2}}\mathbf{B}_{\frak{l}_{st}} \stackrel{\frak{L}}{\Longrightarrow} (2\pi)^{-\frac{1}{2}}\mathbf{B}_{s}^{\prime},  $$
where $\mathbf{B}^{\prime}$ is a copy of standard Brownian motion.

\end{appendix}

\end{document}